\documentclass[11pt,letterpaper]{article}

\usepackage{amssymb,amsmath,latexsym,amscd,amsfonts}
\usepackage{enumerate}
\usepackage{graphicx}
\usepackage{color}
\usepackage{fullpage}
\usepackage[colorlinks,allcolors=black]{hyperref}
\usepackage{hyperref}
\usepackage{amsthm,bbm,mathtools}
\usepackage[T1]{fontenc} 
\usepackage[utf8]{inputenc}
\usepackage{enumitem}
\usepackage[font={small}]{caption}
\usepackage{algorithmicx}
\usepackage{algpseudocode}
\usepackage{amscd}
\usepackage[dvipsnames]{xcolor}
\usepackage[american]{babel}
\usepackage{comment}
\usepackage[capitalise]{cleveref}
\usepackage{thm-restate}
\usepackage{mathrsfs}
\definecolorset{HTML}{}{}{%
slate50,f8fafc;%
slate100,f1f5f9;%
slate200,e2e8f0;%
slate300,cbd5e1;%
slate400,94a3b8;%
slate500,64748b;%
slate600,475569;%
slate700,334155;%
slate800,1e293b;%
slate900,0f172a;%
slate950,020617;%
gray50,f9fafb;%
gray100,f3f4f6;%
gray200,e5e7eb;%
gray300,d1d5db;%
gray400,9ca3af;%
gray500,6b7280;%
gray600,4b5563;%
gray700,374151;%
gray800,1f2937;%
gray900,111827;%
gray950,030712;%
zinc50,fafafa;%
zinc100,f4f4f5;%
zinc200,e4e4e7;%
zinc300,d4d4d8;%
zinc400,a1a1aa;%
zinc500,71717a;%
zinc600,52525b;%
zinc700,3f3f46;%
zinc800,27272a;%
zinc900,18181b;%
zinc950,09090b;%
neutral50,fafafa;%
neutral100,f5f5f5;%
neutral200,e5e5e5;%
neutral300,d4d4d4;%
neutral400,a3a3a3;%
neutral500,737373;%
neutral600,525252;%
neutral700,404040;%
neutral800,262626;%
neutral900,171717;%
neutral950,0a0a0a;%
stone50,fafaf9;%
stone100,f5f5f4;%
stone200,e7e5e4;%
stone300,d6d3d1;%
stone400,a8a29e;%
stone500,78716c;%
stone600,57534e;%
stone700,44403c;%
stone800,292524;%
stone900,1c1917;%
stone950,0c0a09;%
red50,fef2f2;%
red100,fee2e2;%
red200,fecaca;%
red300,fca5a5;%
red400,f87171;%
red500,ef4444;%
red600,dc2626;%
red700,b91c1c;%
red800,991b1b;%
red900,7f1d1d;%
red950,450a0a;%
orange50,fff7ed;%
orange100,ffedd5;%
orange200,fed7aa;%
orange300,fdba74;%
orange400,fb923c;%
orange500,f97316;%
orange600,ea580c;%
orange700,c2410c;%
orange800,9a3412;%
orange900,7c2d12;%
orange950,431407;%
amber50,fffbeb;%
amber100,fef3c7;%
amber200,fde68a;%
amber300,fcd34d;%
amber400,fbbf24;%
amber500,f59e0b;%
amber600,d97706;%
amber700,b45309;%
amber800,92400e;%
amber900,78350f;%
amber950,451a03;%
yellow50,fefce8;%
yellow100,fef9c3;%
yellow200,fef08a;%
yellow300,fde047;%
yellow400,facc15;%
yellow500,eab308;%
yellow600,ca8a04;%
yellow700,a16207;%
yellow800,854d0e;%
yellow900,713f12;%
yellow950,422006;%
lime50,f7fee7;%
lime100,ecfccb;%
lime200,d9f99d;%
lime300,bef264;%
lime400,a3e635;%
lime500,84cc16;%
lime600,65a30d;%
lime700,4d7c0f;%
lime800,3f6212;%
lime900,365314;%
lime950,1a2e05;%
green50,f0fdf4;%
green100,dcfce7;%
green200,bbf7d0;%
green300,86efac;%
green400,4ade80;%
green500,22c55e;%
green600,16a34a;%
green700,15803d;%
green800,166534;%
green900,14532d;%
green950,052e16;%
emerald50,ecfdf5;%
emerald100,d1fae5;%
emerald200,a7f3d0;%
emerald300,6ee7b7;%
emerald400,34d399;%
emerald500,10b981;%
emerald600,059669;%
emerald700,047857;%
emerald800,065f46;%
emerald900,064e3b;%
emerald950,022c22;%
teal50,f0fdfa;%
teal100,ccfbf1;%
teal200,99f6e4;%
teal300,5eead4;%
teal400,2dd4bf;%
teal500,14b8a6;%
teal600,0d9488;%
teal700,0f766e;%
teal800,115e59;%
teal900,134e4a;%
teal950,042f2e;%
cyan50,ecfeff;%
cyan100,cffafe;%
cyan200,a5f3fc;%
cyan300,67e8f9;%
cyan400,22d3ee;%
cyan500,06b6d4;%
cyan600,0891b2;%
cyan700,0e7490;%
cyan800,155e75;%
cyan900,164e63;%
cyan950,083344;%
sky50,f0f9ff;%
sky100,e0f2fe;%
sky200,bae6fd;%
sky300,7dd3fc;%
sky400,38bdf8;%
sky500,0ea5e9;%
sky600,0284c7;%
sky700,0369a1;%
sky800,075985;%
sky900,0c4a6e;%
sky950,082f49;%
blue50,eff6ff;%
blue100,dbeafe;%
blue200,bfdbfe;%
blue300,93c5fd;%
blue400,60a5fa;%
blue500,3b82f6;%
blue600,2563eb;%
blue700,1d4ed8;%
blue800,1e40af;%
blue900,1e3a8a;%
blue950,172554;%
indigo50,eef2ff;%
indigo100,e0e7ff;%
indigo200,c7d2fe;%
indigo300,a5b4fc;%
indigo400,818cf8;%
indigo500,6366f1;%
indigo600,4f46e5;%
indigo700,4338ca;%
indigo800,3730a3;%
indigo900,312e81;%
indigo950,1e1b4b;%
violet50,f5f3ff;%
violet100,ede9fe;%
violet200,ddd6fe;%
violet300,c4b5fd;%
violet400,a78bfa;%
violet500,8b5cf6;%
violet600,7c3aed;%
violet700,6d28d9;%
violet800,5b21b6;%
violet900,4c1d95;%
violet950,2e1065;%
purple50,faf5ff;%
purple100,f3e8ff;%
purple200,e9d5ff;%
purple300,d8b4fe;%
purple400,c084fc;%
purple500,a855f7;%
purple600,9333ea;%
purple700,7e22ce;%
purple800,6b21a8;%
purple900,581c87;%
purple950,3b0764;%
fuchsia50,fdf4ff;%
fuchsia100,fae8ff;%
fuchsia200,f5d0fe;%
fuchsia300,f0abfc;%
fuchsia400,e879f9;%
fuchsia500,d946ef;%
fuchsia600,c026d3;%
fuchsia700,a21caf;%
fuchsia800,86198f;%
fuchsia900,701a75;%
fuchsia950,4a044e;%
pink50,fdf2f8;%
pink100,fce7f3;%
pink200,fbcfe8;%
pink300,f9a8d4;%
pink400,f472b6;%
pink500,ec4899;%
pink600,db2777;%
pink700,be185d;%
pink800,9d174d;%
pink900,831843;%
pink950,500724;%
rose50,fff1f2;%
rose100,ffe4e6;%
rose200,fecdd3;%
rose300,fda4af;%
rose400,fb7185;%
rose500,f43f5e;%
rose600,e11d48;%
rose700,be123c;%
rose800,9f1239;%
rose900,881337;%
rose950,4c0519}

\hypersetup{
  colorlinks   = true, 
  urlcolor     = blue900, 
  linkcolor    = green900, 
  citecolor    = blue900 
}

\usepackage{tikz}
\usetikzlibrary{shapes,arrows,calc,positioning,fit}
\tikzstyle{vertex} = [fill,shape=circle,node distance=80pt]
\tikzstyle{edge} = [fill,opacity=.5,fill opacity=.5,line cap=round, line join=round, line width=50pt]
\tikzstyle{elabel} =  [fill,shape=circle,node distance=30pt]

\def\ignore #1 {}

\newcommand{\norm}[1]{\left\|\,#1\,\right\|}       
\newcommand{\enorm}[1]{\norm{#1}_{\mathrm{2}}}      

\newcommand{\set}[1]{{\left\{#1\right\}}}    

\newcommand{\abs}[1]{\left\lvert #1 \right\rvert}

\newcommand{\reals}{{\mathbb R}}

\DeclareMathOperator{\rank}{rank}
\DeclareMathOperator{\Span}{Span}
\DeclareMathOperator{\image}{Im}
\DeclareMathOperator{\poly}{poly}
\DeclareMathOperator{\polylog}{polylog}

\newcommand{\CC}{\mathbb{C}}

\newcommand{\RR}{\reals}

\newcommand{\dij}{d_{i,j}}
\newcommand{\cj}{c_j}

\newcommand{\whV}{\widehat{V}}
\newcommand{\whG}{\widehat{G}}
\newcommand{\whE}{\widehat{E}}

\newcommand{\wte}{\widetilde{e}}
\newcommand{\wtv}{\widetilde{v}}
\newcommand{\whu}{\widehat{u}}
\newcommand{\whv}{\widehat{v}}

\def\ket#1{ | #1 \rangle}
\def\bra#1{{\langle #1 | }}
\newcommand{\ketbra}[2]{\ket{#1}\!\bra{#2}}        

\newcommand{\braket}[2]{\mbox{$\langle #1  | #2 \rangle$}}

\newcommand{\lin}{\mathcal{L}}

\newcommand{\TFNP}{\textup{TFNP}}
\newcommand{\PP}{\mathbb P}

\newcommand{\PPA}{\textup{PPA}}
\newcommand{\PPAD}{\textup{PPAD}}
\newcommand{\EOPL}{\textup{EOPL}}
\newcommand{\SOPL}{\textup{SOPL}}
\newcommand{\PPADS}{\textup{PPADS}}
\newcommand{\CLS}{\textup{CLS}}
\newcommand{\PLS}{\textup{PLS}}
\newcommand{\NP}{\textup{NP}}
\newcommand{\QMA}{\textup{QMA}}
\newcommand{\QMAo}{\QMA_1}
\newcommand{\Bez}{B\'{e}zout}
\newcommand{\PS}{PRODSAT}

\newtheorem{theorem}{Theorem}
\newtheorem{proposition}[theorem]{Proposition}
\newtheorem{lemma}[theorem]{Lemma}
\newtheorem{corollary}[theorem]{Corollary}
\newtheorem{observation}[theorem]{Observation}

\theoremstyle{definition}
\newtheorem{definition}[theorem]{Definition}
\newtheorem{rem}[theorem]{Remark}
\newtheorem{example}[theorem]{Example}
\newtheorem{fact}[theorem]{Fact}

\def\cH{{\mathcal{H}}}

\newcommand\ZZ{\mathbb{Z}}
\newcommand\QQ{\mathbb{Q}}

\newcommand{\MHS}{\textup{MHS}}
\newcommand{\MHSe}{\MHS(\epsilon)}
\newcommand{\MHSo}[1]{\MHS(#1)}
\newcommand{\SFTA}{\textup{SFTA}}
\newcommand{\bez}{d_{\textit{Béz}}}

\newcommand{\bmat}[1]{\begin{bmatrix}#1\end{bmatrix}}

\title{An unholy trinity: TFNP, polynomial systems, and the quantum satisfiability problem}
\author{Marco Aldi\footnote{Department of Mathematics and Applied Mathematics, Virginia Commonwealth University, USA. Email: maldi2@vcu.edu.} \and Sevag Gharibian\footnote{Department of Computer Science and Institute for Photonic Quantum Systems (PhoQS), Paderborn University, Germany. Email: \{sevag.gharibian, dorian.rudolph\}@upb.de.} \and Dorian Rudolph\footnotemark[2]}
\date{}

\begin{document}

\maketitle

\begin{abstract}
  The theory of Total Function NP (TFNP) and its subclasses says that, even if one is promised an efficiently verifiable proof \emph{exists} for a problem, \emph{finding} this proof can be intractable. Despite the success of the theory at showing intractability of problems such as computing Brouwer fixed points and Nash equilibria, subclasses of TFNP remain arguably few and far between. In this work, we define two new subclasses of TFNP borne of the study of complex polynomial systems: Multi-homogeneous Systems (MHS) and Sparse Fundamental Theorem of Algebra (SFTA). The first of these is based on \Bez's theorem from algebraic geometry, marking the first TFNP subclass based on an algebraic geometric principle. At the heart of our study is the computational problem known as Quantum SAT (QSAT) with a System of Distinct Representatives (SDR), first studied by [Laumann, L\"{a}uchli, Moessner, Scardicchio, and Sondhi 2010]. Among other results, we show that QSAT with SDR is MHS-complete, thus giving not only the first link between quantum complexity theory and TFNP, but also the first TFNP problem whose classical variant (SAT with SDR) is easy but whose quantum variant is hard. We also show how to embed the roots of a sparse, high-degree, univariate polynomial into QSAT with SDR, obtaining that SFTA is contained in a zero-error version of MHS. We conjecture this construction also works in the low-error setting, which would imply $\SFTA\subseteq \MHS$.
\end{abstract}

{\hypersetup{hidelinks}\tableofcontents}

\section{Introduction}~\label{scn:intro}
The genesis of this work consists of three elements: TFNP, \Bez's theorem, and the quantum satisfiability problem. As such, we begin by giving background on these three. The Fundamental Theorem of Algebra's role will then be introduced when stating our results in \Cref{sscn:results}.\\

\noindent \emph{The first element: TFNP.} The late 1980's and early 1990's witnessed the emergence~\cite{johnsonHowEasyLocal1988,megiddoTotalFunctionsExistence1991,Pap94} of a complexity theoretic framework which answered the question: \emph{How can one characterize the complexity of problems for which an efficiently verifiable solution is guaranteed to exist, but finding this solution appears difficult?} Specifically, Total Function NP (\TFNP)~\cite{megiddoTotalFunctionsExistence1991} was defined as the class of NP search problems with a guaranteed witness --- in other words, the \emph{decision} versions of these problems are trivial, so the  challenge is ``just'' to find the witness. This definition encompasses numerous old-school mathematical principles --- Brouwer's fixed point theorem, for example, says that any continuous function $f$ from a non-empty compact convex to itself has a fixed point (i.e. an $x$ such that $f(x)=x$), but \emph{finding} said fixed point appears difficult. Likewise, Nash's theorem states that any non-cooperative game with a finite number of players and a finite number of actions has a Nash equilibrium, but efficiently finding a Nash equilibrium remains elusive. 

Formally, to show that a given search problem $\Pi\in \TFNP$ is intractable, one proves hardness of $\Pi$ for one of the known subclasses of TFNP, each of which is itself based on an old-school mathematical principle. The five most prominent   subclasses are~\cite{johnsonHowEasyLocal1988,Pap94}: 
\begin{itemize}
  \item Pigeonhole Principle (PPP) corresponds to NP search problems guaranteed to have a solution via application of the \emph{pigeonhole principle}. 
  \item Polynomial Parity Argument (\PPA) leverages the \emph{handshaking lemma}: In any finite undirected graph, the number of odd-degree vertices is even. 
  \item Polynomial Parity Argument on Directed Graphs (\PPAD) uses the fact that any directed graph with an unbalanced node (meaning with in-degree $\neq$ out-degree) must have another unbalanced node.
  \item Polynomial Parity Argument on Directed Graphs with a Sink (\PPADS) is identical to \PPAD, except one requires finding an oppositely balanced node.
  
  \item Polynomial Local Search (PLS) uses the fact that every directed acyclic graph has a sink.
\end{itemize}
Although \emph{a priori}, these subclasses appear to have nothing to do with (say) finding fixed points, appearances can be deceiving: Finding a Brouwer fixed point~\cite{Pap94} and a Nash equilibrium~\cite{daskalakisComplexityComputingNash2006,chenSettlingComplexityComputing2009} are both PPAD-complete. Even the ubiquitous gradient descent algorithm has not escaped the reach of this framework --- its complexity was shown $\PPAD\cap\PLS$-complete in a recent breakthrough work~\cite{fearnleyComplexityGradientDescent2022}. 

Unfortunately, beyond the ``Big Five'' subclasses above, defining genuinely new subclasses of \TFNP\ has proven challenging. In fact, some of the handful of other known subclasses of TFNP have surprisingly recently turned out to equal \emph{intersections} of the ``Big Five'': $\CLS=\PPAD\cap\PLS$~\cite{fearnleyComplexityGradientDescent2022}, $\EOPL=\PLS\cap\PPAD$ and $\SOPL=\PLS\cap \PPADS$ \cite{goosFurtherCollapsesTFNP2022} (see also \cite{liIntersectionClassesTFNP2024}). \\
\vspace{-1mm}

\noindent\emph{The second element: \Bez's theorem.} In this work, we first define a new subclass of \TFNP\ based on computing solutions to systems of multivariate polynomial equations, given a mathematical principle guaranteeing the existence of a solution. There is only one line of \TFNP\ work we are aware of in a related direction, which we mention first to set context. Specifically, for \emph{finite} fields, Papadimitriou~\cite{Pap94} defined the problem CHEVALLEY by invoking the Chevalley-Warning theorem, which states: Given is a system of polynomials $\set{f_i}_{i=1}^r$ over $\mathbb{F}_p[X_1,\ldots, X_n]$ for finite field $\mathbb{F}_p$, where polynomial $f_i$ has degree $d_i$. If $n> \sum_{i=1}^r d_j$, then the number of common solutions to the system is divisible by the characteristic $p$ of $\mathbb{F}_p$. CHEVALLEY then asks: Given such a polynomial system and one solution, find a second solution. Although CHEVALLEY is known to be in \PPA~\cite{Pap94}, it is not expected to be \PPA-complete; however, two variants of CHEVALLEY have been shown \PPA-complete~\cite{belovsPolynomialParityArgument2017,goosComplexityModuloqArguments2020}. 

In this work, we instead consider polynomial systems over \emph{complex} numbers. This necessitates a move from the domain of number theory to, for the first time in the study of \TFNP, \emph{algebraic geometry}. The old-school algebraic geometric principle we invoke is \Bez's theorem from $1779$, nowadays stated as follows: Over an algebraically closed field, any system of $n$ homogeneous polynomials in $n+1$ variables always has either an infinite number of solutions, or exactly $d_1\cdots d_n$ solutions, for $d_i$ the degree of the $i$th polynomial. For our purposes, we actually require a more recent \emph{multi}-homogenous extension due to Shafarevich~\cite{Shafarevich1974}, which gives a similar statement for the more general setting of systems of \emph{multi}-homogeneous polynomials (\Cref{def:multipoly}), which we now informally define. 

Recall that a homogeneous polynomial is one whose non-zero monomials all have the same degree. A \emph{multi}-homogeneous polynomial $p\in\CC[x_1,\ldots, x_n]$ 
generalizes this definition: One first partitions the variables $\set{x_i}$ into sets $S_i$ as desired, and then requires that for each $S_i$, if we treat only the elements of $S_i$ as variables, the resulting polynomial is homogeneous. For example, for variable sets $S_1=\set{x_1,x_2}$ and $S_2=\set{y_1,y_2,y_3}$, $x_1y_1y_2 + x_2y_2y_3$ is multihomogeneous, whereas the homogeneous polynomial $x_1+y_1$ is not. (Nevertheless, any homogeneous polynomial is trivially multihomogeneous relative to the partition with one set $S$ containing all variables.) 

The multi-homogeneous \Bez\ theorem (\Cref{thm:bezout}) now first defines, corresponding to the product of degrees $d_1\cdots d_n$ from the original \Bez\ theorem, a more general quantity known as the \emph{\Bez\ number} $\bez$ (\Cref{def:beznumber}). Then, it states that for any multi-homogeneous system of $n$ equations $\set{p_j}_{j=1}^{n}\subseteq\mathbb{C}[x_1,\ldots, x_{n+t}]$, where the variables are partitioned into $t$ sets $S_i$, if $\bez>0$, then the system has a solution. Note this generalizes \Bez's theorem when all variables are placed into one set, $S$, so that $t=1$. Roughly, our first new subclass of TFNP, denoted $\MHS$ (defined shortly in \Cref{def:MHSinformal}), is the set of TFNP problems reducible to a multi-homogeneous system satisfying the multi-homogeneous \Bez\ theorem. 
Importantly, it can be efficiently checked if $\bez>0$, which suffices for our purposes (\Cref{rem:computebez}).\\

\vspace{-1mm}

\noindent\emph{The third element: The quantum satisfiability problem.} With two members of our trinity in hand, TFNP and \Bez's theorem, we introduce the ``unholy'' member of the fellowship: The {quantum} satisfiability (QSAT) problem. We say ``unholy'' because of the unexpected nature of this trio --- not only is this the first time quantum complexity and \TFNP\ have been formally linked, but the {classical} Boolean satisfiability analogue of the problem we consider is a textbook example of an \emph{easy} search problem. To elaborate on the latter, consider $3$-SAT when the constraint system has a System of Distinct Representatives\footnote{Given subsets $S_1,\ldots, S_m\subseteq [n]$, an SDR is a set of distinct elements $r_1,\ldots,r_m$ such that $r_i\in S_i$ for all $i\in[m]$. In the context of $3$-SAT, each $S_i$ is the set of variables in clause $c_i$, and elements $1$ through $n$ correspond to the set of all variables.} (SDR). Then, for each clause $c_i=(x_i\vee y_i \vee z_i)$ of formula $\phi$, one can ``match'' one of the variables in $\set{x_i,y_i,z_i}$ \emph{uniquely} to $c_i$. Since no variable is matched twice in this process, setting each matched literal to true yields a satisfying assignment for $\phi$. As an SDR can be found efficiently (e.g. via reduction to network flow~\cite{ford_fulkerson_1956}), the search version of $3$-SAT with SDR is poly-time solvable.

The \emph{quantum} analogue of this story has played out differently. Here, the Quantum Satisfiability problem ($k$-QSAT) on $n$ qubits generalizes $k$-SAT, and is defined as follows: Given a set of projectors $\set{\Pi_S}_S$, each acting non-trivially\footnote{Formally, one sets $\Pi_S\otimes I_{[n]\setminus S}$ to ensure each projector acts on the correct space, $\CC^{2^n}$.} on some subset $S\subseteq[n]$ of qubits, does there exist an $n$-qubit quantum state $\ket{\psi}\in\CC^{2^n}$ simultaneously satisfying all quantum clauses, i.e. $\Pi_S\ket{\psi}=0$ for all $\Pi_S$? 
First, the commonalities: Just as $3$-SAT is NP-complete, $3$-QSAT is $\QMAo$-complete~\cite{gossetQuantum3SATQMA1Complete2013}, where $\QMAo$ is Quantum Merlin Arthur (QMA) with perfect completeness. Likewise, both $2$-SAT~\cite{aspvallLineartimeAlgorithmTesting1979} and $2$-QSAT~\cite{aradLinearTimeAlgorithm2016,beaudrapLinearTimeAlgorithm2016} can be solved in linear time.
Finally, for $k$-QSAT with SDR, Laumann, L\"{a}uchli, Moessner, Scardicchio, and Sondhi~\cite{laumannProductGenericRandom2010} (see also~\cite{laumannc.r.PhaseTransitionsRandom2010,leePRODSATPhaseRandom2024}) showed that, like SAT with SDR, QSAT with SDR on qubits always has a solution. In fact, the solution is an NP witness, being a \emph{tensor product} state (i.e.\ of form $\ket{\psi_1}\otimes\cdots\otimes\ket{\psi_n}\in(\CC^2)^{\otimes n})$. And this is precisely where the stories diverge: Efficiently \emph{finding} this tensor product state/NP witness for QSAT with SDR appears difficult.

There are two works in this direction to be mentioned now. In the positive direction, Aldi, de Beaudrap, Gharibian and Saeedi~\cite{AdBGS21} gave a \emph{parameterized}\footnote{``Parameterized'' as in parameterized complexity, i.e. the runtime of the algorithm scales polynomially in the input size, but exponentially in structural parameters of the constraint hypergraph.} algorithm solving a special class of QSAT with SDR instances efficiently. In the opposite direction,
Goerdt showed~\cite{goerdtMatchedInstancesQuantum2019} QSAT with SDR and the additional restriction that only \emph{real-valued} solutions are allowed is NP-hard.
Thus, it remained unclear in which direction the complexity of QSAT with SDR should fall.

\subsection{Our results}\label{sscn:results}

Briefly, our main contributions (denoted (b) and (c) below) are the definitions and complexity theoretic study of MHS and a second new TFNP subclass based on the Fundamental Theorem of Algebra (\Cref{thm:FTA}), denoted \emph{Sparse Fundamental Theorem of Algebra (SFTA)}. However, the broader story of this paper involves the following sequence of results, which hold for any local qu\emph{d}it dimension $d\geq 2$: (a) QSAT on qudits has a product state solution if and only if the instance has a \emph{weighted} SDR (WSDR). This yields containment in TFNP. (b) QSAT with WSDR on qudits is complete for $\MHS$. (c) To better understand the complexity of $\MHS$, as well as to build on the theme of TFNP subclasses related to complex polynomials, we show containment of $\SFTA$ into a zero-error version of $\MHS$, and as a bonus, use this construction to obtain NP-hardness results for slight variants of QSAT with SDR.
(d) Finally, special cases of QSAT with WSDR on qudits can be efficiently solved. 

We now discuss our results in detail. Throughout, we refer to instances of QSAT by their interaction hypergraph $G=(V,E)$, where vertices correspond to qudits, and hyperedges to clauses. We do not restrict the type, number, or geometry of clauses allowed per qudit. A ``clause'' for us is\footnote{``Stacking'' multiple rank-$1$ projectors to obtain a $d$-dimensional clause is allowed, but for clarity, we count this as $d$ constraints. This is important for the definition of Weighted SDRs.} a rank-$1$ projector. 

\paragraph{a. Existence results via Weighted SDRs.} We begin by introducing the new framework of \emph{Weighted SDRs (WSDR)}, which underlies much of this work. 
Roughly, a WSDR (\Cref{def:weighted}) generalizes an SDR by introducing a \emph{weight} function $w:V\rightarrow \mathbb Z_{\ge 0}$, such that for any vertex $v\in V$ corresponding to a qudit, $v$ can be matched to $w(v)$ clauses. 
Which weight function should one choose?
In this work, when we say a given QSAT instance $G=(V,E)$ on $n$ qudits of local dimensions $d_1,\ldots, d_n$ has a WSDR, we mean with respect to weight function $w(v_i)=d_i-1$ for each $i\in\{1,\ldots,n\}$.
Thus, on $n$-qu\emph{b}it systems, a WSDR is just an SDR.
Note that checking whether $G$ has an WSDR can be done efficiently (\Cref{rem:computeWSDR}).

Our first main result is that WSDRs are tightly connected to when a QSAT instance on qudits has a product state solution.

\begin{restatable}{theorem}{thmWSDRProdsat}\label{thm:wsdr-product-solution}
  Let $\Pi=\{\Pi_i\}$ be an instance of $QSAT$ on $n$ qu$d$its of local dimensions $d_1,\ldots,d_n$, respectively. If $(G,w)$ admits a WSDR, then $\Pi$ admits a satisfying product assignment. If $(G,w)$ does not admit a WSDR and $\Pi$ is generic, then $\Pi$ has no satisfying product assignment.
\end{restatable}
\noindent \Cref{thm:wsdr-product-solution} is the qu\emph{d}it generalization of~\cite{laumannProductGenericRandom2010}, which showed the analogous result for qu\emph{b}it systems with SDR. 
We thus have that for any $d\geq 2$, QSAT with WSDR on qudits is in TFNP. 
Above, ``generic'' (\Cref{def:generic}) means ``for almost all'' instances. For example, $2$-local constraints are generically entangled, whereas constraints in tensor product form are not.
We remark that high-dimensional quantum systems are natural to study: From a computer science perspective, they can lead to surprising transitions in hardness (e.g. 1D Local Hamiltonian problem on qudits for $d\geq 8$ is QMA-complete~\cite{aharonovPowerQuantumSystems2009,hallgrenLocalHamiltonianProblem2013}, whereas 1D Boolean Satisfiability on dits is in P via dynamic programming), and from a physics perspective, many natural systems (e.g. bosonic/fermionic systems) are high-dimensional systems. 

With this said, while interesting in its own right, the primary appeal of \Cref{thm:wsdr-product-solution} for us here is the techniques behind its proof, which will be crucial for our study of MHS.
Specifically, we give two independent proofs of \Cref{thm:wsdr-product-solution}.
The first (\Cref{sscn:chow}) is completely different than~\cite{laumannProductGenericRandom2010}, and introduces the use of the Chow ring (\Cref{sscn:chow}) to obtain a simple proof of just a few lines.
The second (\Cref{sscn:reductionQubits}) gives a poly-time mapping reduction from QSAT on qudits with WSDR to QSAT on qubits with SDR, and then plugs in~\cite{laumannProductGenericRandom2010}. This reduction, in particular, will play a key role in our $\MHS$-hardness result of \Cref{thm:MHScompleteInformal}.\\

\vspace{-1mm}
\noindent \emph{WSDRs beyond QSAT.} As an aside, we demonstrate the power of WSDRs beyond the study of QSAT by using \Cref{thm:wsdr-product-solution} to give a simple proof of a result of Parthasarathy~\cite{Par04}, which says that any completely entangled subspace\footnote{A subspace is \emph{completely entangled} if it does not contain any product states (\Cref{def:CE}).} has dimension at most $\prod_{i=1}^k d_i - \sum_{i=1}^k d_i + k - 1$ (\Cref{cor:CES}).

\paragraph{b. A new subclass of TFNP based on \Bez's theorem.} We now discuss our first main result, for which we define our first subclass of TFNP, which involves \emph{systems} of \emph{low}-degree, \emph{multi}-variate polynomial equations:
\begin{definition}[Multi-homogeneous Systems (MHS) (Informal; see \Cref{def:MHS})]\label{def:MHSinformal}
  $\MHS$ is the set of total NP search problems poly-time reducible to finding an $\epsilon$-approximate solution to a system $F = \{f_1,\dots,f_n\}\subseteq\mathbb{C}[x_1,\ldots,x_{n+t}]$ of multi-homogeneous equations over $\mathbb{C}$ with $\bez>0$, where $t$ is the number of subsets $S_i$ partitioning the variable set. We require the size $s$ of each $S_i$ and degree $d$ per monomial to be constant, and the precision $\epsilon$ must be at least inverse exponential.
\end{definition}
\noindent Comments regarding the constant bounds on the variable set size $s$ and degree $d$: 
(1) This ensures $\MHS\subseteq\TFNP$ even for inverse exponential $\epsilon$, since poly-time Turing machines can efficiently perform basic arithmetic with polynomial bits of precision.  
(2) For \Cref{thm:MHScompleteInformal} below, $d\in O(1)$ is required for our proof and yields constant locality $k$, whereas $s\in O(1)$ yields constant local dimensional qudits. 
(3) Formally, MHS is a union of complexity classes $\MHS_{s,d}$ over all positive natural numbers $s$ and $d$. 
(4) $\MHS$ does not obviously include general homogeneous systems as a special case due to $s\in O(1)$, i.e. one cannot trivially place all variables into one variable group.
Finally, for precision $\epsilon$, we shall utilize $\MHSe$ when we wish to specify a particular precision $\epsilon$. 

We now show that QSAT with SDR is ``$\MHS$-complete''\footnote{We use the term ``$\MHS$-complete'' in the introduction for simplicity, but the formal statement is more subtle (\Cref{thm:MHScomplete}). In the case of $\MHS$-hardness, for example, it says any problem in $\MHS_{s,d}(\epsilon)$ can be reduced to solving $k$-QSAT on qubits with locality $k\geq (s+1)^d$ within precision $\Theta(\epsilon)$, for $s,d,k\in O(1)$. Namely, our reduction does not produce a fixed $k$ which simultaneously yields hardness \emph{for all} $s$ and $d$. This is similar to how for each level $
\Sigma_k^p$ of the Polynomial-Time Hierarchy (PH), Quantified Boolean Satisfiability with $k-1$ alternations (QBF$_k$) is $\Sigma_k^p$-complete, while problems simultaneously complete for \emph{all} levels of PH are not known.}:

 \begin{theorem}[Informal; formal statement in \Cref{thm:MHScomplete}]\label{thm:MHScompleteInformal}
  For any $\epsilon\in\Omega(1/\exp)$ and constant $d\geq 2$, computing an $\epsilon$-approximate product-state solution to $k$-QSAT on qudits with WSDR is $\MHSo{\Theta(\epsilon)}$-complete.
\end{theorem}
\noindent As even finding common roots of homogeneous polynomial systems in $n+1$ variables and $n$ equations remains an open problem~\cite{Gre14}, we interpret \Cref{thm:MHScompleteInformal} as implying QSAT with SDR is intractable. Thus, we have the surprising juxtaposition that while classical SAT with SDR is easy, its quantum analogue is not.
\begin{figure}[t]
  \begin{center}
    \begin{tabular}{c | c | c}
      Problem & Complexity & Reference\\
      \hline
      SAT with SDR & Poly-time solvable & Folklore (?)\\
      QSAT with SDR & MHS-complete & This paper (\Cref{thm:MHScompleteInformal})\\
      SAT with SDR $+$ $O(1)$ additional clauses & Poly-time solvable & This paper (\Cref{thm:BKS})\\
      QSAT with SDR $+$ one additional clause & NP-complete & \cite{goerdtMatchedInstancesQuantum2019}, this paper (\Cref{thm:added})
    \end{tabular}
  \end{center}
  \caption{The complexity of variants of Classical SAT with SDR (denoted SAT with SDR) versus Quantum SAT with SDR (QSAT with SDR). Formally, "poly-time solvable" means in the complexity class Function Polynomial Time (FP), i.e. a poly-time classical Turing machine can compute a satisfying assignment.}
  \label{fig:complexity}
\end{figure}

\paragraph{c. A new subclass of TFNP based on the Fundamental Theorem of Algebra.} To help understand the complexity of $\MHS$, we give our second main result, which defines a second TFNP subclass, involving a \emph{single}, \emph{high}-degree, \emph{uni}variate polynomial equation. Below, a \emph{sparse} polynomial (\Cref{def:sparsepoly}), is one whose number of non-zero coefficients is logarithmic in its degree.

\begin{definition}[Sparse Fundamental Theorem of Algebra (SFTA)  (Informal; see \Cref{def:SFTA})]\label{def:informalSFTA}
  $\SFTA$ is the set of total NP search problems poly-time reducible to finding an $\epsilon$-approximate root $r\in \CC$ of a sparse monic univariate polynomial $p\in \CC[x]$ of degree $d$, where $\abs{r}\in[0,1+2\log(d)/d]$. We view $d$ as exponentially large in the input size, and require $\epsilon\in \Omega(1/\poly(d))$.
\end{definition}
\noindent As implied by its name, SFTA is inspired by the Fundamental Theorem of Algebra (\Cref{thm:FTA}), which recall states that any non-constant complex polynomial has a complex root $r$. 
Two comments regarding restrictions in the definition: First, the sparsity ensures\footnote{Another possible definition generalizing ours is to encode a non-sparse polynomial succinctly via a poly-size circuit which, given index $i$, outputs the $i$th coefficient of $p$. For us, however, the sparsity is necessary for our proof technique behind \Cref{thm:embedsparseInformal}.} by definition that the degree $d$ is exponential in the encoding size of polynomial $p$. This is important, as root approximations can be computed in $\poly(d)$ time (see e.g. ~\cite{schonhageEquationSolvingTerms1985}, as used in Section 4.4 of \cite{aldiEfficientlySolvableCases2021}), and thus the roots of a non-sparse polynomial can in general be efficiently approximated. Second, requiring $\abs{r}\in[0,1+2\log(d)/d]$ is without loss of generality (\Cref{l:rootbound}), and is in fact necessary in order to prove $\SFTA\subseteq\TFNP$ (\Cref{thm:inTFNP})\footnote{For example, if $d$ is exponential, then $p(2)$ can be \emph{doubly} exponentially large, and thus not representable with polynomially many bits.}. \cite{schonhageEquationSolvingTerms1985}
  
We now ask: \emph{What is the relationship between MHS and SFTA?} We first conjecture $\SFTA\subseteq \MHS$, and are able to prove the following:

\begin{theorem}[SFTA is in zero-error MHS (Informal; see \Cref{thm:embedsparse})]\label{thm:embedsparseInformal}
  Let $p$ be an $s$-sparse polynomial of degree $d$.
  Then, $p$ can be efficiently reduced to an instance $\Pi$ of QSAT with SDR of size $O(s\log(d))$, meaning $p(x/y) = 0$ if and only if $\ket{v}:=\ket{v_1}\otimes\dots \otimes \ket{v_N}$ is an exact solution to $\Pi$, for $\ket{v_1}=(x,y)^T\in\CC^2$. 
\end{theorem}
\noindent In words, $\SFTA$ can be reduced to QSAT with SDR if we require $\ket{v}$ to \emph{perfectly} satisfy all clauses, i.e. $\SFTA$ is contained in the version of $\MHS$ with error $\epsilon=0$. 
(Recall, however, that we do not allow $\epsilon=0$ in \Cref{def:MHSinformal}, as the resulting class does not obviously allow poly-time verification of solutions.)
We believe a more careful analysis of our construction behind \Cref{thm:embedsparseInformal} should yield the desired containment in $\MHS$.

In the reverse direction, we believe $\MHS\not\subseteq\SFTA$. This belief notwithstanding, by leveraging an old result of Canny~\cite{cannyAlgebraicGeometricComputations1988}, we show that generic (\Cref{def:generic}) instances of QSAT with WSDR \emph{can} be embedded into the roots of a single, high-degree polynomial $p$ (\Cref{thm:pspace}). (In fact, one obtains something stronger, known as a \emph{geometric resolution}, i.e. a set of rational functions $\set{r_i}$, so that when $r_i$ is fed the $j$th root of $p$, it produces the $i$th amplitude of the $j$th solution to QSAT.)
The polynomials $p$ and $r_i$, however, are only poly-\emph{space} computable, which is why this cannot yield $\MHS\subseteq\SFTA$.\\

\noindent\emph{NP-hardness results.} Via the construction of \Cref{thm:embedsparseInformal}, we can also show that even \emph{slight} variants of QSAT with SDR are no longer in TFNP (assuming $\PP\neq \NP$), but rather NP-hard. 

\begin{restatable}{theorem}{thmEqual}\label{thm:equal}
  It is $\NP$-hard to decide whether a $3$-QSAT system with an SDR has a product state solution, such that $|x|=|y|$, where $x,y$ are the entries of a prespecified qubit.
\end{restatable}
\begin{restatable}{theorem}{thmAdded}(c.f. \cite{goerdtMatchedInstancesQuantum2019})\label{thm:added}
  It is $\NP$-hard to decide whether a $3$-QSAT system with an SDR and one additional clause has a product state solution.
\end{restatable}

\noindent The second result above was first shown by Goerdt~\cite{goerdtMatchedInstancesQuantum2019} using different techniques. 

Finally, to complete the picture, we show that in contrast to \Cref{thm:added}, classical SAT with SDR with $O(1)$ additional clauses again becomes easy (\Cref{thm:BKS})! This mirrors precisely the behavior \Cref{thm:MHScompleteInformal} exhibits for $\MHS$-hardness of QSAT with SDR versus the fact that classical SAT with SDR is efficiently solvable; see \Cref{fig:complexity}.

\paragraph{d. Efficiently solvable special cases of QSAT with WSDR.} Since the MHS-completeness of \Cref{thm:MHScompleteInformal} suggests QSAT with WSDR cannot be efficiently solved, the last part of this work rounds out our study by showing how to extend the parameterized algorithm of \cite{aldiEfficientlySolvableCases2021} in three different directions to solve new special cases efficiently. 

Our first two results here concern the qubit case, and are complementary. 
In this setting, \cite{AdBGS21} efficiently solves QSAT with SDR for {generic} (\Cref{def:generic}) instances of \emph{transfer type} $b=n-m+1$ (\Cref{def:transfer-type}), where $m$ denotes the number of constraints and $n$ the number of qubits.
Recall \emph{non-generic} instances allow constraints which are not entangled across some bipartite cuts, and a transfer filtration (\Cref{def:transfer-type}) of transfer type $b$ is a type of hyperedge ordering built on an initial subset of $b$ qubits.

We first show that the generic assumption can be dropped if one assumes an ``almost extending edge order'' (\Cref{def:edgeorder}), which in turn implies the existence of an SDR~\cite{AdBGS21}: 
\begin{theorem}[Informal; see \Cref{thm:alg}]\label{thm:algInformal}
  Let $\Pi$ be a $k$-QSAT instance on qubits whose interaction hypergraph $G$ has an almost extending edge order of radius $r$.
  Then an $\epsilon$-approximate solution can be computed in time $\poly(L, \log1/\epsilon, k^r)$, where $L$ is the input size.
\end{theorem}

\noindent We then show that, instead of dropping the generic assumption, one can instead relax the transfer type assumption and still obtain a parameterized algorithm:

\begin{theorem}[Informal; see \Cref{thm:parameterized}]\label{thm:parameterizedInformal}
  Let $\Pi$ be a $k$-QSAT instance on qubits whose interaction hypergraph $G$ is $k$-uniform and has a $(k-1)$-almost extending edge order with radius $r$.
  Then an $\epsilon$-approximate solution can be computed in time $\poly(L, \abs{\log\epsilon}, k^r, m^k)$, where $L$ is the input size.
\end{theorem}

Finally, we sketch how to extend the algorithm of~\cite{AdBGS21} to QSAT on qu\emph{d}its with \emph{W}SDR. 
This allows us to obtain an exponential speedup over brute force for solving a new high-dimensional, non-trivial (but artificial) infinite family of instances on \emph{Pinwheel Hypergraphs} (\Cref{fig:pinwheel}).

\subsection{Techniques}\label{sscn:techniques} 
For brevity, we focus on our main results, (b) and (c). Brief techniques overviews for (a) and (d) are given at the beginning of their respective sections, \Cref{scn:prodSolution} and \Cref{scn:algs}.

\paragraph{b. A new subclass of TFNP based on \Bez's theorem.}  
For the MHS-completeness in \Cref{thm:MHScompleteInformal}, containment in $\MHS$ holds since \PS\ can be written as a special case of solving multi-homogeneous systems as follows.
In the case of $2$-QSAT, for example, a tensor product state $\ket{\alpha_1,\beta_2}:=\ket{\alpha}\otimes\ket{\beta}$ on two qubits satisfies a $2$-local constraint $\ket{\phi}$ if and only if
\begin{align}
  0=\braket{\phi}{\alpha_1,\beta_2}=\sum_{i,j\in[2]}\phi^*_{i,j}\alpha_i\beta_j.
\end{align} 
The right hand side above is a multilinear polynomial in the amplitudes $\set{\alpha_1,\alpha_2}$ (respectively, $\set{\beta_1,\beta_2}$) of $\ket{\alpha}$ (respectively, $\ket{\beta}$).
So, we will treat these amplitudes as variables in a system of multi-linear polynomials.
The catch is that there is an independent normalization condition implicit on each qudit's amplitudes; in our example here, both $\abs{\alpha_1}^2+\abs{\alpha_2}^2=1$ and $\abs{\beta_1}^2+\abs{\beta_2}^2=1$ must be independently satisfied. 
Since we will later work in projective space, however, this normalization is not explicitly enforced (other than the implicit constraint $\ket{\alpha},\ket{\beta}\neq 0$). 
Instead, we must allow the amplitudes of $\ket{\alpha}$ and $\ket{\beta}$ to adhere to different ``length scales'', since the assignments our system gives to them may lead to different norms for each vector.
And now we come to why we require \emph{multi}-homogeneous systems instead of homogeneous systems in this paper --- recall that by definition, a multi-homogeneous system allows us to partition variables into sets $S_i$, so that each polynomial is homogeneous with respect to each $S_i$. 
Thus, by setting $S_i$ to represent the amplitudes of qudit $i$, we obtain that each quantum constraint is independently homogeneous with respect to each qudit $i$. (Each monomial will have degree $0$ or $1$, depending on whether the constraint acts on qudit $i$.) In other words, each qudit's amplitudes implicitly has its own independent normalization.

As for hardness, to reduce multi-homogeneous systems to \PS, the ideal aim is to represent each variable group by a single qudit. 
In other words, if variable group $S_i$ contains $n_i$ variables, we embed each variable as an amplitude of an $n_i$-dimensional qudit $q_i$. 
The first problem this presents is that monomials in a multi-homogeneous system need not be \emph{linear} in each variable set $S_i$.
To thus simulate non-linearity, we create multiple copies of each $q_i$; by placing constraints on these simultaneously, we can create products of amplitudes from $q_i$.
However, this raises a second challenge --- this logic only holds when each copy of $q_i$ has an \emph{identical} assignment! 
The natural way to resolve this is to enforce equality between all copies of $q_i$ by adding projectors onto the antisymmetric subspace. 
This, however, does not work for us, as the rank of the antisymmetric subspace for qu$d$its with $d>2$ is too large, requiring the addition of too many rank-$1$ constraints for an SDR to exist. 
To overcome this, we instead utilize the qudit-to-qubit reduction from our second proof of \Cref{thm:wsdr-product-solution}, which is a mapping iteratively replacing each $d$-dimensional qudit with a pair of $2$- and $(d-1)$-dimensional qudits. 
Thus, each qu$d$it is replaced with $d-1$ qu$b$its, and we show that the mapping preserves \PS\ solutions.
We are finally now in business, because on pairs of \emph{qubits}, the projector onto the antisymmetric subspace is of rank $1$, and thus we can show that there exists an SDR for the instance output by our reduction.

\paragraph{c. A new subclass of TFNP based on the Fundamental Theorem of Algebra.} We discuss the proof of \Cref{thm:embedsparseInformal}, which recall shows how to embed the roots of an arbitrary sparse polynomial $p$ of exponential degree $d$ into the solution set of a QSAT with SDR instance.
The tool we start with is a \emph{transfer function} (used also, e.g., in \cite{bravyiEfficientAlgorithmQuantum2006,laumannProductGenericRandom2010}; see \Cref{lem:transfer-function}), which roughly is the quantum generalization of the following standard classical approach for propagating assignments: Given (e.g.) clause $(x\vee y\vee z)$, if $x=y=0$, then $z=1$ necessarily. 
Via this tool, we show how to design $2$-local (respectively, $3$-local) rank-$1$ QSAT constraints which force a target qubit to encode  any desired \emph{linear} (respectively, \emph{quadratic}) operations on an input state $(x, y)^T$. 
For example, via a $2$-local constraint $\ket{\phi_{12}}$ on qubits $1$ and $2$, we can enforce that if qubit $1$ has assignment $(x,y)^T$, then in order to satisfy $\phi_{12}$, qubit $2$ must be set (proportional to) $(a_1x + a_2 y, b_1x + b_2 y)^T$, for any desired $\abs{a_1}^2+\abs{a_2}^2=\abs{b_1}^2+\abs{b_2}^2=1$.

With these gadgets in hand, we then move to encoding input polynomial $p$ into QSAT by designing three sets of clauses. 
To begin, we homogenize $p(x)$ to a bivariate polynomial $q(x,y)$, and let $\ket{v_0}=(x,y)^T$ denote an assignment to the first qubit.
Ultimately, this $x$ and $y$ will end up encoding our roots to $p$.
Our first set of contraints uses transfer functions and square-and-multiply to create new qubits of various powers of $x$ and $y$, i.e. ``power qubits'' whose assignments must be proportional to $(x^i, y^i)^T$.
Our second set of constraints then combines these power qubits with our transfer function gadgets to recursively construct $q(x,y)$ in a final target qubit, whose assignment must be proportional to $(q(x,y), y^d)^T$.
The third set is a single constraint, which forces the target qubit's state $(q(x,y), y^d)^T$ to be proportional to $(0, 1)$, which enforcing $q(x,y)=0$. 
By ``undoing'' the homogenization, we can then show that $p(x/y)$ must be a root of $p$.

\subsection{Discussion and open questions}\label{sscn:discussion}

\paragraph{Question and answer.} As this work bridges rather disjoint areas of study (TFNP, polynomial systems, and quantum satisfiability), we address possible comments/questions to set further context.
\begin{enumerate}
  
  \item \emph{Are product state solutions to quantum satisfiability problems interesting?} Generally speaking, yes. Although solutions to quantum satisfiability problems are typically entangled, product state solutions have a long history of being used as an ansatz to study properties of local Hamiltonians (i.e. ``quantum constraint satisfiability problems'') in the \emph{mean-field theory} physics literature~\cite{gharibianQuantumHamiltonianComplexity2015}. For example, mean-field ansatzes suffice to efficiently approximate ground state energies of planar~\cite{bansalClassicalApproximationSchemes2009,brandaoProductStateApproximationsQuantum2016} and dense~\cite{gharibianApproximationAlgorithmsQMAComplete2012,brandaoProductStateApproximationsQuantum2016} local Hamiltonians to within any desired relative error $(1\pm \epsilon)$ for $\epsilon>0$. In the case of $2$-local frustration free Hamiltonians (as in $2$-QSAT), \emph{exact} product-state solutions always exist and can be found ~\cite{bravyiBoundsQuantumSatisfibility2009,chenNogoTheoremOneway2011}, which has implications such as the fact that such Hamiltonians cannot be used to prepare resource states for one-way quantum computing~\cite{chenNogoTheoremOneway2011}. 
  
  \item \emph{Why is adding SDRs to the picture interesting?} PRODSAT with \emph{SDR} is interesting as it falls under the ``dimer model'' of physics~\cite{kenyonWhatDimer2005}, which is useful as it is (1) exactly solvable and (2) aids in understanding phase transitions, which are typically difficult to study. For example, the original motivation of~\cite{laumannProductGenericRandom2010} was to understand the SAT-UNSAT phase transition in random QSAT instances. Therein, dimer coverings/SDRs were used to show that for clause densities below a certain $k$-dependent threshold, random $k$-QSAT instances are satisfiable with probability $1$ by a product state solution. While this did not perfectly resolve the exact SAT-UNSAT threshold, it significantly improved previously known lower bounds.

  \item \emph{Typically TFNP classes (e.g. PPAD) are defined via a complete problem whose input is a circuit succinctly encoding an exponentially large object (e.g. a circuit succinctly encoding an exponentially large graph for END-OF-LINE). On the other hand, MHS and SFTA, have their input explicitly written out?} 
  This is a good discussion point. Traditional ``syntactic'' circuit-based definitions have the advantage that the existence principle for the class is captured by a simple combinatorial complete problem, which can make reasoning about the class easier. This, however, has a downside --- proving hardness results for new problems \emph{not} specified by input circuits, which are arguably more natural, can be more challenging (see, e.g. G\"o\"os, Kamath, Sotiraki and Zampetakis'~\cite{goosComplexityModuloqArguments2020}  non-circuit based PPA$_p$-complete problem ($p\geq 3$ a prime) for the Chevalley-Warning theorem). In contrast, MHS and SFTA may be thought of as ``white-box'' TFNP subclasses, in that the object to be studied (i.e. polynomial equations) is specified explicitly, rather than succinctly via circuit. On the negative side, this has the downside of potentially obscuring the relationship between the class and the existence principle. On the positive side, it can bring establishing hardness results for further natural problems within reach, since the artificial circuit input encoding is bypassed. 
  In our case, this motivation is further strengthed  by the fact that MHS and SFTA are based on polynomials, which themselves are ubiquitous in the sciences, yielding a potentially promising route for characterizing the  complexity of new TFNP problems.


  \item \emph{Is there also combinatorial principle underlying MHS?} Yes and no. No, in that the existence principle for MHS is B\'ezout's theorem, which is algebraic geometric. Yes, in that checking if the B\'ezout number $\bez>0$ boils down to checking if a certain bipartite graph has a perfect matching (\Cref{obs:countWSDR}). More generally, computing $\bez$ itself counts the number of perfect matchings in said graph (which is intractable, but also not necessary for our purposes).
   
  \item \emph{Can MHS or SFTA be related to existing TFNP subclasses?} This would be indeed ideal, but our attempts thus far have not succeeded. The most obvious candidate is PPAD, due to its connection~\cite{Pap94} to Brouwer's fixed point theorem. This is because there is a natural algorithm
   using transfer functions to attempt to solve QSAT with SDR; roughly, this algorithm aims to converge to a product state assignment which is a fixed point under all local transfer functions. Unfortunately, Brouwer's theorem requires convex sets, and the set of product state solutions is \emph{not} convex. Moreover, the standard approach of moving to the convex hull of product states (i.e. mixed separable states) seems to break the transfer function formalism. We thus leave this as what we feel is an important and interesting open question.
\end{enumerate}

\paragraph{Conclusion and open questions.} We have defined and studied two TFNP subclasses connected to complex polynomial systems. The first, Multi-Homogeneous Systems (MHS), leads to the first formal proof of a quantum problem which, on the one hand, is guaranteed to have a ``simple'' (i.e. tensor product) solution, and on the other hand, is potentially intractable. 
As even the ``simpler'' setting of finding common roots of homogeneous polynomial systems in $n+1$ variables and $n$ equations is believed difficult~\cite{Gre14}, we thus view MHS-hardness as a viable indicator for computational hardness. 
Our second class, Sparse Fundamental Theorem of Algebra (SFTA), was used to show that the problem of computing roots of sparse high-degree univariate polynomials can be embedded into computing exact solutions to QSAT with SDR, thus showing SFTA is contained in the zero-error version of MHS. We conjecture in fact that $\SFTA\subseteq\MHS$ --- can this be shown? 

As each member of the trinity studied here (TFNP, polynomial systems, and quantum satisfiability problems) is unto itself a research field, many questions in their intersection remain open. For example, which natural \emph{classical} problems might be complete for MHS or SFTA? Are there other TFNP subclasses related to polynomial systems over complex numbers? As discussed in ``question and answer'' above, can MHS or SFTA be related to standard TFNP subclasses such as PPAD? Similarly, how is the setting of ``syntactic'' (i.e. circuit-based) TFNP subclasses to be understood versus our ``white-box'' setting for MHS and SFTA? 


\paragraph{Organization.} \Cref{scn:prelim} states basic definitions, including formally defining QSAT, \PS, and the connection between \PS\ and polynomial systems. \Cref{scn:WSDR} introduces Weighted SDRs (WSDR), which are then used in \Cref{scn:prodSolution} to give our two proofs of \Cref{thm:wsdr-product-solution}, i.e. that QSAT with WSDR always has a solution. \Cref{scn:MHS} defines our class MHS and proves MHS-completeness of QSAT with SDR (\Cref{thm:MHScompleteInformal}). \Cref{scn:sparsepoly} defines class SFTA, studies its relationship to MHS, and gives the NP-hardness results of \Cref{thm:equal} and \Cref{thm:added}. \Cref{scn:algs} give efficient algorithms for special cases of QSAT with WSDR.

\section{Preliminaries}\label{scn:prelim}

We assume a basic background in quantum computation, see e.g.~\cite{NC00}. 
Basic background in algebraic geometry (e.g. definitions of projective space and varieties) would be helpful for \Cref{sscn:chow} in particular, which introduces the Chow ring, though we have attempted to make this accessible with intuition throughout; see e.g.~\cite{Shafarevich1974,coxIdealsVarietiesAlgorithms2015} for references. 

\paragraph{Notation and basic definitions.} We use $:=$ to indicate a definition. For $\ket{\psi}\in\CC^d$, we define $\norm{\ket{\psi}}_p:=(\sum_{i=1}^d\abs{\psi_i}^p)^{1/p}$. For a linear operator $M:\CC^d\rightarrow\CC^d$, we analogously define $\norm{M}_p$ on the singular values of $M$. $\CC[x_1,\ldots, x_n]$ denotes the set of complex polynomials acting on variables $x_1$ through $x_n$. Throughout this work, we work with polynomials over $\CC$, unless stated otherwise.

\begin{definition}[Lipschitz continuity]\label{def:lipschitz}
  We say function $f:\CC\rightarrow\CC$ is $K$-Lipschitz continuous if for all $x,y\in X$, $\abs{f(x)-f(y)}\leq K\abs{x-y}$.
\end{definition}

\begin{fact}\label{f:polylipschitz}
  Let $X\subseteq \CC$ be such that $\forall x\in X$, $\abs{x}\leq r$. Consider any complex polynomial $p=\sum_{k=0}^dc_kx^k$ of degree $d$, with $s$ non-zero coefficients each of magnitude at most $c$. Then, over set $X$, $p$ is $K$-Lipschitz continuous with $K=scr^{d-1}d$.
\end{fact}
\begin{proof}
  Let $S$ be the set of non-zero coefficients of $p$. Then, for any $x,y\in X$,
  \begin{align}
    \abs{p(x)-p(y)}\leq \sum_{i\in S}\abs{c_i} \abs{x^i-y^i}=\abs{x-y} \sum_{i\in S}\abs{c_i} \abs{\sum_{j=1}^{i}x^{i-j}y^{j-1}}\leq \abs{x-y}scr^{d-1}d.
  \end{align}
\end{proof}
\noindent Thus, when $c,d\in O(1)$, $K\in O(1)$. Note that \Cref{def:lipschitz} and \Cref{f:polylipschitz} can be straightforwardly generalized to the setting of multivariate polynomials.

\paragraph{Quantum SAT.} We begin by stating our basic formalism for QSAT on qudits. Formally, our QSAT Hamiltonians act on 
$
\cH=\CC^{d_1}\otimes \CC^{d_2}\otimes \cdots \otimes \CC^{d_n}
$
for some integers $d_1,\ldots,d_n\ge 2$. As is standard, we fix a computational basis $\set{\ket{0},\ldots,\ket{d_i-1}}$ for each qudit, so that an arbitrary vector in $\cH$ can be written 
\begin{equation}
\ket{\psi}=\sum_{j_1=0}^{d_{1}-1}\cdots \sum_{j_n=0}^{d_n-1} a_{j_1\cdots j_n}\ket{j_1\cdots j_n}
\end{equation}
for some choice of complex coefficients $a_{j_1\cdots j_n}$ satisfying $\sum_{j_1=0}^{d_{1}-1}\cdots \sum_{j_n=0}^{d_n-1} \abs{a_{j_1\cdots j_n}}^2=1$. (Since solutions to QSAT are null space vectors, the normalization of $\ket{\psi}$ will often not be important.) 

\begin{definition}[Quantum $k$-SAT on qudits ($k$-QSAT)]\label{def:QSAT}
  For $k$-QSAT on $n$ qudits:
  \begin{itemize}
    \item Input: A pair $\Pi=(\set{\Pi_i}_i, \alpha)$, for rational $\alpha>1/p(n)$ for some fixed polynomial $p$, and for projectors or \emph{clauses} $\Pi_{1},\ldots,\Pi_m\in\lin{(\cH)}$ of the form 
    \begin{align}
      \pi^{-1}( \ketbra{\psi_i}{\psi_i}\otimes I_{n-k})\pi,
    \end{align} 
    where $\pi$ is a permutation of the qudits, $\ketbra{\psi_i}{\psi_i}$ is a rank-$1$ projector acting on the first $k$ qudits, and $I_{n-k}$ is the identity on the remaining $n-k$ qudits. 
    \item Output: Output YES if there exists a unit vector $\ket{\psi}\in \cH$ such that $\Pi_i\ket{\psi}=0$ for all $i$, or NO if for all unit vectors $\ket{\psi}$, $\bra{\psi}\sum_i\Pi_i\ket{\psi}\geq \alpha$.  
\end{itemize}
\end{definition}

\paragraph{\PS\ and homogeneous polynomial systems.} In this paper, we interested in (approximate) \emph{product} solutions to QSAT, for which one defines the following problem, $\epsilon$-approximate \PS.
\begin{definition}[$\epsilon$-approximate $k$-\PS\ on qudits, decision version] \label{def:PS_decision}
  First, $k$-\PS\ is defined as $k$-QSAT on qudits (\Cref{def:QSAT}), except in the output the assignment $\ket{\psi}$ must be a pure tensor product state, i.e. $\ket{\psi}=\ket{\varphi_1}\otimes \cdots\otimes \ket{\varphi_n}$ with $\ket{\varphi_i}\in \CC^{d_i}$ for each $i\in \{1,\ldots,n\}$. Then, $\epsilon$-approximate $k$-\PS\ relaxes the YES case condition to $\bra{\psi}\sum_i\Pi_i\ket{\psi}\leq \epsilon$.
\end{definition}
Our main results, i.e. involving \MHS\ and \SFTA, focus on the search version of this problem, for which we assume (as is standard for QSAT)  that $k,d\in O(1)$:
\begin{definition}[$\epsilon$-approximate $k$-\PS\ on qudits, search version] \label{def:PS_search}
  Defined as $\epsilon$-approximate $k$-\PS, except in the YES case, a satisfying assignment $\ket{\psi}=\ket{\varphi_1}\otimes \cdots\otimes \ket{\varphi_n}$ with $\bra{\psi}\sum_i\Pi_i\ket{\psi}\leq \epsilon$ is to be output. 
  In terms of precision, recalling that $m$ is the number of clauses, it suffices to output each entry of each $\ket{\varphi_i}$ within additive error $\epsilon/\poly(m)$ to verify a YES case in NP (\Cref{rem:verifyNP}).  
\end{definition}

\begin{rem}(Verifying $\epsilon$-approximate $k$-\PS\ in \NP)\label{rem:verifyNP}
  Given $\ket{\psi}=\ket{\varphi_1}\otimes \cdots\otimes \ket{\varphi_n}$, we wish to verify $\bra{\psi}\sum_i\Pi_i\ket{\psi}=\sum_i\bra{\psi}\Pi_i\ket{\psi}\leq \epsilon$. For any $i$, suppose $\Pi_i$ without loss of generality acts on qudits $1$ through $k\in O(1)$. Then, 
  \begin{equation}\label{eqn:onei}
    \bra{\psi}\Pi_i\ket{\psi}= \bra{\varphi_1}\otimes\cdots\otimes\bra{\varphi_k}\Pi_i\ket{\varphi_1}\otimes\cdots\otimes\ket{\varphi_k},
  \end{equation} 
  which only involves matrix multiplication on systems of dimension $d^k\in O(1)$, and thus can be computed using a poly-time Turing machine. Thus, if each entry of each $\ket{\varphi_i}$ is specified within additive error $\epsilon/\poly(m)$, then for any $i$, \Cref{eqn:onei} can also be computed with additive error $\epsilon/\poly(m)$. Note this holds even for inverse exponential $\epsilon$, since the verification is on a classical Turing machine (as opposed to a quantum circuit verifier). Finally, since there are $m$ clauses $\Pi_i$, and each clause is a projector (i.e. has spectral norm $1$), the total additive error over all clauses can be upper bounded by $\epsilon$.
\end{rem}\noindent 

To next connect \PS\ with homogenous polynomial systems, expand both the qudits $\ket{\varphi_i}$ and the (possibly entangled) projectors $\Pi_i$ with respect to the computational basis $\ket{j_1\cdots j_n}$. Then, the problem of finding a satisfying assignment in product form is equivalent to solving a system of $m$ homogeneous equations of the form
\begin{equation}\label{eqn:multilinear}
\sum_{j_1=0}^{d_{1}-1}\cdots \sum_{j_k=0}^{d_n-1} a_{j_1\cdots j_k}x_{i_1,j_1}\cdots x_{i_k,j_k}=0,
\end{equation}
where $i_1,\ldots,i_k$ are the qudits on which the projector acts non-trivially, the constants $a_{j_1\cdots j_k}$ the (complex conjugate of the) amplitudes of the rank-$1$ constraint $\Pi_i$, and each variable $x_{i,j}$ the $j$th amplitude of the $i$th qudit.

\begin{example}
For instance, suppose $d_1=2$ and $d_2=3$ so that the first and second qudits are, respectively, a qubit  $\ket{\varphi_1}=x_{1,0}\ket{0}+x_{1,1}\ket{1}$ and a qutrit $\ket{\varphi_2}=x_{2,0}\ket{0}+x_{2,1}\ket{1}+x_{2,2}\ket{2}$. A general two-local constraint $\Pi_1=\ketbra{\psi}{\psi}$ for $\ket{\psi}=(a_{0,0},a_{0,1},a_{0,2},a_{1,0},a_{1,1},a_{1,2})^T$ being satisfied by assignment $\ket{\varphi_1}\otimes\ket{\varphi_2}$ is equivalent to the multilinear equation
\begin{equation}
a_{0,0}x_{1,0}x_{2,0}+a_{0,1}x_{1,0}x_{2,1}+a_{0,2}x_{1,0}x_{2,2}+a_{1,0}x_{1,1}x_{2,0}+a_{1,1}x_{1,1}x_{2,1}+a_{1,2}x_{1,1}x_{2,2}=0\,.
\end{equation}
\end{example}

\paragraph{Projective space and algebraic geometric view of \PS.} 

In parts of this paper (particularly \Cref{sscn:chow}), it will be useful to view \PS\ via the lens of projective space. Specifically, recall that vectors in $\CC^{d_i}$ differing by non-zero scaling represent the same physical state in the corresponding qudit, and that the property of being a non-zero null vector of a Hamiltonian is invariant under such scaling. Thus, \PS\ solutions correspond to points in $(d_i-1)$-dimensional complex projective space $\PP^{d_i-1}(\CC)$. (Formally, projective space treats two non-zero rays in the same direction as equivalent, regardless of their respective norms.) The drop in dimension from $d_i$ to $d_i-1$ happens since one can rescale each qudit's local assignment $\ket{\varphi_i}\in\CC^{d_i}$ so that its first amplitude is $1$, and thus can be ignored. Of course, this assumes the assignment $\ket{\varphi_i}$ did not set its first amplitude to zero, which is generically the case (\Cref{def:generic}), i.e. holds for almost all positive \PS\ instances.

We thus have that $n$-qudit product states are in correspondence with points of the complex projective variety\footnote{Roughly, a variety is simply the set of solutions to a given set of polynomial equations.}
\begin{equation}\label{eqn:X}
\mathcal X_{d_1,\ldots,d_n}:=\PP^{d_1-1}(\CC)\times \cdots\times \PP^{d_n-1}(\CC)\,.
\end{equation}
In this geometric interpretation, each clause $\Pi_i$ defines a hypersurface $V_i\subseteq \mathcal X_{d_1,\ldots,d_n}$ which is of degree $1$ in each of the variables corresponding to qudits on which $\Pi_i$ acts nontrivially. As a consequence, the problem of finding a product solution to the given instance of QSAT is equivalent to the geometric problem of finding a point in the intersection $V_1\cap V_2\cap \cdots \cap V_m$.

Finally, when we speak of \emph{generic} instances of \PS, we mean with respect to the following definition.
\begin{definition}[{Genericity \cite[Def. 5.6]{CLO05}}]\label{def:generic}
  A property is said to \emph{hold generically} for a set of polynomials $f_1,\dots,f_n$ with indeterminate coefficients $c_{i,j}$ if there is a nonzero polynomial $g$ in the $c_{i,j}$ such that the property holds for all $f_1,\dots,f_n$ for which $g(\cdots)\ne0$.
\end{definition}
\noindent As mentioned above, ``generic'' means ``for almost all'' instances. A simple example of a property which holds generically is that of a $2\times 2$ real matrix $M$ being invertible. In this case, the polynomial $g$ is the determinant $\det(M)=M_{11}M_{22}-M_{12}M_{21}$, since $M$ is invertible if and only if $\det(M)\neq 0$.

\section{Weighted Systems of Distinct Representatives (WSDR)}\label{scn:WSDR}

We now define a Weighted System of Distinct Representatives (WSDR), and prove several properties.

\subsection{Definitions} \label{sscn:defP}

\begin{definition}[Weighted hypergraph]
A {\it weighted hypergraph} is a pair $(G,w)$ consisting of a hypergraph $G$ and a {\it weight function} $w:V(G)\to \mathbb Z_{\ge 0}$.
\end{definition}

\noindent Thus, a hypergraph $G$ without weights on its vertices may be viewed as a weighted hypergraph $(G,1)$ with the weight function defined by $w(v)=1$ for all $v\in V(G)$.

\begin{definition}[Weighted System of Distinct Representatives (WSDR)]\label{def:weighted}
A {\it Weighted System of Distinct Representatives} for weighted hypergraph $(G,w)$ is a mapping $f:E(G)\rightarrow V(G)$, such that 
\begin{enumerate}
  \item (edges contain their representatives) for any $e\in E(G)$, $f(e)\subseteq e$,
  \item (each edge has at least one representative) $\abs{f(e)}\geq 1$ for all $e\in E$, and
  \item (each vertex $v\in V(G)$ is the representative for at most $w(v)$ edges) $|f^{-1}(v)|\le w(v)$ for all $v\in V(G)$.
\end{enumerate}
\end{definition}

\begin{rem}
A hypergraph $G$ has a (non-weighted) system of distinct representatives (SDR) if and only if $(G,1)$ has a WSDR. Hence, WSDRs generalize SDRs.
\end{rem}

\noindent As an aside, a function $f$ that to each edge $e\in E(G)$ assigns a vertex $f(e)\in e$ is more generally known as a hypergraph orientation~\cite{frankkiralykiraly2003}. 
There exist works which study connections between hypergraph orientations and multi-homogeneous polynomial systems (e.g.~\cite{BEKT2022}), but for clarity, as far as we are aware our definition of WSDR appears distinct from the hypergraph orientations used previously in the literature.

\begin{definition}[Vertex set size with respect to a weight function]
 Let $(G,w)$ be a weighted hypergraph and let $S$ a set of vertices of $G$. The {\it size of $S$ with respect to $w$} is the integer 
  \begin{align}
    \abs{S}_w:=\sum_{v\in S} w(v).
  \end{align}
\end{definition}

\begin{example}
If $w$ is the constant function $1$, then $|S|_1=|S|$ is the cardinality of $S$.
\end{example}

\subsection{Existence and computation of WSDRs} 

When does a weighted hypergraph have a WSDR? Hall's classic Marriage theorem gives a necessary and sufficient condition for when a (non-weighted) hypergraph has a (non-weighted) SDR. Here, we state its weighted case. As we were not able to find a proof thereof of such a statement in the literature, we provide one in \Cref{app:proof} for completeness. 

\begin{restatable}{theorem}{thmHMT}(Hall's Marriage Theorem for weighted hypergraphs)\label{thm:HMT}
  Let $(G,w)$ be a weighted hypergraph. For each collection $X$ of edges of $G$, let $V_X$ be the set of vertices that are contained it at least one edge of $X$. Then $(G,w)$ has a WSDR if and only $|V_X|_w\ge |X|$ for every $X\subseteq E(G)$.
\end{restatable}
\noindent In the special case $w=1$, Theorem~\ref{thm:HMT} reduces to the usual Hall's Marriage Theorem.

\begin{rem}
An immediate consequence of Hall's Marriage Theorem is that $|V(G)|_w\ge |E(G)|$ is a necessary condition for $(G,w)$ to have a WSDR. 
\end{rem}

Via \Cref{thm:HMT}, we thus obtain the following sufficient condition for when $G$ has a WSDR.

\begin{corollary}\label{cor:1}
  Let $(G,w)$ be a weighted hypergraph such that $\deg(v)\le |e|_w$ for every $v\in V(G)$ and every $e\in E(G)$, where $\deg(v)$ denotes the degree of the vertex $v$. Then $(G,w)$ has a WSDR.
  \end{corollary}
  
  \begin{proof}
  For every $X\subseteq E(G)$, by double counting,
  \begin{equation}
  |X| \min_{e\in X}|e|_w \le \sum_{e\in X}|e|_w \leq \sum_{v\in V_X} w(v)\deg(v) \le |V_X|_w \max_{v \in V(G)} \deg(v)\le |V_X|_w \min_{e\in X}|e|_w\,.
  \end{equation}
  Hence $|X|\le |V_X|_w$ and the results follows from \Cref{thm:HMT}.
  \end{proof}
  
  In uniform hypergraphs, precise necessary and sufficient criteria can be formulated as follows.

  \begin{definition}[$k$-Uniform Hypergraph]
  A weighted hypergraph $(G,w)$ is {\it $k$-uniform} for some positive integer $k$ if $|e|_w=k$ for every $e\in E(G)$.
  \end{definition}
  
  \begin{corollary}\label{cor:uniform-wsdr}
  Let $(G,w)$ be a $k$-uniform weighted hypergraph such that $\deg(v)=d$ for every $v\in V(G)$. Then $(G,w)$ has a WSDR if and only if $d\le k$.
  \end{corollary}
  
  \begin{proof}
  In one direction this follows immediately from Corollary~\ref{cor:1}. In the opposite direction, if $(G,w)$ has a WSDR, then $|E(G)|\le |V(G)|_w$ by \Cref{thm:HMT}. Hence $d|V(G)|_w=|E(G)|k\le |V(G)|_w k$ from which the result easily follows.
  \end{proof}

\begin{rem}[Computation of WSDRs]\label{rem:computeWSDR}
WSDRs can be efficiently computed. Namely, given a weighted hypergraph $(G,w)$ satisfying the conditions of Theorem~\ref{thm:HMT}, computing a WSDR reduces to computing a maximum matching in the bipartite graph $G'$ with $V(G') = V_1 \cup V_2$, where $V_1 = E(G)$, $V_2 = \{v_i \mid v\in V(G), i\in[w(v)]\}$, and $E(G') = \{\{e,v_i\}\mid e\in V_1,v_i\in V_2,v\in e\}$ (see~\cite{Gal86} for a survey).
Alternatively, the WSDR may also be computed using a maximum flow algorithm (see~\cite{cruz2023survey} for a survey).
\end{rem}

\subsection{WSDRs under graph operations} 
Finally, we study WSDRs under the cartesian product of hypergraphs, defined next. This will be useful in \Cref{sscn:weightedGraphsConst}.

\begin{definition}[Hypergraph Cartesian Product]\label{def:cartesian}
The {\it cartesian product} of two weighted hypergraphs $(G_1,w_1)$ and $(G_2,w_2)$ is the weighted hypergraph $(G_1,w_1)\Box(G_2,w_2)=(G_1\Box G_2,w_1\Box w_2)$ where $G_1\Box G_2$ is the usual cartesian product of hypergraphs such that $V(G_1\Box G_2)=V(G_1)\times V(G_2)$ and\
\begin{equation}
E(G_1\Box G_2)=\left(\bigcup_{v_1\in V(G_1)} \{v_1\}\times E(G_2)\right) \cup \left( \bigcup_{v_2\in V(G_2)} E(G_1)\times \{v_2\} \right)
\end{equation}
while $(w_1\Box w_2)((v_1,v_2))=w_1(v_1)+w_2(v_2)$ for all $v_1\in V(G_1)$ and $v_2\in V(G_2)$.
\end{definition}

\begin{rem}[WSDRs under cartesian products]
Cartesian products preserve WSDRs in the following sense. Let $(G_1,w_1)$ and $(G_2,w_2)$ be weighted hypergraphs admitting, respectively, WSDRs $f_1$ and $f_2$. Let $f_1\Box f_2: E(G_1\Box G_2)\to V(G_1)\times V(G_2)$ be such that $(f_1\Box f_2)(\{v_1\}\times e_2)=(v_1,f_2(e_2))$ for all $e_2\in E(G_2)$, $v_1\in V(G_1)$ and $(f_1\Box f_2)(e_1\times \{v_2\})=(f_1(e_1),v_2)$ for all $e_1\in E(G_1)$, $v_2\in V(G_2)$. Since
\begin{equation}
(f_1\Box f_2)^{-1}(v_1,v_2)=(\{v_1\}\times f_2^{-1}(v_2))\cup (f_1^{-1}(v_1)\times \{v_2\})
\end{equation}
then $f_1\Box f_2$ is a WSDR for $(G_1,w_1)\Box (G_2,w_2)$.
\end{rem}

\begin{example}
Let $C_n$ be a cycle on $n\ge 3$ vertices. Then $C_n$ has an SDR and $(C_n\Box C_m,2)$ has a WSDR for every $n,m\ge 3$. However $C_n\Box C_m$ has no SDR since $|V(C_n\Box C_m)|=nm < 2nm =|E(C_n\Box C_m)|$.
\end{example}


\section{Existence results via Weighted SDRs}\label{scn:prodSolution}
We now show our first main result, \Cref{thm:wsdr-product-solution}, which recall shows that QSAT with WSDR always has a product state solution. We give two proofs of this fact: Via the Chow ring (\Cref{sscn:chow}) and via reduction to the qubit case (\Cref{sscn:reductionQubits}). 

\subsection{Approach 1: Via the Chow Ring}\label{sscn:chow}

Our first proof goes via the Chow Ring from algebraic geometry, which is defined in \Cref{ssscn:chow}. With the necessary definitions in hand, the proof itself is simple and given in \Cref{ssscn:prodsols}.

\paragraph{Brief overview of techniques.} 
To show that QSAT with WSDR always has a solution (\Cref{thm:wsdr-product-solution}), recall we give two proofs, one based on the Chow ring, and the other based on a reduction from qu\emph{d}its to qu\emph{b}its. We now give an overview for the former. (The latter was already sketched in \Cref{scn:intro}.) At a high level, the Chow ring approach uses intersection theory~\cite{fultonbook, 3264book, Shafarevich1974}. One reason for the effectiveness of this approach in the study of \PS\ (i.e. product state solutions to QSAT) is that intersection theory is designed to work with generic constraints. This is in essence why important intersection-theoretic quantities, such as the B\'ezout number, are encoded into the interaction hypergraph. 
More concretely, the key property of the Chow ring we leverage is as follows (\Cref{f:subvarietyrep}): Given a set of rank-$1$ QSAT constraints with solution sets $\set{V_1,\ldots, V_r}$ (formally, hypersurfaces), the Chow ring has a canonical mapping from each $V_i$ to a ``representative'' of the Chow ring itself, denoted ${[V_i]}$. Then, if the product of these representatives is non-zero, i.e. ${[V_1]}\cdots{[V_r]}\neq 0$, one immediately has that $V_1\cap\cdots\cap V_r\neq \emptyset$, i.e. the solution sets to each constraint share a common solution. Conversely, if ${[V_1]}\cdots{[V_r]}= 0$, generically, no joint solution exists.

\subsubsection{Background on the Chow Ring}\label{ssscn:chow}

We refer to~\cite{3264book,fultonbook} for an in-depth discussion of the Chow ring of a variety. Here we limit ourselves to the multi-projective case which is relevant to \PS.
Recall we define $\mathcal X_{d_1,\ldots,d_n}:=\PP^{d_1-1}(\CC)\times \cdots\times \PP^{d_n-1}(\CC)$.

\begin{definition}
The {\it Chow ring} of $\mathcal X_{d_1,\ldots,d_n}$ is the commutative ring
\begin{equation}
CH(\mathcal X_{d_1,\ldots,d_n})=\mathbb Z[H_1,\ldots,H_n]/(H_1^{d_1},\ldots, H^{d_n}_n).
\end{equation}
\end{definition}

\begin{example}
The Chow ring of $\PP^2(\CC)=\mathcal X_{3}$ is $CH(\mathcal X_3)=\mathbb Z[H]/(H^3)$. As a set, it consists of linear combinations $a1+ bH+cH^2$, with $a,b,c\in \mathbb Z$, and multiplication
\begin{equation}
(a1+ bH+cH^2)(a'1+b'H+c'H^2)=aa'1+(ba'+ab')H+(ca'+bb'+ac')H^2.
\end{equation}
\end{example}

\begin{example}
The Chow ring of $\PP^1(\CC)\times \PP^1(\CC)=\mathcal X_{2,2}$ is $CH(\mathcal X_{2,2})=\mathbb Z[H_1,H_2]/(H_1^2,H_2^2)$. As a set it consists of linear combinations $a+bH_1+cH_2+dH_1H_2$, for all $a,b,c,d\in \mathbb Z$ with multiplication
\begin{equation}
(a1+bH_1+cH_2+dH_1H_2)(a'1+b'H_1+c'H_2+d'H_1H_2)=a''1+b''H_1+c''H_2+d''H_1H_2
\end{equation}
where $a''=aa'$, $b''=ab'+ba'$, $c''=ac'+ca'$, and $d''=ad'+bc'+cb'+da'$.
\end{example}

This first proof of \Cref{thm:wsdr-product-solution} will crucially use the notion of ``representatives'' $[V]$ of subvarieties $V$ relative to the Chow ring. For this, let $Z(\mathcal X_{d_1,\ldots,d_n})$ be the free abelian group of {\it cycles}, generated by subvarieties of $\mathcal X_{d_1,\ldots,d_n}$. Linear combinations $n_1V_1+\cdots+n_kV_k$ with positive coefficients can be thought of as the union of $n_1$ copies of the subvariety $V_1$, $n_2$ copies of the subvariety $V_2$, etc.

\begin{definition}[Subvariety representative, {$[V]$}]\label{def:subvarietyrep}
  There is a $\mathbb Z$-linear map $Z(\mathcal X_{d_1,\ldots,d_n})\to CH(\mathcal X_{d_1,\ldots,d_n})$ that, to each subvariety $V$ of $\mathcal X_{d_1,\ldots,d_n}$, assigns an element of the Chow ring denoted by $[V]$. If $V$ is a hypersurface of multidegree $(\delta_1,\ldots,\delta_n)$ (i.e.\ cut out by a polynomial of degree $\delta_i$ in the homogeneous coordinates on $\PP^{d_i-1}(\CC)$), then $[V]=\delta_1H_1+\cdots+\delta_n H_n$. 
\end{definition}

Here is the key fact we will need about subvariety representatives.

\begin{fact}[Sufficient criterion for non-empty intersection, and B\'ezout number]\label{f:subvarietyrep}
  If $V_1,\ldots, V_r$ are hypersurfaces in $\mathcal X_{d_1,\ldots,d_n}$ such that $[V_1]\cdots[V_r]$ is non-zero, then $V_1\cap\ldots \cap V_r$  is non-empty. If $[V_1]\cdots[V_r]=0$ then $W_1\cap \ldots \cap W_r=\emptyset$ for almost all hypersurfaces $W_1,\ldots,W_r$ such that $[W_1]=[V_1]$,\ldots, $[W_r]=[V_r]$ (i.e.\ each $W_i$ has the same multidegree as the corresponding $V_i$). If 
\begin{align}\label{eqn:equality}
  [V_1]\cdots [V_r]=N H_1^{d_1-1}H_2^{d_2-1}\cdots H_n^{d_n-1}
\end{align}
for some positive integer $N$, then the generic intersection $W_1\cap \ldots \cap W_r$ consists of $N$ points and $N$ is referred to as the {\it B\'ezout number}.
\end{fact}
\noindent We remark that later in \Cref{def:beznumber}, we will give a more precise definition of the B\'ezout number (needed for stating B\'ezout's Theorem). 
The definition above suffices for our discussion in this section.

\begin{example}
Let $C$, $C'$ be curves in the complex projective plane $\mathcal X_{3}$ of respective degree $\delta,\delta'$. Then $[C]=\delta H$ and  $[C']=\delta' H$, which implies $[C][C']=\delta\delta' H^2$. Hence the two curves will intersect in at least $\delta\delta'$ points. For generic choices of $C, C'$ as above, the two curves will intersect in exactly $\delta\delta'$ points (B\'ezout's Theorem).
\end{example}

\begin{example}
Let $C$, $C'$ be curves in $\mathcal X_{2,2}$ of respective bidegree $(\delta_1,\delta_2)$ and $(\delta'_1,\delta_2')$. Then $[C][C']=(\delta_1\delta_2'+\delta_2\delta_1')H_1H_2$. This could be zero e.g.\ if $\delta_1=\delta_1'=0$, corresponding to the case in which $C$ and $C'$ are both of the form $\bigcup_i (\PP^1(\CC)\times \{p_i\})$ (which do not intersect for generic choices of $p_i$). On the other hand, consider the case $\delta_1=2$ and $\delta_2'=1$. Then
\begin{equation}
C=(\{p_1\}\times \PP^1(\CC)) \cup (\{p_2\}\times \PP^1(\CC))
\end{equation}
and $C'= \bigcup \PP^1(\CC)\times \{p'\}$
for some $p_1,p_2,p'\in \PP^1(\CC)$. Generically, $p_2\neq p_1$ and $|C\cap C'|=|\{(p_1,p'),(p_2,p')\}|=2=\delta_1\delta_2'$. However, in the nongeneric case $p_1=p_2$, we have $|C\cap C'|=\infty$.
\end{example}

\subsubsection{Proof of \Cref{thm:wsdr-product-solution} via the Chow Ring}\label{ssscn:prodsols}

With \Cref{f:subvarietyrep} in hand, we are ready to give our first proof of \Cref{thm:wsdr-product-solution}. 
For this, let $\Pi=\{\Pi_i\}$ be an instance of QSAT on qudits $\ket{\varphi_1},\ldots,\ket{\varphi_n}$ of dimensions $d_1,\ldots,d_n$, respectively. 
Recall that to such an instance $\Pi$, we assign a weighted hypergraph $(G,w)$ as follows. 
We let $V(G)=\{v_1,\ldots,v_n\}$ and define $E(G)=\{e_1,\ldots,e_m\}$ such that $v_i\in e_j$ if and only if the clause $\Pi_j$ acts non-trivially on the qu-$d_i$-it $\ket{\varphi_i}$. 
The weight function $w$ encodes the information regarding the dimension of the qudits, namely $w(v_i)=d_i-1$ for each $i\in\{1,\ldots,n\}$.

\thmWSDRProdsat*
\begin{proof}
Let $V_i$ be the hypersurfaces corresponding to the clauses $\Pi_i$, $i=1,\ldots,m$. Since $V_i$ is of degree $1$ in the variables corresponding to the qubits on which $\Pi_i$ acts non-trivially and of degree $0$ in the remaining ones (see \Cref{eqn:multilinear}), its image in the Chow ring is
\begin{equation}
[V_i]=\sum_{v_j\in E_i} H_j.
\end{equation}
Hence,
\begin{equation}\label{eqn:repsum}
\prod_i [V_i] = \sum_{v_{j_1}\in E_1,\ldots, v_{j_m}\in E_m} H_{j_1}\cdots H_{j_m},
\end{equation}
which is non-zero if and only if there is a summand in which each $H_j$ appears at most $d_j-1$ times i.e.\ if and only if $(G,w)$ has a WSDR.
The claim now follows from \Cref{f:subvarietyrep}.
\end{proof}

\noindent Actually, the proof shows an additional fact, which we will utilize in \Cref{sscn:reductionQubits}:
\begin{corollary}[Counting number of SDRs and product solutions]\label{cor:number-of-solutions}
Let $N$ denote the B\'ezout number. By the proof above of \Cref{thm:wsdr-product-solution}, if \Cref{eqn:equality} holds (i.e. $\prod_i [V_i] = N H_1^{d_1-1}\cdots H_n^{d_n-1}$), then $N$ equals both the number of WSDRs on $(G,w)$, as well as the generic (and minimum, when counted with multiplicity) number of product solutions to any instance of QSAT with underlying weighted hypergraphs $(G,w)$.
\end{corollary}

\begin{observation}
  If in \Cref{thm:wsdr-product-solution}, the number of clauses matches the total degrees of freedom, meaning if $m = \sum_{i=1}^n d_i-1$, then $\prod_i [V_i] = N H_1^{d_1-1}\cdots H_n^{d_n-1}$ for natural number $N$. This is easiest to see with an explicit example, given next.
\end{observation}

\begin{example}
Consider $QSAT$ on 4 qutrits with underlying weighted graph $(G,w)$ with vertices $V(G)=\{1,2,3,4\}$, and edges $E(G)=\{e_1,\dots, e_8\}$ where $e_1=\{1,2,3\}$, $e_2=\{2,3,4\}$, $e_3=\{3,4,1\}$, $e_4=\{4,1,2\}$, $e_5=e_6=e_7=e_8=\{1,2,3,4\}$.
In this case, $m = \sum_{i=1}^n d_i-1$, and \Cref{eqn:equality} holds, since
\begin{align}
  &(H_1+H_2+H_3)(H_2+H_3+H_4)(H_3+H_4+H_1)(H_4+H_1+H_2)(H_1+H_2+H_3+H_4)^4\\
  &=864H_1^2H_2^2H_3^2H_4^2\,.
\end{align}
To see this without any calculation, pick from each bracketed term a single term $H_i$. Any non-zero summand in \Cref{eqn:repsum} must have picked any $H_i$ at most $d_i-1=2$ times. But since $m = \sum_{i=1}^n d_i-1$, each $H_i$ must be picked at least $d_i-1$ times to ensure all edges are covered. Thus, \Cref{eqn:equality} holds. We conclude that \emph{every} instance of QSAT with interaction graph $(G,w)$ has at least $864$ product solutions (counted with multiplicity) and almost all such instances have exactly $864$ product solutions. Moreover, $(G,w)$ has exactly $864$ WSDRs.
\end{example}

\begin{example}
  If every qudit of dimension $d_i$ occurs in at most $d_i-1$ constraints, then there exists a product solution.
  The WSDR exists trivially because it is impossible to assign a qudit to more than $d_i-1$ constraints.
  To compute a product solution, iterate through the qudits in arbitrary order, keeping track of reduced constraints.
  We can assign each qudit $i$ to a value in the common nullspace of the $\le d_i-1$ (reduced) $1$-local constraints on qudit $i$.
\end{example}

\subsection{Approach 2: Reduction to qubits}\label{sscn:reductionQubits}

We next give a completely different proof of \Cref{thm:wsdr-product-solution}, this time via direct reduction from a Hamiltonian with a weighted SDR on qudits to a Hamiltonian with an SDR on qubits (and subsequently using \cite{laumannc.r.PhaseTransitionsRandom2010}).
The result follows from the main theorem of this section, \Cref{thm:qubit-reduction}, through which a qubit Hamiltonian can be constructed by iteratively replacing a $(d+1)$-qudit by a qubit and a $d$-qudit, while preserving the existence of a WSDR. 
This second proof approach will also prove important later for our second main result on TFNP in \Cref{sscn:MHS}.

\begin{theorem}\label{thm:qubit-reduction}
   Let $\Pi$ be a QSAT instance on a Hilbert space $\cH = \CC^{d+1}\otimes \bigotimes_{i=2}^n \CC^{d_i}$ whose underlying weighted hypergraph $(G,w)$ has a WSDR.
   There exists a linear-time constructible QSAT instance $\Pi'$ on Hilbert space $\cH' = \CC^2\otimes \CC^{d}\otimes\bigotimes_{i=2}^n \CC^{d_i}$ whose underlying weighted hypergraph $(G',w')$ also has a WSDR.
   Given a product state solution to $\Pi'$ $(\Pi)$, we can compute a product solution to $\Pi$ $(\Pi')$ in polynomial time.
\end{theorem}
\begin{proof}
  Let $z$ denote the first qudit in $\Pi$ of dimension $d+1$. 
  To replace $z$ by a qubit $x$ and a qu$d$it $y$, we will define and use a mapping $f:\PP^1\times \PP^{d-1} \to \PP^{d}$,
  \begin{equation}\label{eq:qudit-map}
    f(x,y) \coloneqq \begin{pmatrix}
      x_1y_1\\
      x_2y_d\\
      x_1y_2 - x_2y_1\\
      x_1y_3 - x_2y_2\\
      \vdots\\
      x_1y_d - x_2y_{d-1}\\
    \end{pmatrix}.
  \end{equation}
  Via \Cref{lem:qudit-map}, we will then be able to argue that $f$ allows us to create $\Pi'$ which is satisfiable by a product state if and only if $\Pi$ is.
  
  To begin, let $\Pi_i$ be a constraint of $\Pi$ with associated hyperedge $e_i = \{z,v_2,\dots,v_k\}$.
  We can view $\Pi_i$ as a multilinear polynomial $p$ whose monomials are the entries of  $\ket{z}\otimes\ket{v_2}\dotsm\ket{v_k}$ (taking $z,v_2,\dots,v_k$ as symbolic vectors).
  The corresponding constraint in $\Pi_i'$ with hyperedge $e_i' = \{x,y,v_2,\dots,v_k\}$ is obtained by replacing every occurrence of $z_j$ in $p$ with $f(x,y)_j$.
  $\Pi_i'$ is a valid constraint since its monomials are the entries of $\ket{x}\otimes\ket{y}\otimes\ket{v_2}\dotsm\ket{v_k}$ (see Example~\ref{ex:qutrit-qubit}).
  For constraints $\Pi_i$ not acting on $z$, let $\Pi_i' = \Pi_i$.

  What remains to show is the correspondence between product solutions to $\Pi$ and $\Pi'$ as well as the existence of a WSDR.
  The latter is straightforward, taking the $d$ edges assigned to $z$ in $\Pi$ and assigning one of them to $x$ and the remaining $d-1$ edges to $y$.
  To construct a product solution for $\Pi$ from $\Pi'$, just set $z = f(x,y)$, which is non-zero by Lemma~\ref{lem:qudit-map}.
  For the other direction, assign a preimage of $z$ to $(x,y)$, which again is efficiently computable by Lemma~\ref{lem:qudit-map}.
\end{proof}

\begin{example}\label{ex:qutrit-qubit}
  To illustrate Theorem~\ref{thm:qubit-reduction}, let $\bra{\phi}\in\CC^6$ be a constraint on a qutrit $z$ and a qubit $v$.
  A product state $\ket{z}\otimes \ket{v}$ satisfies this constraint if $p(z,v) = \sum_{i=1}^3\sum_{j=1}^2 \phi_{ij}z_iv_j=0$.
  The construction of Theorem~\ref{thm:qubit-reduction} replaces the qutrit $z$ with two qubits $x,y$.
  The new constraint $\bra{\phi'}$ is defined via the polynomial 
  \begin{align}
    p'(x,y,v) = \sum_{i=1}^3\sum_{j=1}^2 \phi_{ij}f(x,y)_iv_j = \sum_{j=1}^2 (\phi_{1j}x_1y_1 + \phi_{2j}x_2y_2 + \phi_{3j}(x_1y_2-x_2y_1))v_j,
  \end{align} 
  giving $\bra{\phi'} = (\phi_{11},\phi_{12},\phi_{31},\phi_{32},-\phi_{31},-\phi_{32},\phi_{21},\phi_{22})$ (where monomomials $x_iy_j z_k$ are listed in increasing binary order with respect to $ijk\in\set{0,1}^3$).
\end{example}

\begin{lemma}\label{lem:qudit-map}
  The map $f$ given in~\eqref{eq:qudit-map} is well-defined (i.e. $f(x,y)\neq 0$ if $x\neq0,y\neq 0$), and surjective with polynomial-time computable preimage.
\end{lemma}
\begin{proof}
  To show $f$ is well-defined, let $x\in\PP^1,y\in\PP^{d-1}$, i.e., $x\ne0,y\ne0$.
  Consider cases:
  \begin{enumerate}[label=(\roman*)]
    \item ($x_1 = 0$) Then $x_2 \ne 0$. There exists $i$ with $y_i \ne 0$. If $i=d$, then $x_2y_d\ne0$. Otherwise $x_1y_{i+1} - x_2y_i \ne 0$.
    \item ($x_1\ne 0$) Let $i$ be minimal such that $y_i \ne 0$. If $i=1$, then $x_1y_1\ne 0$.
    Otherwise, $x_1y_i - x_2y_{i-1} = x_1y_i \ne 0$.
  \end{enumerate}
  Hence, $f(x,y)\ne0$ and therefore well-defined. 
  
  To next show $f$ is surjective, consider any $z\in\PP^{d}$. We compute $x\in\PP^1,y\in\PP^{d-1}$ such that $f(x,y)=z$ via cases:
  \begin{enumerate}[label=(\roman*)]
    \item ($z_1=0$) Set $x_1=0$ and $x_2 =1$, satisfying the equation $x_1y_1=z_1$. The remaining equations are $y_d = z_2, y_1=-z_3, y_2=-z_4,\dots,y_{d-1}=-z_{d+1}$.
    Since $z_1=0$, there exists an $i$ with $y_i\ne0$.
    \item ($z_1\ne0$) Without loss of generality, assume $z_1=1$.
    Set $x_1=1, y_1=1$ to satisfy the first equation and ensure $x\ne0,y\ne0$.
    Substituting $y_1=1,x_1=1$, the remaining equations are:
    \begin{subequations}
      \begin{align}
         x_2y_d &= z_2\label{eq:subst:x2}\\
         y_2 &= z_3 + x_2\label{eq:subst:y2}\\
         y_3 &= z_4 + x_2y_2\label{eq:subst:y3}\\
         &\;\;\vdots\nonumber\\
         y_d &= z_{d+1} + x_2y_{d-1}\label{eq:subst:yd}
      \end{align}
    \end{subequations}
    Combining Equations~\eqref{eq:subst:y2} to~\eqref{eq:subst:yd}, we have $y_d = x_2^{d-1}+\sum_{i=3}^{d+1}z_i x_2^{d+1-i}$.
    Substituting $y_d$ in Equation~\eqref{eq:subst:x2}, we have $x_2^{d}+\sum_{i=3}^{d+1}z_i x_2^{d+2-i} = z_2$, which is a polynomial with solution in $x_2$.
    Finally, set $y_2,\dots,y_d$ according to Equations~\eqref{eq:subst:y2} to~\eqref{eq:subst:yd}, step by step.
    \qedhere
  \end{enumerate}
\end{proof}

\begin{rem}
  By \Cref{cor:number-of-solutions}, each application of Theorem~\ref{thm:qubit-reduction} increases the number of product solutions by a factor of $d$ (counted with multiplicity).
  This matches the intuition from Lemma~\ref{lem:qudit-map}, where computing the preimage of $f$ requires solving a polynomial of degree $d$. 
\end{rem}

\begin{rem}[Relation to the Segre embedding]
  The map $f:\PP^1\times\PP^{d-1}\to\PP^d$ is a linear map from the Segre embedding of $\PP^1\times\PP^{d-1}$ to $\PP^d$, i.e. $f(x,y) = L\sigma(x,y)$ for some linear map $L$.
\end{rem}

\subsection{Application: Maximal dimension of a completely entangled subspace}\label{sscn:maxdim}

Finally, we demonstrate the applicability of the WSDR framework beyond the setting of QSAT. Specifically, Parthasarathy~\cite{Par04} studies the notion of a completely entangled subspace and gives its maximal dimension.
We can recover this result as a corollary of Theorem~\ref{thm:wsdr-product-solution}.

\begin{definition}[\cite{Par04}]\label{def:CE}
  Let $\cH_1,\dots,\cH_k$ be complex Hilbert spaces of dimension $d_i$ and $\cH = \bigotimes_{i=1}^k \cH_i$.
  A subspace $S\subseteq \cH$ is said to be \emph{completely entangled} if $\ket{\psi_1}\otimes\dotsm\otimes\ket{\psi_k}\notin S$ for any non-zero product vector with $\ket{\psi_i}\in\cH_i$.
\end{definition}

\begin{corollary}[\emph{c.f.} \cite{Par04}]\label{cor:CES}
  The maximal dimension of a completely entangled subspace is $\prod_{i=1}^k d_i - \sum_{i=1}^k d_i + k - 1$.
\end{corollary}
\begin{proof}
  Let $D = \dim(\cH) = \prod_{i=1}^k d_i$.
  Let $S\subset \cH$ be a subspace of dimension $d_S$ and let $\Pi_{S^\bot} = \sum_{i=1}^{D-d_S}\ketbra{\psi_i}{\psi_i}$ be a spectral decomposition of the projector onto the orthogonal complement of $S$.
  If $D-d_S \le \sum_{i=1}^k d_i - k$, then $\Pi_{S^\bot}$ has a WSDR, treating space $\cH_i$ as a qudit of dimension $d_i$.
  Hence, if $d_S \ge \prod_{i=1}^k d_i - \sum_{i=1}^k d_i + k$, $S$ must contain a product state by Theorem~\ref{thm:wsdr-product-solution}.
  Equivalently, if $S$ is completely entangled, $d_S \le \prod_{i=1}^k d_i - \sum_{i=1}^k d_i + k - 1$.
  This bound is tight because generic instances without WSDR have no product solution.
\end{proof}


\section{Low-degree, multi-homogeneous systems and TFNP}\label{scn:MHS}

We next study low-degree, multi-homogeneous polynomial systems. 
\Cref{sscn:bez} first defines multihomogeneous polynomial systems, and states the multihomogeneous \Bez\ Theorem. 
\Cref{sscn:MHS} then defines our first new TFNP subclass, \MHS, and shows \MHS-completeness of QSAT with SDR. 
The latter uses the WSDR techniques of \Cref{sscn:reductionQubits}.

\subsection{Definitions and \Bez's Theorem}\label{sscn:bez}

We begin with a formal definition of a multi-homogeneous polynomial. (For clarity, recall we consider polynomials over $\CC$ in this work.)

\begin{definition}[Multi-homogeneous polynomial \cite{MS87}]\label{def:multipoly}
  A polynomial $f$ is multi-homogeneous if there are $m$ sets of variables $Z_j = \{z_{0,j},\dots,z_{n_j,j}\}$ and $d_1,\dots,d_m\in\ZZ_{\ge0}$ with at least one $d_j>0$ such that
  \begin{equation}\label{eq:multihom}
    f = \sum_{\stackrel{I_1,\dots,I_m:}{\forall j \;\;\abs{I_j}=d_j}} a_{I_1,\dots,I_m} Z_{1}^{I_1}\dotsm Z_{m}^{I_m},
  \end{equation}
  where $I_j=(i_{0,j},\dots,i_{n_j,j})\in\ZZ_{\ge0}^{n_j+1}$, $\abs{I_j}\coloneqq\sum_{k=0}^{n_j}i_{k,j}=d_j$, $Z_j^{I_j} = z_{0,j}^{i_{0,j}}\dotsm z_{n_j,j}^{i_{n_j,j}}$, and coefficients $a_{I_1,\dots,I_m}\in\CC$.
\end{definition}
\noindent Let us repeat this in words, and subsequently give it context relative to QSAT. Above, each variable set $Z_j$ has $n_j+1$ variables. 
Each $Z_j^{I_j}$ term is a product of some subset of $d_j$ variables from $Z_j$, with the precise choice of variables given by index subset $I_j$. 
Thus, $d_j$ can be thought of as the \emph{degree} of the polynomial relative to variables $Z_j$.

\begin{example}
  A simple example of a multi-homogeneous polynomial is $x_1y_1y_2 + x_2y_2y_3$, where $Z_1=\set{x_1,x_2}$, $Z_2=\set{y_1,y_2,y_3}$, $d_1=1$, and $d_2=2$. 
\end{example}

Let us return to product-state solutions for QSAT (i.e. \PS).
Why is \emph{multi}-homogeneous the right formulation? 
When each monomial of $f$ in \Cref{def:multipoly} contains at most one variable from each $Z_j$ (i.e. $d_j\in\set{0,1}$ for all $j\in[m]$), the equivalence between \Cref{eq:multihom} and a QSAT constraint is straightforward.
Each set of variables $Z_j$ corresponds to the $n_j+1$ amplitudes of the $j$th qu$d$it of our system with $d=n_j+1$. 
Thus, the number $m$ of subsets $Z_j$ is the number of qudits in our system.
Any projective constraint $\ket{\psi}$ acting on subset of qudits $S\subseteq [m]$ is now equivalent to a polynomial $f$ with $d_j=1$ for $j\in S$ and $d_j=0$ for $j\in[m]\setminus S$.
As for the more general case where $f$ has $d_j>1$ for some $Z_j$, when higher degree terms in the variable sets are permitted, the reduction from a multi-homogeneous system back to QSAT is non-trivial, and given shortly in \Cref{thm:MHScomplete}. 

\paragraph{\Bez's theorem.} We now state the mathematical principle on which our TFNP subclass rests, \Bez's theorem. For this, we first define the \Bez\ number.
Below, the terms $n_i$ are from \Cref{def:multipoly}.

\begin{definition}[\Bez\ number \cite{MS87}]\label{def:beznumber}
  Let $F = \{f_1,\dots,f_n\}$ be a system of $n=n_1+\dotsm+ n_m$ multi-homogeneous polynomials with degrees $\{\dij\mid i\in[n], j\in[m]\}$.
  The \emph{Bézout number} $\bez$ of $F$ is defined as the coefficient of $\prod_{j=1}^m \alpha_j^{n_j}$ in $\prod_{i=1}^n\sum_{j=1}^m \dij\alpha_j$, where $\alpha_1,\dots,\alpha_m$ are symbolic variables representing the $m$ variable sets.
\end{definition}
\begin{rem}\label{rem:computebez}
  Computing $\bez$ in general is difficult~\cite{malajovichComputingMinimalMultihomogeneous2005}. Checking if $\bez$ is non-zero, however, is tractable, which suffices for our purposes.
\end{rem}

\noindent For clarity, as in \Cref{def:multipoly} the system $F$ is defined over variable subsets $Z_j$, each of size $n_j+1$. For each polynomial $f_i$, $\dij$ is now the degree of $f_i$ relative to variable set $Z_j$.

\begin{example}\label{ex:bezout}
  Let $F=(f_1,f_2,f_3)$ with 
  \begin{subequations}
  \begin{align}
      f_1&=x_1y_1y_2 + x_2y_2y_3  \quad &d_{1,1}=1\qquad\qquad &d_{2,1}=2\\
      f_2&=x_1y_1+x_2y_2  &d_{1,2}=1 \qquad\qquad&d_{2,2}=1\\
      f_3&=y_1y_2+y_2y_3  &d_{1,3}=0 \qquad\qquad&d_{2,3}=2,      
  \end{align} 
  \end{subequations}
  where $Z_1=\set{x_1,x_2}$, $Z_2=\set{y_1,y_2,y_3}$, $n_1=1$, $n_2=2$, $m=2$. Then, 
  \begin{equation}
    \prod_{i=1}^3\sum_{j=1}^2 \dij\alpha_j= (\alpha_1+2\alpha_2)(\alpha_1+\alpha_2)(2\alpha_2)=2\alpha_1^2\alpha_2+6\alpha_1\alpha_2^2+4\alpha_2^3.
  \end{equation}
  The coefficient of $\alpha_1\alpha_2^2$, and thus the \Bez\ number, is $\bez=6$.
\end{example}

\begin{observation}[Number of weighted SDRs equals \Bez\ number]\label{obs:countWSDR}
  The number of weighted SDRs in a \PS\ instance is equal to the Bézout number of the corresponding multi-homogeneous system. (For clarity, by definition of the \Bez\ number (\Cref{def:beznumber}), we mean for the case of $n=n_1+\dotsm+ n_m$.) To see this, observe that in $\prod_{i=1}^n\sum_{j=1}^m \dij\alpha_j$, the product is over all $n$ equations, and for each equation $f_i$, the \Bez\ number corresponds to choosing from the inner sum (which represents variable groups) a single variable from a single variable group $Z_j$, such that this variable appears in $f_i$ (i.e. $\dij>0$). The coefficient of $\prod_{j=1}^m \alpha_j^{n_j}$ then counts the number of ways we can ``cover'' all $f_i$ in this manner using using variables from each group $Z_j$ precisely $n_j$ times. The claim follows by observing that in the corresponding QSAT instance, any single such covering is equivalent to a single weighted SDR.
\end{observation}

With the \Bez\ number $\bez$\ in hand, we state \Bez's theorem, which gives a sufficient condition for a multi-homegenous system having a solution.

\begin{theorem}[Bézout's Theorem \cite{MS87,Shafarevich1974}]\label{thm:bezout}
  A multi-homogeneous system $F(Z)=0$ has no more than $\bez$ geometrically isolated solutions in $\PP^{n_1}(\CC)\times \cdots\times \PP^{n_m}(\CC)$.
  If $F(Z) = 0$ does not have an infinite number of solutions in $\PP^{n_1}(\CC)\times \cdots\times \PP^{n_m}(\CC)$, then it has exactly $\bez$ solutions, counting multiplicities.
\end{theorem}
\noindent Applied to \Cref{ex:bezout}, this tells us that either the number of solutions to $F=(f_1,f_2,f_3)$ is infinite, or there are exactly $\bez=6$ solutions. Thus, if the \Bez\ number is positive, there is a solution.

\subsection{The class MHS and completeness results}\label{sscn:MHS}

Since a positive \Bez\ number implies the existence of a solution, and finding an approximate solution is clearly in TFNP, we now define a new subclass of TFNP to capture this, MHS.

\begin{definition}[(Low-Degree) Multi-homogeneous Systems (MHS)]\label{def:MHS}
  Define $\MHS_{s,d}$ as the set of TFNP relations $R(x,y)$ poly-time reducible (as defined in~\cite{Pap94}) to finding an $\epsilon$-approximate solution to a system $F = \{f_1,\dots,f_n\}$ of $n$ multi-homogeneous equations, where
  \begin{enumerate} 
    \item (a solution exists) $\bez>0$, 
    \item (at most $s$ variables per variable group $Z_j$) for all $j\in[m]$, $n_j\le s$, 
    \item (each equation $f_i$ is of total degree at most $d$) for all $i\in [n]$, $\sum_{j=1}^{m}\dij\le d$, and
    \item $\epsilon \in \Omega(2^{-\poly(n)})$.
  \end{enumerate}
  For clarity, $\epsilon$ and $n$ are inputs and thus may depend on $\abs{x}$, whereas $s$ and $d$ are parameters and considered constants independent of $\abs{x}$.
  More formally, there exist $\poly(\abs{x})$-time computable functions $g$ and $h$, such that $g(x)$ outputs $\epsilon$ and a description of a multi-homogeneous system $F$, and $R(x,h(x,Y))$ holds, where $Y$ is an approximate solution to $F(Y)=0$ with $\sum_{k=1}^n\abs{f_k(Y)}\le \epsilon$, assuming each equation $f_i$ and variable group $Z_j$ is normalized in the Euclidean norm\footnote{This is to prevent trivial solutions such as setting all variables to approximately $0$. Formally, we mean the coefficient vector of each $f_i$ is normalized with respect to the Euclidean norm, and likewise for each variable group $Z_j$, the corresponding assignment vector. For example, to normalize $f=x_1y_1y_2+x_2y_2y_3$, the right hand side is multiplied by $(\enorm{f}\enorm{x}\enorm{y})^{-1}$, for $f=(1,1)$, $x=(x_1,x_2)$, and $y=(y_1,y_2)$.}. Finally, define 
  \begin{equation}\label{eqn:MHScap}
      \MHS := \bigcup_{s,d\in \Theta(1)} \MHS_{s,d}. 
  \end{equation}
\end{definition}
\noindent In words, \Cref{eqn:MHScap} says \MHS\ requires constant bounds on the variable set sizes $s$ and total degree $d$ per equation (i.e. the number of variables in each monomial), and allows up to inverse exponential precision additive error $\epsilon$. 

As remarked in \Cref{scn:intro}, the following observation follows straightforwardly since poly-time Turing machines can efficiently perform basic arithmetic with polynomial bits of precision, and since the degrees and set sizes in \MHS\ are constant.
\begin{observation}\label{obs:MHSinTFNP}
  $\MHS\subseteq\TFNP$.
\end{observation}
We now show that \PS\ captures the complexity of MHS.

\begin{figure}[t]
  \begin{center}
    \includegraphics[width=1\textwidth]{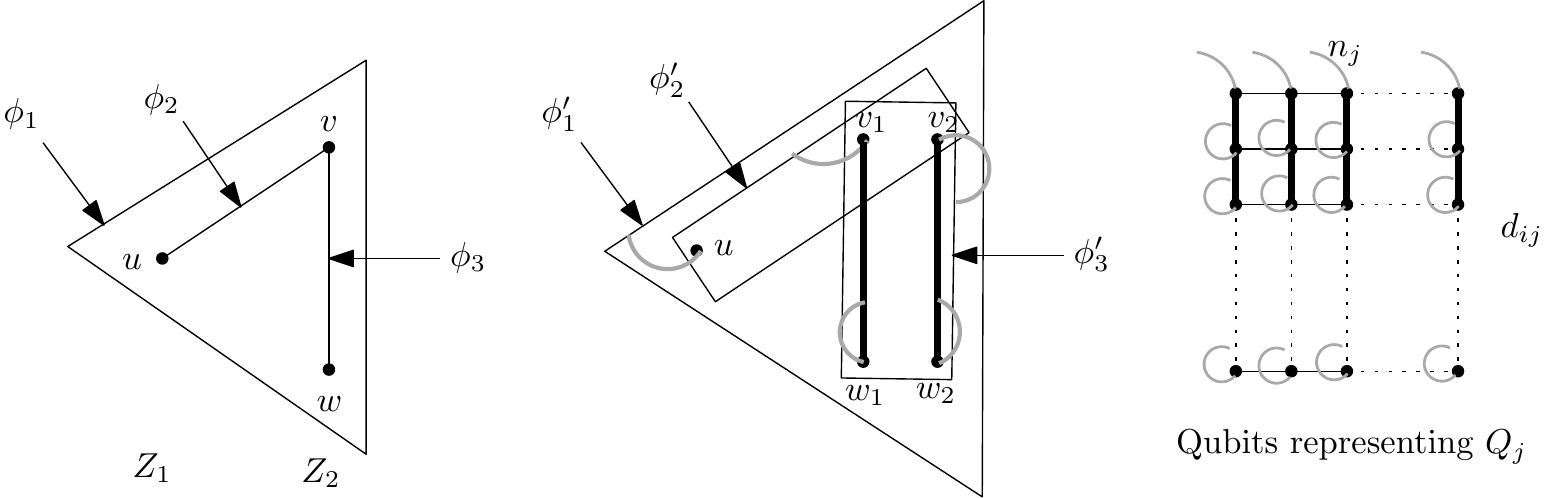}
  \end{center}
\caption{(Left) The reduction of \Cref{thm:MHScomplete} before the reduction to qubits and without equality constraints, as illustrated on \Cref{ex:bezout}. The latter has equations $f_1=x_1y_1y_2 + x_2y_2y_3$, $f_2=x_1y_1+x_2y_2$, and $f_3=y_1y_2+y_2y_3$ with variable sets $Z_1=\set{x_1,x_2}$ and $Z_2=\set{y_1,y_2,y_3}$, $n_1=1$, $n_2=2$, $m=2$, and $d=3$. Variable sets $Z_1$ and $Z_2$ are represented by vertex sets $\set{u}$ and $\set{v,w}$, respectively. (For simplicity, the reduction actually creates $d=3$ vertices for each $Z_i$, in order to be able to accommodate monomials of degree $3$ in each $Z_i$. However, the system $f_1,f_2,f_3$ is at most linear in $Z_1$ and quadratic in $Z_2$, so $3$ vertices per $Z_i$ is overkill; we thus depict only the vertices needed to encode $f_1,f_2,f_3$.) Vertices $u,v,w$ correspond to $\ket{\psi_{1,1}}\in\CC^2$ and $\ket{\psi_{1,2}},\ket{\psi_{2,2}}\in\CC^3$, respectively. The joint product state assignment thus has form $\ket{\psi_{1,1}}\ket{\psi_{1,2}}\ket{\psi_{2,2}}=\sum_{i=0}^1\sum_{j,k=0}^2 \alpha_i\alpha_j\alpha_k\ket{i}\ket{j}\ket{k}\in\CC^2\otimes(\CC^3)^{\otimes 2}$. Each constraint $f_i$ is encoded into a rank-$1$ projector onto $\ket{\phi_i}$. Specifically, $\ket{\phi_1}=\ket{001}+\ket{112}$ (acting on all three systems), $\ket{\phi_2}=\ket{00}+\ket{11}$ (acting on the first two systems), and $\ket{\phi_3}=\ket{01}+\ket{12}$ (acting on the last two systems). (Middle) The figure on the left after the reduction to qubits is applied, followed by addition of equality constraints via $2$-local projectors onto the antisymmetric subspace. Here, $v,w\in\CC^3$ have been mapped to $v_1,v_2\in\CC^2$ and $w_1,w_2\in\CC^2$, respectively. Edge $\set{u,v}$ is now a hyperedge $\set{u,v_1,v_2}$. Thick block edges represent equality constraints. Thinner gray edges represent the SDR, i.e. which qubit is matched to which hyperedge. (Right) A ``close-up'' of all qubits representing $Q_j$ when the full reduction is applied to a general multihomogeneous system. Thick black edges represent equality constraints. Thinner gray edges represent the SDR. The first row, labelled $q_{i,1}$ through $q_{i,n_j}$ in the proof, are matched with the $n_j$ hyperedges incident on $Q_j$ corresponding to the original equations $f_i$ (hyperedges not depicted). Vertices in rows $i$ with $i>1$ are matched with their incident edge to row $i-1$.
}
\label{fig:reduction}
\end{figure}
\begin{theorem}\label{thm:MHScomplete}
  Let $M$ denote input size, and consider any $\epsilon\in\Omega(2^{-\poly(M)})$.
  \begin{enumerate}
    \item (Containment in \MHS) For any local dimension $d\in O(1)$ and locality $k\in O(1)$, $\epsilon$-approximate \PS\ with WSDR for $k$-local constraints on qudits of dimension $d$ is in $\MHS_{d-1,k}(\epsilon)$.
    \item (\MHS-hardness) $\MHS_{s,d}(\epsilon)$ is poly-time reducible to $\Theta(\epsilon)$-approximate \PS\ on qubits (i.e. local dimension $2$) with an SDR and locality $k\geq (s+1)^d$.
\end{enumerate} 
\end{theorem}
\begin{rem}
    Recall from \Cref{scn:intro} that our result does not specify a single $k$ for which $\MHS$-hardness is obtained \emph{for all} $s,d\in O(1)$. Regarding this, a blowup in $k$ is perhaps expected, since in the $k=2$ case (i.e. $2$-QSAT on qubits), producing a satisfying assignment is well-known to be efficiently solvable, even without an SDR (assuming a satisfying assignment exists)~\cite{bravyiEfficientAlgorithmQuantum2006}. It is, however, plausible that \Cref{thm:MHScomplete} can be extended in the $k=2$ case on qudits for some local dimension $d>2$, since $2$-QSAT on qudits remains $\QMAo$-complete~\cite{eldarQuantumSATQutritCinquit2008,nagajLocalHamiltoniansQuantum2008,rudolphQuantum2SATLow2024}. As for the bound $k\geq (s+1)^d$ in \Cref{thm:MHScomplete}, in the simplest non-trivial case of quadratic equations on variable sets $Z_i$ of size $2$ each (i.e. $s=1$), this bound yields $k=4$.
\end{rem}
\begin{proof}[Proof of \Cref{thm:MHScomplete}]
  For containment in $\MHS$, as argued above, any \PS\ system with SDR can be represented as a system of multi-homogeneous equations.
  For simplicity, we consider the case of qubits; the qudit case is analogous.
  Without loss of generality, we may assume there are $m$ qubits and $n=m$ clauses, since if $n<m$ an SDR cannot exist, and if $n>m$ we can add trivially satisfied constraints to the system.
  An equation $f_i$ corresponding to a $k$-local constraint is multilinear in $k$ variable groups, so we get $\sum_{j=1}^m \dij=k$ for all equations $f_i\in[n]$.
  Since the \PS\ system only contains qubits, we have $n_j=1$ for all $j\in[m]$ and thus $s=1$.
  By \Cref{obs:countWSDR}, the Bézout number equals the number of SDRs, which is at least one.
  Finally, by construction and the definition of $\MHSe$, cumulatively satisfying all $f_i$ within total additive $\epsilon$ precision immediately yields a PRODSAT solution with $\epsilon$ precision.

  For \MHS-hardness, consider a multi-homogeneous system $F=\{f_1,\dots,f_n\}$ with variable sets $Z_1,\dots,Z_m$, $\sum_{j=1}^m \dij\le d$ for all equations $i\in[n]$, $n_j\le s$ for all variable sets $j\in[m]$, and $\bez>0$.
  First, we embed $F$ into a qudit system.
  Each variable group $Z_j$ has, by definition, $n_j+1$ variables, and so each assignment to these variables can be represented by an $(n_j+1)$-dimensional state $\ket{\psi_j}$. 
  However, $F$ need not be multi-linear, meaning monomials in equation $f_i$ each contain exactly $\dij$ variables (counting multiplicity) from $Z_j$. 
  To simulate this non-linearity, we instead create $\cj:=\max_{i\in[n]}\dij$ states in our system, $\ket{\psi_{1,j}},\dots,\ket{\psi_{\cj,j}}$, each again of dimension $n_j+1$.
  Let $Q_j$ denote the set of qudits created by this mapping for $Z_j$, and consider any $f_i$ acting on some set of variable sets $A_i\subseteq\set{Z_1,\ldots, Z_m}$.
  Since $f_i$ has degree $\dij$ in variable set $Z_j$, we will construct our corresponding clause $\ket{\phi_i}$ to act without loss of generality on the first $\dij$ qudits in $Q_j$. (Assume the qudits in $Q_j$ have an arbitrary, fixed order.)
  Under this mapping, let $B_i\subseteq Q_1\cup\cdots\cup Q_m$ denote the corresponding set of qudits to be acted on by $\ket{\phi_i}$.
  To now design $\ket{\phi_i}$, ideally for any $j\in[m]$, we would like all qudits in $Q_j$ to have identical local assignments, i.e. $\ket{\psi_{1,j}}=\dotsm=\ket{\psi_{\dij,j}}$.
  In such a case, we can represent the multi-homogeneous polynomial $f_i$ by a projective rank-$1$ constraint $\ket{\phi_i}$ acting on $B_i$, since the amplitudes (with respect to the computational basis) of $\bigotimes_{j=1}^{m}\bigotimes_{i=1}^{\dij}\ket{\psi_{i,j}}$ are in one-to-one correspondence with all possible monomials of $f_i$, as given by~\Cref{eq:multihom}.
  \Cref{fig:reduction} illustrates the construction thus far. 
  
  \emph{Enforcing equality.} To indeed enforce equality among all qudits in $Q_j$, since we are considering product state assignments, it suffices to place $2$-local projectors onto the antisymmetric subspace for each consecutive pair of qudits in $Q_j$.
  Unfortunately, this would add too many constraints when our qu\emph{d}its have local dimension $d>2$, so that a WSDR cannot exist.
  To see this, assume the worst case scenario in which $\cj=\dij$ for all $i\in[n]$, i.e. each variable group $Z_j$ has the same degree in all equations.
  Now, by \Cref{obs:countWSDR}, each variable set $Z_j$ must ``cover'' $n_j$ equations $f_i$, and so in principle each $Q_j$ must also cover these same $n_j$ equations. 
  Recalling we have $n=\sum_{j=1}^m n_j$ equations, observe that a WSDR on our qudits can cover at most 
  \begin{align}
  \label{eqn:WSDRmax}
    \sum_{j=1}^m \cj n_j
  \end{align}
  clauses in our construction. (Each $Q_j$ has $\cj$ qudits, each of dimension $n_j+1$, meaning each qudit in $Q_j$ affords a WSDR $n_j$ degrees of freedom.) 
  Since $Q_j$ must cover $n_j$ of the equations $f_i$, in order for a WSDR to exist, it is necessary for our construction to implement \emph{all} equality constraints for $Z_j$ using at most $n_j(\cj-1)$ rank-$1$ projectors.
  At least $\cj-1$ $2$-local constraints are necessary to ensure equality among $\cj$ qudits, implying each equality constraint must have rank at most $n_j$.
  Unfortunately for $d>2$, the antisymmetric subspace on two qudits of dimension $n_j+1$ has dimension $(n_j+1)^2-\binom{n_j+2}{2}>n_j$ for $n_j>1$~\cite{watrous_2018}.
  In fact, \emph{no} projector of rank $n_j$ can enforce equality between qudits of dimension $n_j+1$ (\Cref{obs:low-rank-swap-test}).
  
  To overcome this obstacle, we instead apply the reduction to qubits from Theorem~\ref{thm:qubit-reduction}, and then use the projectors onto the antisymmetric subspace to force the equality among the resulting qubits (\Cref{fig:reduction}, middle).
  Specifically, consider any $Q_j$ consisting of $\dij$ qudits of dimension $n_j+1$.
  Label these qudits $q_1,\ldots, q_{\dij}$.
  \Cref{thm:qubit-reduction} replaces each $q_i$ with $n_j$ qubits which we label here as $q_{i,1},\ldots,q_{i,n_j}$, such that any hyperedge acting on $q_i$ now acts instead on $q_{i,1},\ldots,q_{i,n_j}$.
  To simulate equality between the qudits $q_i$, by the construction of \Cref{thm:qubit-reduction}, it now suffices to place projectors onto the singlet state $\ket{01}-\ket{10}$ between $q_{i,k}$ and $q_{i+1,k}$ for all $i\in\set{1,\ldots, c_j-1}$ and $k\in\set{n_j}$ (thick vertical edges in \Cref{fig:reduction}, middle).
  This yields $\sum_{j=1}^m (\dij-1) n_j$ equality constraints for $Q_j$.

  \emph{The SDR.} It remains to show that the resulting QSAT instance on qubits has an SDR. 
  The argument is similar to the discussion surrounding \Cref{eqn:WSDRmax}, i.e. we have $\sum_{j=1}^m \dij n_j$ degrees of freedom which which to cover all clauses, where each degree of freedom corresponds to a unique qubit in our system.
  Note \Cref{thm:qubit-reduction} does not alter the number of hyperedges; thus, our system has precisely $n=\sum_{i=1}^m n_i$ clauses corresponding to $\set{f_i}_{i=1}^n$ to cover. 
  Now, since $\bez>0$ for $\set{f_i}_{i=1}^n$, and since we assumed each clause $\ket{\phi_i}$ acts without loss of generality on the first $\dij$ qudits in $Q_j$, by \Cref{obs:countWSDR} we may use the set of $n_j$ qubits in $Z_j$ which replaced the first qudit, $q_1$, in our use of \Cref{thm:qubit-reduction} to cover all clauses acting on $Q_j$. 
  The remaining qubits $q_{i,k}$ for $i>1$ can now be straightforwardly used to cover all $\sum_{j=1}^m (\dij-1) n_j$ equality constraints (\Cref{fig:reduction}, right).

  \emph{Precision.} That an $\epsilon$-approximate solution for the \PS\ instance suffices to produce an $\Theta(\epsilon)$-approximate solution for \MHS\ follows by the Lipschitz continuity of polynomials on a compact set and the fact that degrees and group sizes are bounded by $O(1)$.
\end{proof}

\begin{rem}
  $\MHS$-hardness in \Cref{thm:MHScomplete} is stated in terms of qubits; however, the statement holds for any constant local dimension $d$. (The case of $d=2$ simply yields the strongest hardness result.) Specifically, hardness can be shown by embedding each qubit output by our reduction into a qu\emph{d}it and adding projector onto $\Span(\ket{2},\ldots, \ket{d-1})$ onto each qudit.
\end{rem}


Finally, in the proof of \Cref{thm:MHScomplete}, we claimed no low rank projector could test for equality --- this follows by the definition of the antisymmetric subspace, but we include an explicit proof below for completeness.
\begin{observation}\label{obs:low-rank-swap-test}
  For $d>2$, there exists no projector $\Pi\in\CC^{d^2\times d^2}$ of rank $\le d-1$ such that for all $\ket\psi,\ket\phi\in\CC^d$, $\Pi\ket{\psi}\ket{\phi} = 0$ iff $\ket{\psi}\propto\ket{\phi}$.
\end{observation}
\begin{proof}
  Assume there exists such a projector $\Pi$.
  If $\rank(\Pi)\le d-2$, we can easily find orthogonal $\ket{\psi},\ket{\phi}$ such that $\Pi\ket{\psi}\ket{\phi}=0$.
  Thus we must have $\rank(\Pi)= d-1$ and $\Pi$ has the spectral decomposition $\Pi = \sum_{i=1}^{d-1} \ketbra{v_i}{v_i}$.

  The constraint $\braket{v_i}{\psi,\phi}=0$ is then equivalent to $(L_i\ket\psi)\ket{\phi} = 0$ for $L_i\ket\psi := \bra{v_i}(\ket{\psi}\otimes I)$.
  Let $V = \Span\{(L_i\ket\psi)^\dagger\mid i\in[d-1]\}$.
  By construction, $V^\perp$ is the set of all vectors $\ket{\phi}$ such that $\braket{v_i}{\psi,\phi}=0$ for all $i=1,\dots,d-1$, i.e., $\Pi\ket{\psi}\ket{\phi}=0$.
  By assumption, this only holds for $\ket{\phi} \propto \ket{\psi}$.
  Thus, $V^\perp = \Span\{\ket{\psi}\}$ and $\dim(V) = d-\dim(V^\perp)=d-1$.
  Therefore, the $\{L_i\ket\psi \mid i\in[d-1]\}$ are linearly independent for any $\ket{\psi}$.

  The multi-homogeneous system $\sum_{i=1}^{d-1} x_i L_i\ket{\psi} = 0$ ($d$ equations) with variable sets $x\in \PP^{d-2}$ and $\ket\psi\in\PP^{d-1}$ then has a solution by \Cref{thm:bezout}.
  Therefore, $\sum_{i=1}^{d-1}x_i(L_i\ket{\psi})=0$ with $\ket{\psi}\ne 0$ and $x\ne 0$, which contradicts the linear independence of the $L_i\ket{\psi}$.
\end{proof}

\subsection{A brief aside: Solving a special case of \PS\ on qudits with multi-homogeneous systems}\label{sec:qudit-star}

We have seen that any $O(1)$-approximate \PS\ instance reduces to an MHS instance (\Cref{thm:MHScomplete}), which raises the question: Can one leverage techniques from solving multi-homogeneous systems to solve \PS\ instances?
Here, we briefly mention one such application, though it is not intended to be a focus of this work. 
Namely, Safey El Din and Schost~\cite{SS18} give an exact algorithm for computing all non-singular solutions (i.e. where the Jacobian matrix of the polynomial system has full rank) of \emph{dehomogenized} rational multi-homogeneous systems with a finite number of solutions.
We will not state their result, but note that in applying \cite{SS18} to \PS\, the computational complexity is polynomial in the number of WSDRs after removing one edge, which can generally be exponentially greater than just the number of WSDRs.
On some hypergraphs, however, this number is bounded, and thus \cite{SS18} provides a poly-time algorithm for \PS. 
For example, a star of $n+1$ qu$d$its, such that there are $d$ edges to $d-1$ qudits and $d-1$ edges to the others, only has a polynomial number of WSDRs for a fixed $d$, even after removing one edge (\Cref{fig:qutrit-star}).

\begin{figure}[t]
\begin{center}
\begin{tikzpicture}
  \usetikzlibrary {shapes.geometric}
  \tikzset{xx/.style={draw, circle, inner sep=2mm,very thick},dbl/.style={double distance=1mm,thick},dbl2/.style={double distance=1.6mm,thick},edg/.style={thick}}
  \node[xx](c) at (0,0) {};
  \node[xx](v1) at (60:1.5) {};
  \node[xx](v2) at (120:1.5) {};
  \node[xx](v3) at (180:1.5) {};
  \node[xx](v4) at (240:1.5) {};
  \node[xx](v5) at (300:1.5) {};
  \node[xx](v6) at (0:1.5) {};
  \draw[dbl] (c) -- (v1);
  \draw[dbl] (c) -- (v2);
  \draw[dbl] (c) -- (v3);
  \draw[dbl] (c) -- (v5);
  \draw[dbl2] (c) -- (v4);
  \draw[edg] (c) -- (v4);
  \draw[dbl2] (c) -- (v6);
  \draw[edg] (c) -- (v6);
\end{tikzpicture}
\end{center}
\caption{A \PS\ instance with a star-like topology. The circles represent qutrits. All edges have size $2$ and there are $9$ WSDRs (assign one from each set of triple edges to the center).}
\label{fig:qutrit-star}
\end{figure}

\section{High-degree, sparse univariate polynomials and TFNP}\label{scn:sparsepoly}

\Cref{scn:MHS} focused on low-degree multi-homogeneous systems and their relationship to TFNP. 
In this section, we study roots of a single high-degree univariate sparse polynomial.
\Cref{sscn:sparsedefs} first defines a new subclass of TFNP based on the Fundamental Theorem of Algebra, denoted \SFTA.
\Cref{sscn:SFTAinTFNP} shows that $\SFTA\subseteq\TFNP$.
\Cref{sscn:SFTAinMHS} shows how to reduce computing a root of a sparse univariate polynomial to QSAT with SDR. We can currently prove this reduction works in the exact case. We conjecture it also works in the approximate case, which would imply $\SFTA\subseteq \MHS$.
Finally, \Cref{sscn:MHSinSFTA} studies the converse question --- could $\MHS\subseteq \SFTA$?

\subsection{Definitions, the Fundamental Theorem of Algebra, and SFTA}\label{sscn:sparsedefs}

Sparse polynomials are well studied in the polynomial systems literature (e.g.~\cite{jindalEfficientlyComputingReal2017}). 
For our purposes, we use the following definition.

\begin{definition}[Sparse polynomial]\label{def:sparsepoly}
An \emph{$s$-sparse polynomial} $p(x)\in\CC[x]$ of degree $d$ has only $s\in O(\polylog(d))$ non-zero coefficients $a_i\in\CC$. The specification of $p$ is a list of $\lceil \log d \rceil$-bit approximations\footnote{One could also consider, e.g., exact representations via field extensions. For simplicity, we use approximate representations, which suffices as our goal is to find approximate roots.} $\widetilde{a_i}$ of each non-zero $a_i$, along with the corresponding indices $i\in\set{0,\ldots, d}$. 
\end{definition}

\noindent Thus, the degree is, by definition, exponentially larger than the input size.
In this paper, we only consider \emph{univariate} sparse polynomials.

Next, we recall the Fundamental Theorem of Algebra:
\begin{theorem}[Fundamental Theorem of Algebra]\label{thm:FTA}
  Every non-constant univariate polynomial $p\in\CC[x]$ has at least one complex root.
\end{theorem}

\noindent We can now define our second complexity class, $\SFTA$. For this, recall that a \emph{monic} polynomial has the coefficient of its highest degree non-zero term set to $1$.

\begin{definition}[Sparse Fundamental Theorem of Algebra (SFTA)]\label{def:SFTA}
Define $\SFTA$ as the set of TFNP relations $R(a,b)$ poly-time reducible (as defined in~\cite{Pap94}) to finding an $\epsilon$-approximate root $r\in \CC$ of a sparse monic univariate polynomial $p\in \CC[x]$ of degree $d$, where $\abs{r}\in[0,1+2\log(d)/d]$, and $\epsilon$ and $d$ may be functions in the input size.
That is, there exist poly-time computable functions $g$ and $h$, such that $g(a)$ outputs a sparse polynomial $p$, and $R(a,h(a,r))$ holds, where $r$ satisfying $\abs{r}\in[0,1+2\log(d)/d]$ is an approximate root of $p$ with $\abs{p(r)}\le \epsilon$. 
\end{definition}
\noindent Note the two restrictions to (1) roots in $[0,1+2\log(d)/d]$ and (2) $p$ being monic. We use both to obtain containment in \TFNP\ in \Cref{sscn:SFTAinTFNP}. 
For clarity, $2\log(d)$ can be replaced with any asymptotically larger term scaling as $\polylog(d)$, and containment in TFNP would still hold (\Cref{thm:inTFNP}).

\subsection{\SFTA\ is in \TFNP} \label{sscn:SFTAinTFNP}
Ideally, we would like $\SFTA\subseteq\TFNP$. And here we run into our first obstacle. Given a sparse polynomial $p$, it is not difficult to see that via square-and-multiply, the number of \emph{field operations} over $\CC$ to compute $p(x)$ is $\poly(n)$, for $n$ the size of input $a$ in \Cref{def:SFTA}. However, \TFNP\ is a class concerning \emph{bit complexity}, not field operation complexity. Unfortunately, it is immediate that if, say, $x=2$, then $p(x)=x^{2^n}$ for $x=2$ requires $2^n$ bits to represent, which is exponential in the input size. 
Moreover, even if the $p(x)$ itself has an encoding of size $\poly(n)$, the intermediate terms in its calculation (e.g. each monomial's value on $x$) may require exponentially large encodings. This phenomenon is sometimes referred to as \emph{intermediate expression swell}, and occurs for example in Euclid's GCD algorithm~\cite{gathenModernComputerAlgebra2003}. 

To circumvent this in our setting, we require two tricks. 
First, in \Cref{def:SFTA} we restrict attention to complex numbers $x$ satisfying $\abs{x}\in [0,1+\polylog(d)/d]$. Since $(1+
\polylog(d)/d)^d\in O(\polylog(d))$, this avoids the exponential blowup seen in the example above. More formally, one can show that $p(x)$ can be evaluated on this interval to within additive error $2^{-L}$ in time polynomial in $L$ and $n$. The following is essentially identical to Lemma 1 of \cite{jindalEfficientlyComputingReal2017}, except for a constant factor overhead since we are dealing with complex numbers, whereas \cite{jindalEfficientlyComputingReal2017} considers real numbers. 
This overhead disappears into the Big-Oh notation.

\begin{lemma}[Adaptation of Lemma 1 of \cite{jindalEfficientlyComputingReal2017}]\label{l:evalpoly}
  Let $p\in\CC[x]$ be an $s$-sparse polynomial, $x\in \CC$, and $L\geq 0$ an integer. Then, $f(x)$ can be computed to within additive error $2^{-L}$ with bit complexity
  \begin{align}
    \tilde{O}((s+\log d)(L+d\log[\max(1,\abs{x})]+\log d + s)),
  \end{align}  
  where $\tilde{O}$ omits logarithmic factors. 
\end{lemma}
\noindent The following corollary is immediate.
\begin{corollary}
  For $s$-sparse polynomial $p$ with encoding size $n$, $p(x)$ can be computed within additive error $2^{-L}$ for any $x\in[0,1+\polylog(d)/d]$ with bit complexity
  \begin{align}
    \tilde{O}((s+\log d)(L+s+\log d))\in \poly(n).
  \end{align}
\end{corollary}
\noindent The proof of \Cref{l:evalpoly} follows identically to \cite{jindalEfficientlyComputingReal2017}: By choosing 
\begin{align}
  K \in\Omega(L+\log s + d\log d \cdot \log[\max(1,\abs{x})]),
\end{align} 
one can approximately evaluate $p(x)$ (using square-and-multiply to compute powers) by truncating intermediate expressions to their $K$ most significant bits, while keeping the accrued additive error under control. The only difference here is that we need to independently track the error accumulated on both real and imaginary components of each complex number, which adds a constant factor overhead in the bit complexity. The details are omitted.

The second trick we need for containment in TFNP is to argue that we have not ``broken'' the Fundamental Theorem of Algebra in restricting to range $\abs{x}\in[0,1+\polylog(d)/d]$ --- namely, we must show that there always \emph{exists} a root in this range.
This is where the monic property of our polynomial will play a role, coupled with an application of Landau's inequality~\cite{landauQuelquesTheoremesPetrovitch1905}.
\begin{lemma}\label{l:rootbound}
  Let $p=\sum_{i=0}^da_ix^i$ be an $s$-sparse polynomial as per \Cref{def:sparsepoly}, which is additionally monic. 
  Then, there exists an $x\in\CC$ with\footnote{We thank an anonymous reviewer for catching a minor bug in a previous statement of this lemma, which claimed $\abs{x}>0$ (which is incorrect) in addition to the upper bound currently stated (which is correct). Indeed, the polynomial $x^d$ is sparse, but has only $0$ as roots, thus violating $\abs{x}>0$. As this lower bound is not necessary for our purposes in this work, we have omitted it.} 
  \begin{align}
    \abs{x}\leq 1+\left(\frac{\ln(\sqrt{s}d)}{d}\right).
  \end{align}
  such that $p(x)=0$.
\end{lemma}
\begin{proof}
  Assume without loss of generality $d$ is a power of $2$, by which $\set{\abs{a_i}}\leq d$ for all $i\in[0,\ldots, d]$. 
  The \emph{Mahler measure} of $p$ is defined $M(p)=\abs{a_d}\Pi_{j=1}^d \max(1,\abs{z_j})$, where $\set{z_j}_{j=1}^d$ is the set of roots of $p$, and in our setting the leading coefficient $a_d=1$ by assumption. An upper bound on $M(p)$ can be derived as follows. Landau's inequality~\cite{landauQuelquesTheoremesPetrovitch1905} says 
  \begin{align}
    M(p)\leq \sqrt{\sum_{j=0}^d \abs{a_j}^2}.
  \end{align} 
  Combining this with the fact that, by \Cref{def:sparsepoly}, each coefficient $a_i$ of $p$ satisfies $\set{\abs{a_i}}\leq d$,  we have
  \begin{align}\label{eqn:Mahler}
    M(p) \leq \sqrt{s}d.
  \end{align}
We now obtain a contradiction by lower bounding $M(p)$.  Assume, for sake of contradiction, that all roots $\abs{z_j}$ of $p$ satisfy $\abs{z_j} > (1+c/d)^{c}$ for natural number $c\gg 1$ to be chosen shortly. Then,
\begin{align}
  M(p) = \Pi_{j=1}^d \max(1,\abs{z_j}) > \left(1+\frac{c}{d}\right)^{cd}\geq\left(1+\frac{c}{d}\right)^{d+\frac{c}{2}}\geq e^{c},
\end{align}
where the third statement holds for $c\in\Theta(\polylog(d))$, and the last inequality follows since for all positive reals $n$ and $t$, $(1+t/n)^{n+t/2}\geq e^t$.
Setting $c=\ln(\sqrt{s}d)$ completes the proof of the upper bound.
  

\end{proof}

Combining \Cref{l:evalpoly} and \Cref{l:rootbound} immediately yields the desired claim.
\begin{theorem}\label{thm:inTFNP}
  $\SFTA\subseteq\TFNP$.
\end{theorem}

\subsection{Embedding univariate polynomials into QSAT with SDR: NP-hardness and towards $\SFTA\subseteq\MHS$}\label{sscn:SFTAinMHS}

\Cref{thm:inTFNP} showed $\SFTA\subseteq \TFNP$. Does the stronger containment $\SFTA\subseteq\MHS$ also hold? The main contribution of this section is to give a poly-time many-one reduction from \SFTA\ to exact \MHS\, i.e. to \MHS\ with $\epsilon=0$:

\begin{restatable}{theorem}{thmembedsparse}\label{thm:embedsparse}
  Let $P$ be an $s$-sparse polynomial of degree $d$.
  There exists an efficiently computable set $\Pi=\{\Pi_i\}_{i\in[m]}$ of $m = O(s\log(d))$ $3$-local and one $2$-local rank-$1$ constraints on $N=O(s\log d)$ qubits with an SDR, such that $P(x/y) = 0$ iff $\Pi (\ket{v_1}\otimes\dots \otimes \ket{v_N})=0$ with unit vector $\ket{v_1}=(x,y)^T\in\CC^2$. 
\end{restatable}

From this, we immediately obtain the following.

\begin{corollary}\label{cor:MHSsparse}
  Given monic $s$-sparse polynomial $p(x)\in\CC[x]$ of degree $d$, the problem of computing a root $x$ such that $p(x)=0$ is in $\MHS_{s',d'}(\epsilon)$, with number of equations $n=O(s\log d)$, at most $s'=2$ variables per group, total degree at most $d'=3$ per equation, and precision $\epsilon=0$. 
\end{corollary}

\noindent Recall, however, that in \Cref{def:MHS} we defined MHS with an allowed error tolerance $\epsilon$ at least inverse exponential in the input size, whereas \Cref{thm:embedsparse} requires $\epsilon=0$. 
We believe the construction of \Cref{thm:embedsparse} also yields an analogous result for the approximate case of inverse exponential $\epsilon$, but have not yet been able to prove it.
The main challenge appears to be controlling the error in the reduction, i.e. one would like to say that if one can find an $\epsilon$-approximate solution to PRODSAT with SDR, then one can resolve the roots of $P$ within some controlled precision $f(\epsilon)$. 
This is tricky, as the degree of $P$ is exponential, which may amplify errors.
We thus conjecture the following.

\begin{restatable}{conjecture}{conjSFTAinMHS}\label{conj:SFTAinMHS}
  $\SFTA\subseteq\MHS$.
\end{restatable}

\noindent In the meantime, \Cref{thm:embedsparse} will allow us to obtain \emph{NP-hardness} results for variants of QSAT with SDR, as given in \Cref{ssscn:NPhardness}.

\paragraph{Organization.} \Cref{ssscn:blocks} first develops tools for embedding univariate polynomials into QSAT instances. \Cref{ssscn:embeddingsparse} shows the analogue of \Cref{thm:embedsparse} for \emph{non-sparse} polynomials, i.e. for polynomial degree $d$ (\Cref{thm:embed-polynomial}). This will be useful for our NP-hardness results in \Cref{ssscn:NPhardness}. \Cref{ssscn:embeddingsparse} then gives the proof of \Cref{thm:embedsparse}, which proceeds similarly to \Cref{thm:embed-polynomial}.

\subsubsection{Building blocks}\label{ssscn:blocks}

We now give the basic building blocks, using $3$-local and $2$-local constraints, to design \PS\ instances whose solutions correspond to the roots of a univariate polynomial.
For this, we use the concept\footnote{Transfer functions are a formal generalization of the transfer matrix formalism, which has appeared in previous works, e.g.~\cite{bravyiEfficientAlgorithmQuantum2006,laumannc.r.PhaseTransitionsRandom2010}} of \emph{transfer functions} on qubits from \cite{AdBGS21}, for which we give a slightly simplified construction. Intuitively, a transfer function gives a necessary and sufficient condition for a rank-$1$ $k$-local clause $\ket{\phi}$ to be satisfied, given a partial assignment $\ket{\varphi_1}\cdots\ket{\varphi_{k-1}}$ to its first $k-1$ qubits.

\begin{restatable}{lemma}{lemTransferFunction}(Transfer function, $g$)\label{lem:transfer-function}
  Let $\ket{\phi}$ be a $k$-local constraint on qubits.
  There exists a polynomial $g:(\CC^2)^{k-1}\to \CC^2$ such that, for any partial assignment $v_1,\ldots,v_{k-1}$, the clause $\ket{\phi}$ is satisfied (i.e. $\braket{\phi}{v_1,\dots,v_k}=0$) iff\footnote{$\propto$ means up to scaling up to non-zero constant.} $\ket{v_k}\propto g(v_1,\dots,v_{k-1})$ or $g(v_1,\dots,v_{k-1})=0$.
  Moreover, $g$ is linear in the coefficients of each $v_i$.
\end{restatable}
\begin{proof}
  If $g(v_1,\dots,v_{k-1})=0$, we are trivially done, since the partial assignment already satisfies $\ket{\phi}$.
  For the remaining case, let $v' \coloneqq v_1\otimes\dotsm\otimes v_{k-1}$ and $x\coloneqq(v'\otimes I)^\dagger \phi$.\footnote{We do not use Dirac notation here as we make use of complex conjugates ($\overline{a}$) and transpositions ($a^T$) on their own.}
  Note that $g$ has the desired property if $g(v_1,\dots,v_{k-1}) = y$ is orthogonal to $x$, i.e. if $x^\dagger y=0$.
  To compute $y$, first compute $\overline x = (v'\otimes I)^T\overline \phi$.
  Then $y \coloneqq ZX\overline{x}$.
  For $x = (x_1,x_2)^T$, we have $y^\dagger = (x_2, -x_1)$ and thus $y^\dagger x=0$.
  Since we are on qubits, $y$ is the \emph{unique} choice of satisfying assignment for $v_k$, given $v_1,\ldots,v_{k-1}$.
  $\overline x$ is clearly linear in the coordinates of each $v_1,\dots v_{k-1}$.
  We also define $f(v_1,\dots,v_{k-1}) \coloneqq \overline{x}$.
\end{proof}

\noindent \emph{Simulating linear operations via $2$-local constraints.} Consider first the transfer function for a $2$-local constraint $H=\phi\phi^\dagger$.
By \Cref{lem:transfer-function}, $g(v_1) = ZX\overline x$ with 
\begin{equation}\label{eq:2local}
  \overline x = (v_1\otimes I)^T\phi = \left(\bmat{x_1\\y_1}\otimes I\right)^T\left(\bmat{a_{1}\\a_{2}}\otimes \bmat{1\\0} + \bmat{b_{1}\\b_{2}}\otimes \bmat{0\\1} \right) 
  = \bmat{a_1x_1 + a_2y_1\\b_1x_1+b_2y_1}.
\end{equation}
In words, the assignment on the second qubit must be \emph{orthogonal} to the the right hand side, $[a_1x_1+a_2y_1,b_1x_1+b_2y_1]^T$, in order to satisfy constraint $\ket{\phi}$. 
Note for the second equality that $[a_1,a_2]^T$ and $[b_1,b_2]^T$ are not necessarily orthogonal. In words, we can choose $H$ such that $g$ encodes an arbitrarily chosen linear combination of $x_1$ and $y_1$ in both coordinates.
\begin{example}\label{ex:2local}
  Suppose one wishes to enforce equality (up to rescaling) on product states on qubits $1$ and $2$, and suppose qubit $1$'s state is $(x_1, y_1)^T$. Setting $a_1=0$, $a_2=-1$, $b_1=1$, and $b_2=0$, the right hand side of \Cref{eq:2local} equals $(-y_1, x_1)^T$. The unique assignment to qubit $2$ orthogonal to this is $(x,y)$, thus enforcing qubit $2$ to equal qubit $1$.
\end{example}
\noindent \emph{Simulating quadratic operations via $3$-local constraints.} Similarly, we can choose $3$-local $H$ such that 
\begin{equation}\label{eq:3local}
\begin{aligned}
  \overline x &= \left(v_1 \otimes v_2 \otimes I\right)^T \phi 
  = \left(\bmat{x_1\\y_1}\otimes \bmat{x_2\\y_2} \otimes I\right)^T \phi\\
  &= \left(\bmat{x_1x_2\\x_1y_2\\y_1x_1\\y_1y_2}\right)^T \left(\bmat{a_{1}\\a_2\\a_3\\a_4}\otimes \bmat{1\\0} + \bmat{b_{1}\\b_{2}\\b_3\\b_4}\otimes \bmat{0\\1} \right)\\
  &= \bmat{a_1x_1x_2+a_2x_1y_2+a_3y_1x_1+a_4y_1y_2\\b_1x_1x_2+b_2x_1y_2+b_3y_1x_1+b_4y_1y_2}
  =\sum_{i,j\in[2]}\bmat{a_{ij}v_{1,i}v_{2,j}\\b_{ij}v_{1,i}v_{2,j}}
\end{aligned}
\end{equation}
and can therefore encode arbitrary linear combinations of the products $x_1x_2,x_1y_2,x_2y_1,x_2y_2$.

\begin{example}\label{ex:3local}
  Suppose given assignment $(x,y)^T$ to qubits $1$ and $2$, we wish to enforce qubit $3$'s assignment to encode the state (proportional to) $(x^2,y^2)^T$. Setting $a_1=a_2=a_3=0$, $a_4=-1$, $b_1=1$, and $b_2=b_3=b_4=0$, the right hand side of \Cref{eq:2local} equals $(-y^2, x^2)^T$. The unique assignment to qubit $3$ orthogonal to this is $(x^2,y^2)$, as desired.
\end{example}

\subsubsection{Embedding sparse polynomials into \PS}\label{ssscn:embeddingsparse}

With our building blocks in hand, we first show how to embed non-sparse polynomials into QSAT instances, i.e. where the degree $d$ is polynomial in the input size. Once we have this, a similar proof will yield \Cref{thm:embedsparse}.
\begin{theorem}\label{thm:embed-polynomial}
  Let $p$ be a polynomial of degree $d$ with $p(0)\ne0$.
  There exists an efficiently computable set $\Pi=\{\Pi_i\}_{i\in[m]}$ of $m = O(d)$ $3$-local and one $2$-local rank-$1$ constraints on $N=O(d)$ qubits with an SDR, such that $p(x/y) = 0$ iff $\Pi (\ket{v_1}\otimes\dotsm\otimes \ket{v_N})=0$ with unit vector $\ket{v_1}=(x,y)^T\in\CC^2$. 
\end{theorem}
\begin{proof}

Write $p(x) =\sum_{i=0}^dc_ix^i$ with $c_d \ne 0$ and $c_0\ne0$.
First, we homogenize $p$ by adding a variable $y$ such that $q(x,y) := \sum_{i=0}^dc_ix^iy^{d-i}$.
We now construct three sets of qubits and corresponding constraints.\\

\vspace{-1mm}
\noindent
\emph{First set.} The first set sets up the basic powers of $x$ and $y$ we need to simulate $q$.
Specifically, the first qubit $v_0=(x,y)^T$ represents variables $x,y$ in $q$.
With a $2$-local constraint, we create $\ket{v_1}\propto \ket{v_0}$ (see \Cref{ex:2local}).
Then, we can use $3$-local constraints and square-and-multiply to construct terms $\ket{v_i}:=(x^i, y^i)^T$ for any required $2\leq i\leq d-2$ (see \Cref{ex:3local}).
Observe that each time we add such a rank-$1$ constraint, we also add a new qubit to store the ``answer'' to the arithmetic operation the constraint is simulating.\\

\begin{figure}[t]
  \begin{center}
    \includegraphics[width=0.4\textwidth]{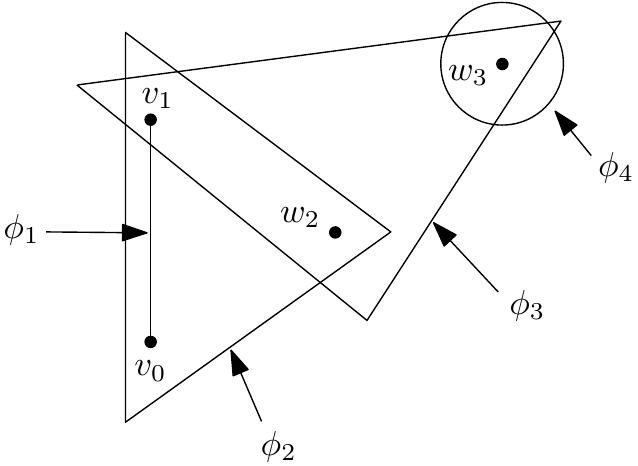}
  \end{center}
\caption{Construction of \Cref{thm:embed-polynomial} illustrated on input $p(x)=x^3-4x+5$, i.e. $d=3$. Then, $q(x,y)=x^3-4xy^2+5y^3$. Constraint $\phi_1$ is the equality constraint enforcing $\ket{v_1}\propto\ket{v_0}$. So, $\ket{v_0}=\ket{v_1}=[x,y]^T$. Next, we wish to enforce $\ket{w_2}=[x^2-4y^2,y^2]$. This is achieved via constraint $\phi_2$. Next, $\phi_3$ enforces $\ket{w_3}=[q(x,y),y^3$]. Finally, $\phi_4$ enforces \Cref{eq:final-constraint}. Observe this QSAT instance has an SDR: $(v_0,\phi_1),(v_1,\phi_2),(w_2,\phi_3),(w_3,\phi_4)$.}

\label{fig:univariatereduction}
\end{figure}

\vspace{-1mm}
\noindent\emph{Second set.} We next embed $q$ by recursively constructing a qubit with state $\ket{w_d} = (q(x,y),y^d)^T$.
The base case is $\deg(q)\leq 2$, i.e., $q(x,y) = c_2x^2 + c_1xy + c_0y^2$. 
Then, $\ket{w_2}=(c_2x^2 + c_1xy + c_0y^2, y^2)$ is constructed with a $3$-local constraint on $v_0$ and $v_1$ with $a_1=a_2=a_3=0$, $a_4=-1$, $b_1=c_2$, $b_2=c_1$,$b_3=0$, $b_4=c_0$.
For $\deg(q)>2$, we embed $q$ recursively, assuming that we can embed polynomials of degree $<d$.

For each step $t\geq 1$ of the recursion, let $j_t>0$ be minimal such that $c_{j_t} \ne 0$.
We construct polynomial $q_t(x,y)$ with degree $d_t$ defined as
\begin{equation}\label{eq:decomposition}
  q_t(x,y) := \sum_{i=0}^{d_t} c_{t,i}x^iy^{d_t-i} = x^{j_t}\cdot\underbrace{\sum_{i=j}^{d_t}c_{t,i}x^{i-j}y^{{d_t}-i}}_{r_t(x,y)} + c_{t,0}y^j\cdot y^{{d_t}-j}.
\end{equation}
Note that $t=1$ encodes our starting polynomial, i.e. $q_1(x,y):=q(x,y)$ of degree $d_1=d$. 
In timestep $t$, we recursively construct $\ket{w_{d_t-{j_t}}}:=(r_t(x,y), y^{d_t-j_t})^T$. 
(Note that $\ket{w_{d_t-j_t}}=0$ iff $x=y=0$.)
Given $\ket{w_{d_t-{j_t}}}$, we then construct $\ket{w_{d_t}}$ by adding a $3$-local constraint on $v_{j_t}$, $w_{d_t-j_t}$, and new qubit $w_{d_t}$ with $a_1=a_2=a_3=0$, $a_4=-1$, $b_1=1$, $b_2=b_3=0$, $b_4=c_{t,0}$ (as per \Cref{eq:3local}).

\medskip\noindent\emph{Third set.}
Thus far, our constraints force the ground space of our QSAT instance to encode $q(x,y)$. 
We need a final check to enforce this to correspond to a root for the original polynomial $p$.
For this, we add a $1$-local constraint $\ket{0}$ onto $w_d$, enforcing the equality
\begin{equation}\label{eq:final-constraint}
  \bmat{q(x,y)\\y^d} = \alpha \bmat{0\\1},
\end{equation}
where $\alpha$ is some non-zero constant of proportionality, which is \emph{a priori} unknown. 
The full construction is illustrated in \Cref{fig:univariatereduction}.

\medskip\noindent\emph{Correctness.} First, if there exists $x$ such that $p(x)=0$, $\ket{v_1}=(x,1)^T$ satisfies \Cref{eq:final-constraint}.
(All other constraints are immediately satisfied since they enforce the logic of the building blocks in \Cref{ssscn:blocks}.)
Conversely, consider some satisfying assignment to the set of QSAT constraints constructed.
It must necessarily also satisfy \eqref{eq:final-constraint} on qubit $w_d$ for some $\alpha \ne 0$.
Observe that $y\ne 0$, as otherwise $x=0$ as well (since $c_d\ne0$), which is not permitted for homogeneous coordinates.
Finally, since \eqref{eq:final-constraint} implies $q(x,y)=0$, we must have by homogeneity
\begin{equation}
  p\left(\frac xy\right) = q\left(\frac xy, \frac yy\right) = \frac{q(x,y)}{y^d}=0.
\end{equation}

\medskip
\noindent\emph{SDR.} To see that the constructed QSAT instance has an SDR, note first that we can trivially make it $3$-uniform by adding two ancilla qubits. Then, since all but the last recursive step of our construction simultaneously adds a new hyperedge and a new qubit, the system has an \emph{almost extending edge order} (defined later in \Cref{def:edgeorder}). The claim now follows from \Cref{cor:SDR}.
\end{proof}

\begin{rem}\label{rem:reduction}
   \Cref{thm:embed-polynomial} is not yet for sparse polynomials, but it will nevertheless be instructive to recall that in the definition of \SFTA, we focused on roots of polynomial $p\in\CC[x]$ in range $[0,1+2\log(d)/d]$, for $d$ the degree. 
   Given any root $x^*\in [0,1+2\log(d)/d]$ of $p$, the constructed QSAT instance of \Cref{thm:embed-polynomial} has a solution with 
   \begin{align}\label{eqn:prop}
      \ket{v_1}\propto(x^*,1)^T.
   \end{align} 
   The bounds $x^*\in [0,1+2\log(d)/d]$ now ensure $\enorm{(x^*,1)^T}$ is constant, so that the proportionality factor in \Cref{eqn:prop} is constant.
\end{rem}

We now proceed to showing the sparse version of \Cref{thm:embed-polynomial}, but first remark that \Cref{thm:embed-polynomial} suffices already to show NP-hardness results of slight variants of QSAT with SDR in \Cref{ssscn:NPhardness}.\\

\vspace{-1mm}
\noindent\emph{The sparse case.} The proof of the sparse case now proceeds analogously to the non-sparse case.

\thmembedsparse*
\begin{proof}
  The key observation is that in recursive step $t$ of \Cref{eq:decomposition}, we factor $x_{j_t}$ for $j_t>0$ the minimal value satisfying $c_{j_t}\neq 0$. 
  This implies the number of recursive calls scales with sparsity $s$, not degree $d$.  
  Thus, a construction analogous to \Cref{thm:embed-polynomial} can be used, except in the first set of constraints, we will need to construct $O(s\log d)$ terms $\ket{v_i}=(x^i,y^i)$, where $i$ can now be exponential in the input size. 
  This is easily handled by using square-and-multiply on qubits encoding the various $\ket{v_i}$ to obtain high powers $i$ using $\log(i)$ steps (similar to \Cref{ex:3local}).
\end{proof}

\subsubsection{Detour: NP-hardness results for slight variants of QSAT with SDR}\label{ssscn:NPhardness}

With \Cref{thm:embed-polynomial} (non-sparse case) in hand, we first immediately obtain an $\NP$-hardness result for a variant of QSAT with SDR.
This complements Goerdt's result that deciding whether there exists a \emph{real} product state solution is $\NP$-hard~\cite{goerdtMatchedInstancesQuantum2019}.

\thmEqual*
\begin{proof}
  This theorem follows from the $\NP$-hardness of deciding whether a sparse polynomial has a root of modulus $1$~\cite{Pla84}.
\end{proof}

\noindent Goerdt also shows that deciding whether a QSAT instance with an SDR and \emph{just one} additional constraint is $\NP$-hard. 
We can also recover this result here via our construction.

\thmAdded*
\begin{proof}
  Plaisted proves that it is $\NP$-hard to decide whether two sparse polynomials have a common root~\cite{Pla84}.
  We can embed this problem into \PS\ by adding a second adding a second polynomial in the above construction, which requires only a single unmatched edge.
\end{proof}

This stands in stark contrast to the classical setting, where deciding whether a CNF-SAT formula with an SDR and $O(1)$ additional clauses is still in P.
The following theorem generalizes a result due to Berman, Karpinski, and Scott~\cite{BKS07}, who prove that satisfiability of $(3,4^{(k)})$-SAT (i.e. a $3$-SAT instance in which $k$ variables occur $4$ times and the remaining variables $3$ times) is efficiently solvable.

\newcommand{\scrC}{\mathscr{C}}
\begin{theorem}\label{thm:BKS}
  Let $\scrC$ be the set of clauses of a SAT instance in CNF on $n$ variables $V$ such that there exists a subset $\scrC'\subseteq \scrC$ with an SDR, i.e. a perfect matching between $\scrC'$ and $V$.
  Satisfiability of $\scrC$ can be determined in time $(2n)^k\poly(n)$ for $k:=|\scrC|-n$.
\end{theorem}
\begin{proof}
  This proof follows the same outline as ~\cite[Theorem 1]{BKS07}, but we need to give a different argument for the existence of a surjective witness function.
  Consider a satisfying assignment $\phi$ to $\scrC$ and define a \emph{witness function} $w:\scrC\to V$ such that for each $C\in\scrC$, the variable $x=w(C)$ occurs in $C$ and its literal evaluates to true under $\phi$, i.e., if $\phi(x)=1$, then $C$ contains the literal $x$, and otherwise $\neg x$.
  We argue that if $\scrC$ is satisfiable, then there exists a satisfying assignment with a surjective witness function.
  Let $\phi$ be a satisfying assignment with witness function $w$.
  If $w$ is surjective, we are done.
  Otherwise, there exists a variable $x\notin \image(w)$.
  Let $C\in\scrC'$ be the clause assigned to $x$ in the SDR.
  Create $\phi',w'$ by only changing $\phi(x)$ and $w(C)$ such that the literal of $x$ in $C$ evaluates to true and $w(C)=x$.
  $\phi'$ is still a satisfying assignment as $x\notin\image(w)$ and $\phi'(y)=\phi(y)$ for all $y\ne x$.
  Repeat until $\image(w)=V$, which takes at most $n$ iterations since each iteration increases the number of clauses $C\in\scrC'$ such that $w(C)$ is matched with $C$ in the SDR. The remainder of the proof is the same as \cite{BKS07}.
\end{proof}

\subsection{Is \MHS\ in \SFTA?}\label{sscn:MHSinSFTA}

We now ask the question --- could $\MHS\subseteq\SFTA$? 
In words, can the solutions of a low-degree multi-homogeneous system be mapped to the roots of a high-degree univariate polynomial? 
We conjecture $\MHS\not\subseteq \SFTA$, according to which no such efficient reduction should be possible. However, one can still show a non-trivial result in this direction --- we show that in the generic setting (\Cref{def:generic}), a low-degree multi-homogeneous system can be reduced to a single high-degree univariate polynomial $p$, where $p$ requires polynomial \emph{space} to compute.
Under the hood, this utilizes a clever lemma of Canny, which we first state.

\begin{lemma}[Canny's Lemma (Lemma 2.2 of \cite{cannyAlgebraicGeometricComputations1988})]\label{l:canny}
  Let $p_1$ through $p_n$ be homogeneous polynomials in variables $x_0,\ldots x_n$, with $D\leq d_1\cdots d_n$ isolated solution rays $(\alpha_{0,j},\ldots,\alpha_{n,j})$, $j=1,\ldots, D$. Let $N\leq D$ be the number of solution rays not at infinity, for example, with $\alpha_{0,j}\neq 0$. Then there is a univariate polynomial $q(x)$ of degree $N$, and rational functions $r_1(x),\ldots,r_n(x)$, such that every solution ray not at infinity is of form $(1,r_1(\theta), \ldots, r_n(\theta))$ for some root $\theta$ of $q(x)$. The polynomials $q(x)$ and $r_k(x)$ can be computed in polynomial space.
\end{lemma}

We will also require two further tools from algebraic geometry (see, e.g., \cite{CLO05}).

\begin{definition}[Newton polytope (page 310 of \cite{CLO05})]
  Let $f\in\CC[x_1,\ldots, x_n]$ be such that 
  $
    f=\sum_{\alpha\in\ZZ^n_{\geq 0}}c_\alpha x^\alpha.
  $
  The Newton polytope of $f$ is $\operatorname{Conv}(\set{\alpha\in\ZZ^n_{\geq 0}\mid c_\alpha\neq 0}$).
\end{definition}

\begin{theorem}[Berstein-Khovanksii-Kushnirenko (BKK; theorem 5.4 of \cite{CLO05})]\label{thm:bkk}
  Given Laurent polynomials $f_1,\ldots f_n$ over $\CC$ with finitely many common zeroes in $(\CC^\ast)^n$, let $P_i$  be the Newton polytope of $f_i$ in $\RR^n$. Then the number of common zeroes of the $f_i$ in $(\CC^\ast)^n$ is bounded above by the $n$-dimensional mixed volume of $(P_1,\ldots, P_n)$ (Definition 4.11 of \cite{CLO05}). Moreover, for generic choices of the coeﬃcients in the
$f_i$, the number of common solutions is exactly the $n$-dimensional mixed volume of $(P_1,\ldots, P_n)$.
\end{theorem}

We are now ready to prove the results of this section. 
\begin{proposition}\label{prop:help}
Let $(G,w)$ be a weighted hypergraph with a WSDR and such that $|V(G)|_w=|E(G)|$. Let $H$ be a generic instance of QSAT with underlying weighted hypergraph $(G,w)$. Then every product ground state of $H$ is of the form
\begin{equation}\label{eq:40}
  \ket{\psi_t}=(\ket{0}+t_{1,1}\ket{1}+\cdots t_{1,w(1)}\ket{w(1)})\otimes \cdots\otimes (\ket{0}+t_{|V(G)|,1}\ket{1}+\cdots t_{|V(G)|,w(|V(G)|)} \ket{w(|V(G)|)})
\end{equation}
with $t_{i,j}\neq 0$ for all $i=1,\ldots,|V(G)|$, and $j=1,\ldots,w(i)$.
\end{proposition}

\begin{proof}
Let $H_e$ be the clause corresponding to $e\in E(G)$ and consider the multivariate polynomial $p_e(t)$ in the variables $t_{i,j}$ such that $p_e(t)=0$ if and only if $H_e\ket{\psi_t}=0$. The Newton polytope $Q_e$ of $p_e$ is the product of simplices of dimension $w(i)$, one for each vertex $i\in e$. Hence for $\lambda_e>0$, $e\in E(G)$,
\begin{align}
V\left( \sum_{e\in E(G)} \lambda_e Q_e\right) = \prod_{i\in V(G)} \frac{\left(\sum_{v\in e} \lambda_e\right)^{w(i)}}{w(i)!} = N(G,w) \prod_{e\in E(G)}\lambda_e+\textrm{lower order terms}
\end{align}
where $N(G,w)$ is the number of WSDRs of $(G,w)$. On the other hand, by definition, $N(G,w)$ is the mixed volume of the polytopes $Q_e$, $e\in E(G)$. Therefore, by the BKK theorem (\Cref{thm:bkk}), there are $N(G,w)$ product solutions of the form \eqref{eq:40} with all $t_{i,j}\neq 0$. But since $N(G,w)$ is also equal to the B\'ezout number of the multi-homogeneous system associated with $H$, we conclude that, generically, this accounts for {\it all} product solutions of $H$.
\end{proof}

We can now show that, generically, QSAT with SDR can be reduced in polynomial space to solving for the roots of a single high degree univariate polynomial.

\begin{theorem}\label{thm:pspace}
Let $(G,w)$ be a weighted hypergraph with a WSDR and such that $|V(G)|_w=|E(G)|$. Let H be a generic instance of QSAT with underlying weighted hypergraph $(G,w)$. Then there is a univariate polynomial $q(x)$ of degree at most 
\begin{align}
D=\prod_{e\in E(G)}|e|
\end{align}
and rational functions $r_{i,j}(x)$ for every $i=1,\ldots,|V(G)|$ and $j=1,\ldots,w(i)$ such that if $x$ is a root of $q$ and
\begin{align}
r(x)=\prod_{i=1}^{|V(G)|}\prod_{j=1}^{w(i)} r_{i,j}(x)\neq 0
\end{align}
then
\begin{align}
(\ket{0}+r_{1,1}(x)\ket{1}+\cdots +r_{1,w(1)}\ket{w(1)})\otimes \cdots \otimes (\ket{0}+r_{|V(G)|,1}(x)\ket{1}+\cdots+r_{|V(G)|,w(|V(G)|)}\ket{w(|V(G)|)})
\end{align}
is a product solution of $H$. Conversely, every product solution is of this form for some root $x$ of $q$ such that $r(x)\neq 0$. Moreover, $q(x)$ and all the rational functions $r_{i,j}(x)$ can be calculated in polynomial space.
\end{theorem}

\begin{proof} Consider product solutions of $H$ of the form 
\begin{equation}
\ket{\psi_t}=(\ket{0}+t_{1,1}\ket{1}+\cdots t_{1,w(1)}\ket{w(1)})\otimes \cdots\otimes (\ket{0}+t_{|V(G)|,1}\ket{1}+\cdots t_{|V(G),w(|V(G)|)} \ket{w(|V(G)|)}).
\end{equation}
 Let $P_e$ be the homogenization of $p_e$ obtained by adding the single variable $t_0$ so that $P_e=0$ defines a hypersurface $X_e$ of degree $|e|$ in $\mathbb P^{|V(G)|_w}$. By Canny's Lemma (\Cref{l:canny}), there is a polynomial $q(x)$ of degree $N\le D$ and rational functions $r_{i,j}(x)$ for every $i=1,\ldots,|V(G)|$ and $j=1,\ldots,w(i)$ such that every point in $(\bigcap_{e\in E(G)} X_e)\setminus \{t_0=0\}$ has coordinates $t_0=1$ and $t_{i,j}=r_{i,j}(x)$ whenever $x$ is a root of $q(x)$. Then $r_{i,j} (x)=0$ for some $i$ and $j$ if and only if the corresponding element of $\bigcap_{e\in E(G)} X_e$ belongs to one of the coordinate planes and thus represent a ``spurious'' solution in the sense that the corresponding product state $\eqref{eq:40}$ is not a solution of $H$ (since by the BKK Theorem (\Cref{thm:bkk}) all product solutions satisfy the additional condition $t_{i,j}\neq 0$ for all $i$ and $j$). The last statement of the claim follows directly from Canny's Lemma. 
\end{proof}

\begin{rem}
When $w=1$ (so that all qu-$d$-its are qubits), we can be more precise about the degree $N$ of $q(x)$. By Canny's Lemma, $D-N$ is the number of points in the intersection of $\bigcap_{e\in E(G)} X_e$ with the hyperplane at infinity. On the other hand, setting $t_0=0$ drastically reduces the polynomial $P_e$ to $\prod_{i\in e} t_i$. Let $f$ be the Boolean function in CNF form with all positive literals and underlying hypergraph $G$. If $n$ denotes the number of satisfying assignments of $f$, then $N=D-n+1$.  
\end{rem}



\section{Efficiently solvable special cases of QSAT with WSDR}\label{scn:algs}

We next give parameterized classical algorithms for QSAT with SDR, which allow for efficient solutions in special cases. 

\vspace{-1mm}
\paragraph{Brief overview of techniques.} We first briefly sketch the ideas for two of our three algorithms for special cases of QSAT with (W)SDR. The first algorithm we discuss, which solves \emph{non-generic} \PS\ instances (\Cref{thm:algInformal}), begins with the same approach as~\cite{AdBGS21}. At a high-level, this approach takes the qubits comprising the hard ``core'' of the instance, sets these qubits in a specific manner so as project onto a smaller space, and subsequently forces assignments onto all other qubits via transfer functions. 
This approach breaks down in the non-generic case, which can have unentangled constraints that can prevent this propagation of assignments. The classical analogue to this problem can be seen with constraint $x\vee y$: When $x=1$, the constraint is already satisfied, and thus no assignment is propagated onto $y$. (Note all such classical SAT constraints are unentangled when embedded into QSAT.) 
To overcome this, the key idea we introduce is that, when this algorithm gets stuck, we prove that we can actually \emph{recurse} the entire process, as its existing ``almost extending order'' remains valid. 

The second algorithmic contribution we discuss is our extension of using \emph{transfer filtrations} (the framework enabling transfer functions) to QSAT on qu\emph{d}its.
This requires a careful arrangement of clauses into a convenient order (exploiting the geometry of the instance) so as to reduce the problem to a system with fewer equations in fewer variables. The trade-off is that the degree of the resulting equations can be rather large. Nevertheless, we show that for certain non-trivial infinite families of interaction hypergraphs, such as the Pinwheel graph (\Cref{fig:pinwheel}), we can efficiently solve the corresponding instance of \PS, exponentially outperforming the brute force approach.

\vspace{-1mm}
\paragraph{Organization.}
\Cref{sscn:solvableprelims} introduces necessary definitions and lemmas.
\Cref{sscn:nonGeneric} solves \emph{non-generic} special cases of QSAT on qubits with an SDR; this improves on \cite{AdBGS21}, which worked only for generic instances.
\Cref{sscn:solvingtransferk} returns to the generic setting with SDR, but instead widens the class of qubit QSAT instances one can efficiently solve generically beyond \cite{AdBGS21}.
\Cref{sscn:wtransfer} shows how to extend the transfer filtration technique of \cite{AdBGS21} from qu\emph{b}its to qu\emph{d}its and WSDRs, solving the Pinwheel graph in \Cref{ssscn:efficient} exponentially faster than via brute force.

\subsection{Transfer functions, filtrations, and extending edge orders}\label{sscn:solvableprelims}

We begin by restating the notion of transfer functions for convenience:

\lemTransferFunction*

\paragraph{Transfer filtrations.} In the qubit setting, \cite{AdBGS21} efficiently solves QSAT with SDR
for \emph{generic} instances of \emph{transfer type} $b=n-m+1$ (Definition~\ref{def:transfer-type} below), where $m$ denotes the number of constraints and $n$ the number of qubits.
This transfer type restriction is important, as it allows~\cite{AdBGS21} to reduce the entire QSAT with SDR instance to approximating a root of a single univariate polynomial.
Note also the algorithm is parameterized, i.e. its runtime is polynomial in the input size but exponential in the \emph{foundation size} (\Cref{def:transfer-type}) and \emph{radius} (Definition \ref{def:radius}).

We begin by stating the required definitions, and give intuition as to why transfer type $b=n-m+1$ allows reductions to the univariate polynomial case in~\cite{AdBGS21}.
We first recall the definition of a \emph{transfer filtration}, which is a particular type of hyperedge ordering useful for solving \PS.

\begin{definition}[Transfer filtration \cite{AdBGS21}]\label{def:transfer-type}
A hypergraph $G=(V,E)$ is of {\it transfer type $b$} if there exists a chain of subhypergraphs (denoted a {\it transfer filtration of type $b$})  $G_0\subseteq G_1\subseteq \cdots \subseteq G_m=G$ and an ordering of the edges $E(G)=\{e_1,\dots,e_m\}$ such that
\begin{enumerate}[label=(\arabic*)]
\item $E(G_i)=\{e_1,\ldots,e_i\}$ for each $i\in \{0,\ldots,m\}$,
\item $|V(G_i)|\le |V(G_{i-1})|+1$ for each $i\in \{1,\ldots,m\}$,
\item if $|V(G_i)|= |V(G_{i-1})|+1$, then $V(G_i)\setminus V(G_{i-1}) \subseteq e_i$,
\item $|V(G_0)|=b$, where we call $V(G_0)$ the \emph{foundation},
\item and each edge of $G$ has at least one vertex not in $V(G_0)$.
\end{enumerate}
\end{definition}

\begin{definition}[Radius of transfer filtration \cite{AdBGS21}]\label{def:radius}
Let $G$ be a hypergraph admitting a transfer filtration $G_0\subseteq \cdots \subseteq G_m=G$ of type $b$. Consider the function $r:\{0,\ldots,m\}\to \{0,\ldots,m-1\}$ such that $r(0)=0$ and $r(i)$ is the smallest integer such that $|e_i\setminus V(G_{r(i)})|= 1$ $\forall i\in\{1,\ldots,m\}$. The {\it radius of the transfer filtration $G_0\subseteq \cdots \subseteq G_m=G$ of type $b$} is the smallest integer $\beta$ such that $r^\beta(i)=0$ for all $i\in \{1,\ldots,m\}$ ($r^\beta$ denotes composition of $r$ with itself $\beta$ times). The {\it type $b$ radius of $G$} is the minimum value $\rho(G,b)$ of $\beta$ over the set of all possible transfer filtrations of type $b$ on $G$.
\end{definition}

\noindent\emph{Intuition.} We can view the transfer filtration as a sequence of edges wherein each edge adds at most one extra node, as enforced by condition (2) above.
The foundation is made up by all but one of the vertices in edge $e_1$.
Then transfer type $b=n-m+1$ implies that $n=b+m-1$ and thus one edge does not add an additional vertex (i.e. $V(G_i) = V(G_{i+1})$ in (2)).
Note, given a product assignment to the qubits in $V(G_{i-1})$, we can satisfy the constraint of edge $e_i$ using the corresponding transfer function (see Lemma~\ref{lem:transfer-function}).
This leaves a single \emph{non-extending} constraint that does not add a new qubit, and thus cannot immediately be satisfied.
To solve the system, assign the foundation qubits $v_1=\dotsm=v_{b-1}=\ket{0}$, and $v_b=\ket{0} +x\ket{1}$.
The transfer functions then set each qubit to a polynomial expression in $x$.
Satisfying the non-extending constraint then reduces to finding a root of a univariate polynomial of degree exponential in the radius.
Note, the above algorithm outline does not quite match~\cite{AdBGS21}, where qubits are duplicated so that every edge adds a new qubit and then equality of copies is enforced via \emph{qualifier constraints}.

\paragraph{Extending edge order.}
As outlined above, the transfer filtration gives us an order of the constraints that we can use to solve the system.
We formalize this notion by defining the \emph{extending edge order}, which turns out to be equivalent to the transfer filtration, but is useful in handling vanishing transfer functions algorithmically.

\begin{definition}[Extending Edge Order]\label{def:edgeorder}
  Let $G = (V, E)$ be a hypergraph.
  An edge order $e_1,\dots,e_m$ is \emph{extending} if $e_i\setminus V_{i-1} \ne \emptyset$ for $i\in [m]$, where $V_i \coloneqq \bigcup_{j=1}^i V(e_i)$ and $V_0=\emptyset$.
  We say the order is $a$-\emph{almost extending} if $|\{i:V_i = V_{i-1}\}|\le a$.
  We say it is \emph{almost extending} if $a=1$.
\end{definition}

\begin{lemma}\label{lem:transfer-type-almost-extending}
  Let $G=(V,E)$ be a hypergraph, $b^*$ its minimum transfer type and $a^*$ minimal such that $G$ has an $a^*$-almost extending edge order.
  Then $b^* = n-m+a^*$.
\end{lemma}
\begin{proof}
  First, show $a^* \le b^* - n + m$ by constructing an $a$-almost extending order given a transfer filtration of type $b=n-m+a$.
  Let $G_0\subseteq\dotsm\subseteq G_m=G$ be a transfer filtration of type $b$.
  By Definition~\ref{def:transfer-type}, $E(G_i) = \{e_1,\dots,e_i\}$.
  Let $V_i = \bigcup_{j=1}^i e_j$.
  Then $V(G_i) = V_i \cup V(G_0)$.
  We have $n = b+m-a$.
  So if $a=0$, every edge must cover one additional vertex and $e_1,\dots,e_m$ is an extending edge ordering.
  If $a>0$, then there are exactly $a$ edges that do not cover a new vertex, since one edge can add at most one new vertex.
  Let $i_1<\dotsm<i_{m-a}$ the indices of edges that add a new vertex (i.e. $\abs{V(G_i)}= \abs{V(G_{i-1})} = 1$), and $j_1<\dotsm <j_a$ the indices of the remaining edges (i.e. $V(G_i) = V(G_{i-1})$).
  Note, by definition $e_1$ always adds at least one vertex.
  Then $e_{i_1},\dots,e_{i_{m-a}},e_{j_1},\dots,e_{j_a}$ is $a$-almost extending.

  Second, we show $b^* \le n - m + a^*$.
  Let $e_1,\dots,e_m$ be an $a^*$-almost extending order.
  Without loss of generality, $e_1,\dots,e_{m-a^*}$ are extending.
  Define vertices $u_1,\dots,u_{m-a^*}$ such that $u_i\in e_i\setminus V_{i-1}$.
  Then we argue a valid foundation is given by the ``redundant vertices'' $R \coloneqq \bigcup_{i=1}^{m-a^*} (e_i\setminus V_{i-1}\setminus \{u_i\})$.
  Hence, the transfer filtration is defined with $V(G_0) = R$ and $E(G_i) = \{e_1,\dots,e_i\}$.
  The transfer type is then $b = |R| = n - m + a^*$.
  Conditions (1) to (4) are satisfied by construction.
  For condition (5) we have to show that $e\nsubseteq R$ for all $e\in E$.
  For an extending edge $e_i$, we have $u_i\notin R$, because $u_i\notin V_{j<i}$ and $u_i\in V_{j\ge i}$, and thus $u_i\notin R$.
  For a non-extending edge $e_i$, we argue that that $e_i\subseteq R$ would violate minimality of $a^*$:
  Suppose there exists a minimal $j$ such that $e_i \subseteq R_j \coloneqq \bigcup_{i=1}^{j} (e_i\setminus V_{i-1}\setminus \{u_i\})$.
  Then we could construct a new ($a^*-1$)-almost extending edge order by moving $e_i$ in between $e_{j-1}$ and $e_j$.
  Then $e_i$ would be extending because it contains at least one of the ``redundant vertices'' of $e_j$ and $e_j$ is still extending as it adds $u_j$.
  The edges $e_{j+1},\dots,e_{m-a^*}$ remain extending because $e_i\subseteq R_j\subseteq V_j$.
\end{proof}

Finally, we state a corollary which we used in \Cref{sscn:SFTAinMHS}.

\begin{corollary}[\cite{aldiEfficientlySolvableCases2021}]\label{cor:SDR}
  Let $G$ be a $k$-uniform hypergraph for any $k>0$. If $G$ has an almost extending edge order, then $G$ has an SDR.
\end{corollary}
\begin{proof}
  This follows immediately from \Cref{lem:transfer-type-almost-extending} and the fact that if $G$ is a $k$-uniform hypergraph of transfer type $b$ and such that $|E(G)| = |V (G)| - b+1$, then $G$ has an SDR~\cite{aldiEfficientlySolvableCases2021}.
\end{proof}

\subsection{Solving non-generic instances on qubits of transfer type \texorpdfstring{$b=n-m+1$}{b=n-m+1}}\label{sscn:nonGeneric}

We now introduce an efficient algorithm for QSAT with SDR on qubits without genericity requirements, i.e. that can handle constraints which are ``edge cases'' (e.g. Schmidt rank-$1$ or unentangled constraints). For this, we define the radius of an almost extending edge order as the radius of the transfer filtration constructed in the proof of Lemma~\ref{lem:transfer-type-almost-extending}.

\paragraph{The challenge.} In the non-generic case, one issue we need to deal with is that transfer functions can become $0$, i.e., after assigning the first $k-1$ qubits of a $k$-local constraint, the corresponding constraint is already satisfied for every choice of the $k$-th qubit (this is the case of $g=0$ in \Cref{lem:transfer-function}).
For example, $\ket{\phi} = \ket{000}_{123}$ with $\ket{v_1} =\ket{1}$ is satisfied for all choices of $\ket{v_2}$ and $\ket{v_3}$.
As a result, assignments to a subset of qubits are not propagated throughout the system. 
This issue is circumvented in~\cite{AdBGS21} through the genericity assumption, which we shall remove.

\paragraph{The algorithm.} The next theorem generalizes the algorithm of~\cite[Section 4.4]{AdBGS21}, which solved \emph{generic} instances of transfer type $b=n-m+1$.
We say a product state $\ket\psi = \ket{\psi_1,\dots,\psi_n}$ is an $\epsilon$-approximate solution to a \PS\ instance if $|\Pi_i\ket{\psi}|\le \epsilon$ for all constraints $\Pi_i$.
We require an approximately normalized solution, i.e., $\braket{\psi_i}{\psi_i}\in[1-\epsilon,1+\epsilon]$ for all $i\in [n]$.
The error incurred by normalization was not considered in~\cite{AdBGS21}.
Here we handle this issue by mostly computing with exact representations of algebraic numbers.

\begin{theorem}\label{thm:alg}
  Let $\Pi$ be a QSAT instance on qubits with coefficients in $\QQ[i]$ such that the constraint hypergraph $G$ has an almost extending edge order of radius $r$, and edges have size at most $k$.
  Then an $\epsilon$-approximate solution can be computed in time $\poly(L, \log\epsilon^{-1}, k^r)$, where $L$ is the input size.
  For sufficiently generic instances, an exact representation of a solution can be obtained.
\end{theorem}

\noindent Before giving the proof, a comment on the dependence of the runtime above on radius $r$: The function $r$ in the definition of radius divides the edges into layers such that layer $\beta$ consists of the edges such $e_i$ such that $r^\beta(i) =0 \ne r^{\beta-1}(i)$.
Note, the radius generally depends on the choice of vertices $u_1,\dots,u_{m-1}$.
Kremer~\cite{Kre24} gives a poly-time algorithm to compute an almost extending edge order and choice of vertices $u_1,\dots,u_{m-1}$ that minimize the radius.

\begin{proof}[Proof of \Cref{thm:alg}]
  Let $e_1,\dots,e_m$ be an almost extending edge order such that $e_m$ is the single non-extending constraint.
  Let $u_1,\dots,u_{m-1}$ be defined as in the proof of Lemma~\ref{lem:transfer-type-almost-extending}.
  We also assume that $u_{m-1}\notin e_m$, i.e., $e_{m}\nsubseteq e_1\cup\dots\cup e_{m-2}$.
  This is valid because once we have found a product solution that satisfies the non-extending constraint, it becomes trivial to add more extending constraints and find product assignments for the added qubits that satisfy the added constraint.

  Next we describe the algorithm.
  Let $R$ be the set of ``redundant'' vertices as in the proof of Lemma~\ref{lem:transfer-type-almost-extending}.
  We say a vertex $v$ depends on a vertex $u$ if we reach $v$ from $u$ when following the edge order.
  There must be at least one vertex $u_0\in R$, such that $u_{m-1}$ \emph{depends} on $u_0$, even after removing $R\setminus\{u_0\}$.

  Next, add a $1$-local constraint $\ket{1}$ to all qubits in $R\setminus \{u_0\}$.
  Assign all qubits corresponding to vertices $v\in R\setminus\{u_0\}$ to $\ket0$.
  Next, remove all $1$-local constraints (hyperedges of size $1$) on vertices besides $u_{m-1}$ by assigning the orthogonal state to the corresponding qubit and reducing the remaining constraints.
  The resulting edge order remains almost extending, although there may now be a single $1$-local constraint on $u_{m-1}$.
  However, either $e_{m}$ or $e_{m-1}$ remains of size $\ge2$ because $u_{m-1}$ depends on $u_0$ and $1$-local residual constraints on a vertex $u_i$ are only created after all vertices in $e_{i}\setminus\{u_i\}$ have been assigned, which is not possible on the path from $u_0$ to $u_{m-1}$.
  Repeat these two steps until either the edge order is extending or we obtain an almost extending edge order with a single redundant vertex $u_0$.

  We may now assume that $u_0$ is the single redundant vertex, and therefore $u_0\in e_1$.
  Then via the transfer functions, we can write any vertex as a homogeneous polynomial in the amplitudes of the qubit $u_0$, i.e., $g_i(u_0) = u_i$ (see Lemma~\ref{lem:transfer-function}).
  For a satisfying assignment, we have $g_{m-1}(u_0) = \lambda g_m(u_0)$ (for some $\lambda\in\CC^*$), or equivalently $q(u_0)=f_{m-1}(u_0)^T\cdot g_m(u_0) = 0$, where $f_{m-1}$ is defined as in Lemma~\ref{lem:transfer-function}.
  $q$ is then a homogeneous polynomial of degree at most $(k-1)^{r}$ (see~\cite{AdBGS21} for more details).
  $q$ is not constant since $u_{m-1}$ depends on $u_0$ and so one of $e_{m-1},e_m$ is not $1$-local and $f_{m-1}$ or $g_m$ is not constant.
  First, we check whether $\ket{u_0} = \ket{0}$ gives an $\epsilon$-approximate solution.
  If not, let $\ket{u_0} = x\ket{0} + \ket{1}$ and compute a root $x$ of $q(x) \coloneqq q(x\ket{0} + \ket{1})$.
  $x$ has an exact representation in the field of algebraic numbers, which can be obtained in polynomial time in the degree and description size~\cite[Theorem 8]{AS20}.
  After computing $x$, we can compute the $g_i(x)$ with~\cite[Theorem 4]{AS20}.
  As argued in~\cite{AdBGS21}, we have $g_i(x) \ne 0$ for all $i$ if the constraint system is chosen generically, and we have an exact representation of a \PS\ solution.

  However, for non-generic instances, we can have $g_i(x)=0$.
  In that case, compute the non-zero $g_1(u_0),\dots,g_{m}(u_0)$ up to significant $\tau\ge \poly(m\log{\epsilon^{-1}})$ bits in polynomial time (in $\tau$ and the bit size of the constraints) using \cite[Theorem 2]{AS20} to compute the and \cite[Proposition 1]{AS20} to lower bound the non-zero values.%
  \footnote{The reason for rounding to the rationals is that if we continue in the exact regime, the degree of algebraic numbers grows doubly exponentially in the number of recursions because every application of \cite[Theorem 8]{AS20} introduces a new algebraic number whose degree is only bounded by the product of the previous solutions.}

  For all $i=0,\dots,m-2$, assign $\ket{u_i} = g_i(x)$ if $g_i(x)\ne 0$.
  Then reduce the remaining constraints and again compute the amplitudes up to $\tau$ significant bits and then normalize.
  The additive error in the assigned qubits and the reduced constraints is then $\poly(\epsilon/m)$ for a sufficiently large $\tau$.
  We have to reduce the system so that it either becomes extending or remains almost extending.
  First note that the reduction produces no $1$-local constraints on a vertex $u_i$ with $i<m-1$, because then we would have $g_i(x)\ne 0$.
  Thus, the remaining reduced edges from $e_1,\dots,e_{m-2}$ are still extending.
  If both $g_{m-1}(x)\ne0$ and $g_{m}(x)\ne0$, then we can assign $\ket{u_{m-1}} = g_{m-1}(x)=g_m(x)$ and the remaining edge order is extending.
  If $g_{m-1}(x)=g_m(x)=0$, then the order remains almost extending.
  If $g_{m-1}(x)=0$ and $g_m(x)\ne0$ (or vice versa), then we obtain a new $1$-local constraint on $u_{m-1}$.
  But only one of $e_{m-1},e_{m}$ becomes $1$-local, and thus we can solve the residual system recursively.
  In total, we need at most $r$ recursions.
  The error increases additively with each recursion, so the total error is at most $\poly(\epsilon)$:
  Assuming we can compute a solution with error $\epsilon'$ on the residual system, we get total error $\epsilon' + \poly(\epsilon/m)$.
\end{proof}

\subsection{Solving generic instances on qubits of transfer type \texorpdfstring{$b=n-m+k-1$}{b=n-m+k-1}}\label{sscn:solvingtransferk}

\Cref{sscn:nonGeneric} showed how to improve on the paramaterized algorithm of \cite{aldiEfficientlySolvableCases2021} by keeping the transfer type fixed to $b=n-m+1$, but extending to non-generic instances. 
Here, we do the opposite --- we give a parameterized algorithm for the generic case, but now extend the set of transfer types we are able to handle to $b=n-m+k-1$, so that for any constant $k$, we obtain an efficient algorithm (under the assumption, as before, that radius $r\in O(\log n)$).

\begin{lemma}\label{lem:break-transfer}
  Let $H$ be a generic \PS\ instance on qubits with underlying hypergraph $G=(V,E)$, such that $G$ has an SDR and $\abs{V} = \abs{E}$.
  Then $G$ has $\bez$ product solutions, and none of these solutions \emph{breaks} any transfer function (i.e. no transfer function in $G$ maps a solution of $G$ to $0$).
\end{lemma}
\begin{proof}
  Consider some constraint $\ket{\phi}$ corresponding to the edge $e=\{v_1,\dots,v_k\}\in E$ on qubits $1,\dots,k$ and let $t\colon(\CC^2)^{k-1} \to \CC^2$ be the associated transfer function from qubits $v_1,\dots,v_{k-1}$ to $v_k$.
  We can write $\ket\phi = \ket{\phi_0}_{v_1,\dots,v_{k-1}}\ket{0}_{v_k} + \ket{\phi_1}_{v_1,\dots,v_{k-1}}\ket{1}_{v_k}$.
  Then $t(v_1,\dots,v_{k-1}) = 0$ iff $\braket{\phi_0}{v_1,\dots,v_{k-1}} = \braket{\phi_1}{v_1,\dots,v_{k-1}}=0$, where $v_i\in \CC^2$ also denotes the assignment to  qubit $v_i$.
  Denote by $H'$ the \PS\ instance obtained by replacing constraint $\ket{\phi}_e$ by $\ket{\phi_0}_{e'}$ and $\ket{\phi_1}_{e'}$, where $e' = \{v_1,\dots,v_{k-1}\}$, and let $G'=(V',E')$ be its underlying hypergraph.
  The product solutions of $H'$ are precisely the product solutions of $H$ that also break the transfer function $t$.
  Since $\ket{\phi}$ is the direct sum of $\ket{\phi_0}$ and $\ket{\phi_1}$ (up to permutation), the coefficients of $\ket{\phi}$ split into two disjoint subsets: the coefficients of $\ket{\phi_0}$ and those of $\ket{\phi_1}$.
  Hence, $H'$ is still generic.
  Since $\abs{V'} < \abs{E'}$, $H'$ does not have an SDR and generically no solutions by \cite{laumannProductGenericRandom2010}.
  Thus, $H'$ there exists a polynomial $g'$ in the coefficients of $H'$ such that $H'$ is unsolvable if $g'(\cdot)\ne 0$.
  There also exists a polynomial $g$ in the coefficients of $H$, such that $H$ has exactly $\bez$ solutions if $g(\cdot)\ne 0$.
  Since $H$ and $H'$ have the same coefficient set, we have that $H$ has no solution that breaks the transfer function $t$ if $gg'(\cdot)\ne 0$.
  By the same argument, generically none of the solutions of $H$ break \emph{any} transfer function.
\end{proof}

\begin{lemma}[\cite{AdBGS21}]
  Let $G = (V,E)$ be a $k$-uniform hypergraph of transfer type $b = n -m+k-1$ (equivalently, an $(k-1)$-almost extending edge order). Then $G$ has an SDR.
\end{lemma}

\begin{theorem}\label{thm:parameterized}
  Let $H$ be a generic \PS\ instance with constraints in $\QQ[i]$ on qubits with underlying $k$-uniform hypergraph $G = (V,E)$ of transfer type $b= n -m+k-1$ (equivalently, a $(k-1)$-almost extending edge order) with radius $r$.
  We can compute an $\epsilon$-approximate product state solution in time $\poly(L, \abs{\log\epsilon}, k^r, m^k)$, where $L$ is a bound on the bit size of the instance's rational coefficients, and $\epsilon$ the Euclidean distance to the closest product state solution.
\end{theorem}
\begin{proof}
  Kremer \cite{Kre24} gives a polynomial time algorithm to compute an edge order with minimum radius as well as the corresponding transfer filtration.
  The key insight is that the last vertex in an extending edge order must have degree $1$, which allows us to greedily partition the edges into layers, starting with all edges containing a vertex of degree $1$ as last layer.
  By trying all combinations for the $k-1$ non-extending constraints, we can compute the $(k-1)$-almost extending edge order of minimum radius in time $m^{O(k)}$.

  Observe that every transfer function depends on at least $k-1$ foundation variables.
  Via the transfer functions, we can write all qubits as a polynomial in the foundation qubits of degree at most $(k-1)^r$.
  Hence, every non-extending constraint is a polynomial in at least $k-1$ variables, of degree at most $k(k-1)^r$.
  The next step is to remove foundation qubits so that there exists a finite number of solutions generically, while maintaining the existence of an SDR.
  We argue that $G$ has an SDR matching only $k-1$ of the foundation vertices $V(G_0)$.

  Let $\wte_1,\dots,\wte_{k-1}$ be the non-extending edges, and choose distinct vertices $\wtv_1,\dots,\wtv_{k-1}$, such that $\wtv_i\in\wte_i$ for $i=1,\dots,k-1$.
  These exist by Hall's marriage theorem.
  For each extending edge $e_i\in E$, let $v(e_i)\in V(G_i)\setminus V(G_{i-1})$ be the added vertex.
  Construct a directed graph $\whG = (\whV,\whE)$ with $\whV = V\cup \{\whu,\whv\}$ and
  \begin{equation}
    \begin{aligned}
    \whE &= \{(\whu, v)\mid v\in V(G_0)\} \cup \{(\wtv_i, \whv)\mid i\in[k-1]\} \\
    &\cup \{(u, v(e))\mid e\in E\setminus\{\wte_1,\dots,\wte_{k-1}\}, u\in e\setminus\{v(e)\}\},
    \end{aligned}
  \end{equation}
  i.e., edges from $\whu$ to the foundation, edges from all nodes in a hyperedge $e$ to the added vertex $v(e)$, and edges from the $\wtv_1,\dots,\wtv_{k-1}$ to $\whv$.
  Note that each vertex in $V\subseteq \whV$ has at least $k-1$ incoming edges from a ``lower layer''.
  Hence, one has to remove at least $k-1$ vertices from $\whG$ to disconnect $\whu,\whv$.
  By Menger's theorem, there exist $k-1$ internally disjoint paths from $\whu$ to $\whv$.
  By construction, each of these paths goes via a foundation vertex $u_i\in V(G_0)$ to $\wtv_i$.
  We can construct an SDR by matching $\wte_i$ to $\wtv_i$, then matching $e$ s.t. $v(e) = \wtv_i$ to the predecessor of $\wtv_i$ in the path.
  Iterate until reaching the foundation.
  We can assign all remaining edges $e$ to $v(e)$, since their $v(e)$ are outside the $k-1$ paths.
  Set the unmatched foundation qubits to $\ket0$ and let the resulting system be $H'$ on graph $G'=(V',E')$.
  The SDR constructed above is also valid for $G'$.
  $H'$ still has generic constraints, since setting variables to $\ket{0}$ just means we discard coefficients, but not change them.

  Let $F$ be the multi-homogeneous system obtained by writing every qubit of $G'$ as polynomials in the entries of the foundation qubits via the transfer functions.
  The solutions of $F$ also contains the foundation qubits of all solutions of $H'$, which can be extended to the qubits outside the core via the transfer functions.
  However, the solution set of $F$ can also contain assignments to the foundation that break transfer functions.
  By \Cref{lem:break-transfer}, none of the transfer functions are broken if the foundation is set to an actual solution to $H'$.
  An additional polynomial inequality $g$ of degree at most $n(k-1)^r$ ensures that we only find solutions that break no transfer functions.
  We can use the existential theory of the reals to find a solution that satisfies both $F$ and $g$.
  For rational entries, Renegar's algorithm \cite[Theorem 1.2]{Ren92d} can compute an $\epsilon$-approximate solution in time $\poly(L, k^r,\abs{\log\epsilon})$, where $L$ is a bound on the bit size of the constraints.
  We introduce separate variables for the real and imaginary parts, which allows us to also use complex conjugates in our constraints.
\end{proof}

\subsection{Solving higher dimensional systems via weighted transfer filtrations}\label{sscn:wtransfer}

Finally, we show how to extend the technique of transfer filtrations (\Cref{def:transfer-type}) from qu\emph{b}its to qu\emph{d}its, and give an explicit family of high-dimensional QSAT with WSDR instances which we can solve exponentially faster than brute force (\Cref{ssscn:efficient}).

The basic idea is to still consider a hypergraph with a filtration $G_0\subseteq G_1\subseteq\cdots\subseteq G_m=G$ but now allowing for the addition of more edges, and potentially more vertices, at each step in the filtration.
The most straightforward generalization is to maintain the requirement $|V(G_i)|\le |V(G_{i-1})|+1$ for each $i\in \{1,\ldots,m\}$ but, in the case in which $|V(G_i)|= |V(G_{i-1})|+1$ to allow for as many edges to be added as the weight of the new vertex (while maintaining the provision that each new edge must contain the new vertex). These type of {\it weighted transfer filtrations} can be used, to explicitly (and in some cases, depending on the growth of the radius, efficiently) construct solutions to the corresponding instances of \PS\ along the lines of \Cref{sscn:nonGeneric}. 

More generally we can relax the condition $|V(G_i)|\le |V(G_{i-1})|+1$ to the requirement that the induced subhypergraph of $G_i$ induced by $V(G_i)\setminus V(G_{i-1})$ has itself a transfer filtration of type $b=n-m+1$. 
As formalizing the high-dimensional case in full generality becomes technically cluttered, for pedagogical purposes we instead demonstrate the idea with concrete examples.

\paragraph{Qubits on a 1D periodic lattice.}
In order to set up the notation for more general examples, we begin by considering a system of $n$ qubits located at the vertices of a 1D periodic lattice, i.e.\ a cycle of length $n$. This system is efficiently solvable via transfer functions~\cite{bravyiEfficientAlgorithmQuantum2006,beaudrapLinearTimeAlgorithm2016} using transfer functions, along the following lines. 
We parametrize the $i$-th qubit state as $x_i^0|0\rangle+x_i^1|1\rangle$, for $i\in \mathbb Z/n\mathbb Z$. Each edge corresponds to a 2-local QSAT constraint $\varphi_i$ of the form
\begin{equation}
\sum_{p,q=0}^1 \varphi_i^{pq} x_i^px_{i+1}^q=0\,.
\end{equation}
Passing to affine coordinates $z_i=\frac{x_i^0}{x_i^1}$, this translates to
\begin{equation}
z_{i+1}=-\frac{\varphi_i^{01}z_i+\varphi_i^{11}}{\varphi_i^{00}z_i+\varphi_i{10}}
\end{equation}
which, after $n$ iterations, leads to an expression of the $z_i$ as solution of a quadratic equation $a_iz_i^2+b_iz_i+c_i=0$ whose coefficients coefficients $a_i,b_i,c_i$ are multilinear polynomials of total degree $n$ in the variables $\varphi_i^{pq}$.

\paragraph{Qutrits on a 1D periodic lattice.}

Our next stepping stone is to keep the same interaction hypergraph (the 1D periodic lattice), but to allow qubits to be replaced by $n$ qu\emph{t}rits. We parametrize the $i$-th qutrit as $x_i^0|0\rangle + x_i^1|1\rangle +x_i^2 |2\rangle$, $i\in \mathbb Z/n\mathbb Z$. Each edge corresponds to a 2-local QSAT constraint $\varphi_i$ of the form
\begin{equation}
\sum_{p,q=0}^2 \varphi_i^{pq} x_i^px_{i+1}^q =0\,.
\end{equation}
We can further impose 1-local constraints on each qutrit, which, in terms of affine coordinates $z_i^0=\frac{x_i^0}{x_i^2}$, $z_i^1=\frac{z_i^1}{z_i^2}$ can be written as $z_i^0=\alpha_i^1 z_i^1+\alpha_i^2$. Substituting into the constraint we obtain
\begin{equation}
z_{i+1}^1=-\frac{A_iz_i^1+B_i}{C_iz_i^1+D_1}
\end{equation}
where
\begin{align*}
A_i^1 &= \varphi_i^{00}\alpha_i^1\alpha_{i+1}^2+\varphi_i^{10}\alpha_{i+1}^2+\varphi_i^{02}\alpha_i^1+\varphi_i^{12}\\
B_i^1 &= \varphi_i^{00}\alpha_i^2\alpha_{i+1}^2+\varphi_i^{02}\alpha_i^2+\varphi^{20}_i\alpha_{i+1}^2+\varphi_i^{22} \\
C_i^1 &=\varphi_i^{00}\alpha_i^1\alpha_{i+1}^1+\varphi_i^{10}\alpha_{i+1}^1+\varphi_i^{01}\alpha_i^1+\varphi_i^{11}\\
D_i^1 &= \varphi_i^{00}\alpha_i^2\alpha_{i+1}^1+\varphi_i^{01}\alpha_i^2+\varphi_i^{20}\alpha_{i+1}^1+\varphi_i^{21}
\end{align*}
After $n$ iterations, we obtain the $z_i^1$ as solutions of quadratic equations whose coefficients are polynomials of total degree 3n, linear in each of the $\varphi_i^{pq}$ and quadratic in each of the $\alpha_i^{pq}$s.

\paragraph{Qutrits on a 2D periodic lattice.}

Now we are ready to describe our first example of genuinely more general transfer filtrations in presence of qutrits. Specifically, consider now a system of $mn$ qutrits located at the vertices of a square lattice with periodic boundary conditions. We parametrize the qutrit on the $(i,j)$ node of the lattice as
\begin{equation}
x_{i,j}^0|0\rangle+x_{i,j}^1|1\rangle+x_{i,j}^2|2\rangle
\end{equation}
for all $i\in \mathbb Z/m\mathbb Z$ and $j\in \mathbb Z/n\mathbb Z$. We have ``horizontal'' 2-local constraint
\begin{equation}
\sum_{p,q=0}^2 \varphi_{i,j}^{p,q}x_{i,j}^p x_{i+1,j}^q=0
\end{equation}
as well as ``vertical'' ones
\begin{equation}
\sum_{p,q=0}^2 \psi_{i,j}^{p,q}x_{i,j}^p x_{i,j+1}^q=0
\end{equation}
for each $i\in \mathbb Z/m\mathbb Z$ and $j\in \mathbb Z/n\mathbb Z$. We work in affine coordinates $z_{i,j}^0=\frac{x_{i,j}^0}{x_{i,j}^2}$ and $z_{i,j}^1=\frac{x_{i,j}^1}{x_{i,j}^2}$ and impose arbitrary 1-local constraints on the qutrits of one of the ``rows'' of the lattice, say, $z_{i,0}^0=\alpha_{i,0}^1z_{i,0}^1+\alpha_{i,0}^2$
for all $i\in \mathbb Z/m\mathbb Z$. Then we solve the $0$-th row using the method outlined above expressing each $z_{i,0}^1$ as a solution of a quadratic equation with coefficients of total degree $2m$ in the alphas. Then imposing the $\psi_{i,0}$ constraints, we obtain constraints of the form $z_{i,1}^0 = \alpha_{i,1}^1z_{i,1}^1+\alpha_{i,1}^2$ where the $\alpha_{i,1}^p$ are fractions with both numerator and denominator are linear in the $z_{i,0}^1$. Iterating this process $n$-times we can solve the rows one by one in terms of the $\alpha_{i,0}^p$ until, thanks to the periodic boundary conditions, return to $z_{i,0}^p$. This results to a system of equations in the $\alpha_{i,0}^p$ whose degree is (simply) exponential in $n$.

\subsection{Weighted graphs with constant weights}\label{sscn:weightedGraphsConst}

The example of qutrits on a 2D periodic lattice can be generalized to qudits of local dimension $d$ on a periodic $(d-1)$-dimensional lattice, i.e.\ on the weighted graph $(C_{m_1}\Box C_{m_2} \Box\cdots\Box C_{m_N},d-1)$, for $\Box$ the graph Cartesian product (\Cref{def:cartesian}). This can be done iteratively. For instance, when $d=4$, and the corresponding graph is $C_{m_1}\Box C_{m_2}\Box C_{m_3}$, we can isolate a 2-dimensional slice, say, $C_{m_1}\Box C_{m_2}\Box \{1\}$, impose 1-local constraints on each of its vertices, solve using the method above, and then use the constraints corresponding to edges ``orthogonal'' to the 2D slice to reduce by one unit the local dimension of the qudits of the slice $C_{m_1}\Box C_{m_2}\Box \{2\}$ and repeat.

More generally, one can replace the cyclic graphs $C_{m}$ with pseudoforests (i.e.\ a disjoint union of graphs having at most one cycle). This is because~\cite{bravyiEfficientAlgorithmQuantum2006,aradLinearTimeAlgorithm2016,beaudrapLinearTimeAlgorithm2016}, instances of 2-QSAT on qubits whose interaction graph is a pseudoforest are solvable in linear time. Moreover we know that, since pseudoforests have SDRs and the property of admitting a WSDR is preserved under cartesian products, the cartesian product of N pseudoforests admits a WSDR with constant weight $w=N$.

Consider a graph $G$ together with a finite filtration by subgraphs $G_0\subseteq G_1\subseteq\cdots\subseteq G_n=G$ constructed as follows. 
First, we let $G_0$ (the foundation) be a graph with no edges. Then let $P_0$ be an arbitrary pseudoforest. Then $G_1$ is constructed by adding edges to $G_0+P_0$ connecting vertices of $G_0$ to vertices of $P_0$ with the provision that the degree of the vertices of $P_0$ increases at most by one. Similarly, $G_2$ is constructed by taking the disjoint union of $G_1$ with a pseudoforest $P_1$ and adding edges to $G_1+P_1$ connecting vertices of $G_1$ to vertices of $P_0$ in a way that the degree of the vertices of $P_1$ increases by at most one unit. And so forth.

\subsubsection{An explicit example with exponential speedup: The Pinwheel graph}\label{ssscn:efficient}

The goal of our next example is to illustrate how a modification of the 2D lattice construction can give rise to an infinite family of instances of 2-QSAT on qutrits that are efficiently solvable.

For each positive integer $n$, consider the graph $\Gamma_n$, which we refer to as a \emph{Pinwheel graph} (\Cref{fig:pinwheel}). The vertices are $v_0$, located at the origin and $v_{j,k}$ located at the point in the plane with polar coordinates $(j,2^{1-j}\pi k)$ for all $j=1,\ldots,n$ and $k\in \mathbb Z/2^j\mathbb Z$.

\begin{figure}[t]
\begin{center}
\scalebox{0.8}
{
\begin{tikzpicture}
\node[circle,fill=red,inner sep=1.5pt] at (1,0) {};
\node[circle,fill=red,inner sep=1.5pt] at (-1,0) {};
\node[circle,fill=red,inner sep=1.5pt] at (0,0) {};
\foreach \x in {0,...,31}
{\node[circle,fill=red,inner sep=1.5pt] at (11.25*\x:5) {};}

\foreach \x in {0,...,15}
{
\node[circle,fill=red,inner sep=1.5pt] at (22.5*\x:4) {};
\draw[blue] (22.5*\x:4) -- (22.5*\x:5);
\draw[blue] (22.5*\x:4) -- (11.25+22.5*\x:5);
}
\foreach \x in {0,...,7}
{
\node[circle,fill=red,inner sep=1.5pt] at (45*\x:3) {};
\draw[blue] (45*\x:3) -- (45*\x:4);
\draw[blue] (45*\x:3) -- (22.5+45*\x:4);
}
\foreach \x in {0,...,3}
{
\node[circle,fill=red,inner sep=1.5pt] at (90*\x:2) {};
\draw[blue] (90*\x:2) -- (90*\x:3);
\draw[blue] (90*\x:2) -- (45+90*\x:3);
}
\foreach \x in {0,...,1}
{
\node[circle,fill=red,inner sep=1pt] at (180*\x:1) {};
\draw[blue] (180*\x:1) -- (180*\x:2);
\draw[blue,bend right=30] (180*\x:1) to (90+180*\x:2);
}

\draw[green] (180-11.25:5) -- (360-11.25:5);

\draw[black] (0,0) circle (1);
\draw[black] (0,0) circle (2);
\draw[black] (0,0) circle (3);
\draw[black] (0,0) circle (4);
\draw[black] (0,0) circle (5);
\draw[blue] (1,0) -- (-1,0);

\end{tikzpicture}
}
\end{center}
\caption{Pinwheel graph $\Gamma_n$ for the case of $n=5$.}
\label{fig:pinwheel}
\end{figure}

There are three kinds of edges:
\begin{enumerate}
\item $e_{j,k}$ connecting $v_{j,k}$ to $v_{j,k+1}$ for each $k\in \mathbb Z/2^j\mathbb Z$ (colored in black in the picture);
\item $\epsilon_{j,k}$ connecting $v_{j,k}$ to $v_{j-1,k/2}$ if $k$ is even and to $v_{j-1,{k-1}/2}$ if $k$ is odd (colored in blue in the picture);
\item $\varepsilon_i$ connecting $v_{n,2^{n-i}-1}$ to $v_0$ for $i\in \mathbb Z/2\mathbb Z$ (colored in green in the picture).
\end{enumerate}

$\Gamma_n$ has a total of $1+2+4+\cdots 2^n=2^{n+1}-1$ vertices and $(2^{n+1}-2)+(2^{n+1}-2)+2=2(2^{n+1}-1)$ edges. Hence placing a qutrit at each vertex and a 2-local constraint at each edge we obtain a system with as many degrees of freedom as constraints and thus finitely many solutions.

Moreover, $\Gamma_n$ has a natural WSDR with constant weight $w=2$ defined by $f(\varepsilon_0)=v_0=f(\varepsilon_1)$ and $f(e_{j,k})=v_{j,k}=f(\epsilon_{j,k})$ for all $j=1,\ldots,n$ and $k=1,\ldots,2^{j}$.

Starting with an arbitrary assignment of the qutrit located at $v_0$ and imposing the constraints corresponding to the edges $\epsilon_{1,\bullet}$ we reduce the qutrits located at $v_{1,\bullet}$ to qubits subject to the 2-local constraints corresponding to the edges $e_{1,\bullet}$ This is a 1D periodic lattice of qubits that can be solved in linear time. Imposing the constraints corresponding to the edges $\epsilon_{2,\bullet}$ we reduce the qutrits located at $v_{2,\bullet}$ to qubits and iterate the previous until we have a product assignments for all qutrits in terms of the initial assignment at $v_0$ that satisfies all $e$ and $\epsilon$ constraints. At this point we impose the $\varepsilon_\bullet$ constraints and realize admissible assignments at $v_0$ as the solution of a system of two polynomial equations in two variables. This can be solved using, say, the resultant (see, e.g.~\cite{coxIdealsVarietiesAlgorithms2015}). Note that both the degree of these polynomials and the number of degrees of freedom grows (simply) exponentially with $n$.

\subsection*{Acknowledgements}
We thank Niel de Beaudrap, Neal Bushaw, Bruno Grenet, David Gosset, Christian Ikenmeyer, Pascal Koiran, Gr\'{e}goire Lecerf and Thomas Vidick for helpful discussions.
We thank Simon-Luca Kremer for pointing out a mistake in an earlier version of the proof of \Cref{thm:parameterized}.
SG was supported by
the DFG under grant numbers 432788384 and 450041824, the BMBF within the funding program “Quantum Technologies - from Basic Research to Market” via project PhoQuant (grant
number 13N16103), and the project “PhoQC” from the programme “Profilbildung 2020”, an initiative of the Ministry of Culture and Science of the State of Northrhine Westphalia.
DR was supported by the DFG under grant number 432788384. MA was supported in part by VCU Quest Award ``Quantum Fields and Knots: An integrative Approach''.
Some of the results in this paper were obtained while MA was visiting Paderborn University. MA is grateful for the hospitality and the excellent working conditions.

\bibliography{3QSAT-lipics-bib,Sevag_Gharibian_Central_Bibliography_Abbrv,Sev,more}

\appendix
\section{Proof of Hall's Marriage Theorem for weighted hypergraphs}\label{app:proof}

The proof below is a simple adaptation to the weighted case of the proof found in~\cite{juknaExtremalCombinatoricsApplications2011}.

\thmHMT*
\begin{proof}
Assume $(G,w)$ has a WSDR $f:E(G)\to V(G)$. Since $f(e)\in e$ for every $e\in E(G)$, then $f(X)\subseteq V_X$ and thus $\sum_{v\in V_X}|f^{-1}(v)|=|X|$ for each $X\subseteq E(G)$. Hence
\begin{equation}
|V_X|_w=\sum_{v\in V_X} w(v) \ge \sum_{v\in V_X} |f^{-1}(v)| = |X|\,.
\end{equation}
Conversely, assume $|V_X|_w\ge |X|$ for every $X\subseteq E(G)$. If $G$ has a single edge $e$, by assumption that edge contains a vertex $v$ such that $w(v)\ge 1$ and the assignment $e\mapsto v$ is the required WSDR. We now work by induction on the number of edges, and assume the statement is proved for all hypergraph with less than $m$ edges. Let $E(G)=m$. We distinguish two cases.

{\it Case 1}. Suppose that $|V_X|>|X|$ whenever $|X|<m$. Pick $e\in E(G)$ and $v\in e$ such that $w(v)\ge 1$. Let $(G',w')$ be the weighted hypergraph such that $V(G')=V(G)$, $E(G')=E(G)\setminus \{e\}$, $w'(z)=w(z)$ if $z\in V(G)\setminus\{v\}$ and $w'(v)=w(v)-1$. Then for every $X'\subseteq E(G')$
\begin{equation}
|V_{X'}|_{w'}=\sum_{v\in V_{X'}} w'(v)\ge -1 + \sum_{v\in V_{X'}} w(v) > -1+|X'|\,.
\end{equation}
Since necessarily $|X'|<m$, by induction we have that $(G',w')$ has a WSDR $g$. Let $f:E(G)\to V(G)$ such that $f(e')=g(e')$ for every $e'\in E(G')$ and $f(e)=v$. Then $f$ is a WSDR for $(G,w)$.

{\it Case 2}. Suppose there exists $X\subseteq E(G)$ such that $|V_X|=|X|<m$. By induction, the weighted hypergraph $(G_1,w_1)$ such that $V(G_1)=V(G)$, $E(G_1)=X$ and $w_1=w$ has a WSDR $f_1$. Consider the weighted hypergraph $(G_2,w_2)$ such that $V(G_2)=V(G)$, $E(G_2)=E(G)\setminus X$, and $w_2(v)=w(v)-|f_1^{-1}(v)|$ for every $v\in V(G)$. Suppose $(G_2,w_2)$ has no WSDR. By induction, there would exist $Y\subseteq E(G_2)$ such that $|V_Y|<|Y|$. Since $w(v)=w_2(v)$ for all $v\in V_Y\setminus V_X$, this would imply
\begin{equation}
|V_{X\cup Y}|_w = |V_X\cup V_Y|_w = \sum_{v\in V_X} w(x) + \sum_{v\in V_Y\setminus V_X} w(v) < |X|+|Y|
\end{equation}
which contradicts the assumption. Hence $(G_2,w_2)$ has a WSDR $f_2:E(G_2)\to V(G)$. Let $f:E(G)\to V(G)$ be such that $f(e)=f_1(e)$ if $e\in X$ and $f(e)=f_2(e)$ otherwise. Then
\begin{equation}
|f^{-1}(v)|=|f_1^{-1}(v)|+|f_2^{-1}(v)|\le |f_1^{-1}(v)|+w_2(v)=w(v)
\end{equation}
for all $v\in V(G)$ and thus $f$ is a WSDR for $(G,w)$.
\end{proof}

\end{document}